\documentclass[sn-mathphys,Numbered]{sn-jnl}

\usepackage{graphicx}%
\usepackage[caption=false]{subfig}%
\usepackage{multirow}%
\usepackage{amsmath,amssymb,amsfonts}%
\usepackage{amsthm}%
\usepackage{mathrsfs}%
\usepackage[title]{appendix}%
\usepackage{xcolor}%
\usepackage{textcomp}%
\usepackage{manyfoot}%
\usepackage{booktabs}%
\usepackage{algorithm}%
\usepackage{algorithmicx}%
\usepackage{algpseudocode}%
\usepackage{listings}%
\usepackage{tikz}%
\usepackage{braket}%
\usetikzlibrary{arrows.meta,bending,decorations.markings}%
\usepackage{dsfont}%
\usepackage{accents}%
\usepackage{mathtools}%

\tikzset{
    Vertex/.style = {fill,circle,inner sep=1.5pt},
    insert vertex/.style = {decoration={
         markings,
         mark=at position #1 with {\node[Vertex]{};},
        },
     postaction=decorate
     }     
}
\newcommand\I{\mathrm{i}}
\newcommand\E{\mathrm{e}}
\newcommand\D{\mathrm{d}}

\newcommand\R{\mathbb{R}}
\newcommand\C{\mathbb{C}}
\newcommand\N{\mathbb{N}}
\newcommand\Z{\mathbb{Z}}
\newcommand\EE{\mathbb{E}}
\newcommand{\llangle}{\langle\!\langle}
\newcommand{\rrangle}{\rangle\!\rangle}
\newcommand\numberthis{\addtocounter{equation}{1}\tag{\theequation}}

\newtheorem{theorem}{Theorem}
\newtheorem*{theorem*}{Theorem}
\newtheorem{proposition}[theorem]{Proposition}%
\newtheorem{corollary}[theorem]{Corollary}
\newtheorem{lemma}[theorem]{Lemma}

\theoremstyle{remark}%
\newtheorem{remark}{Remark}%

\theoremstyle{definition}%
\newtheorem{definition}{Definition}%

\theoremstyle{plain}

\begin{document}

\title[]{Relative Entropy and Mutual Information in Gaussian Statistical Field Theory}


\author*[1,2]{\fnm{Markus} \sur{Schr\"{o}fl}}\email{markus.schroefl@uni-jena.de}

\author[2]{\fnm{Stefan} \sur{Floerchinger}}\email{stefan.floerchinger@uni-jena.de}

\affil[1]{\orgdiv{Institut f\"{u}r Theoretische Physik}, \orgname{Universit\"{a}t Heidelberg}, \orgaddress{\street{Philosophenweg 16}, \city{Heidelberg}, \postcode{69120}, \country{Germany}}}

\affil*[2]{\orgdiv{Theoretisch-Physikalisches Institut}, \orgname{Friedrich-Schiller-Universit\"{a}t Jena}, \orgaddress{\street{Max-Wien-Platz 1}, \city{Jena}, \postcode{07743}, \country{Germany}}}


\abstract{Relative entropy is a powerful measure of the dissimilarity between two statistical field theories in the continuum. In this work, we study the relative entropy between Gaussian scalar field theories in a finite volume with different masses and boundary conditions. We show that the relative entropy depends crucially on $d$, the dimension of Euclidean space. Furthermore, we demonstrate that the mutual information between two disjoint regions in $\R^d$ is finite if the two regions are separated by a finite distance and satisfies an area law. We then construct an example of ``touching'' regions between which the mutual information is infinite. We argue that the properties of mutual information in scalar field theories can be explained by the Markov property of these theories.}




\maketitle

\newpage

\tableofcontents

\newpage

\section{Introduction}\label{sec1}

Interest in information-theoretical aspects of physical theories has grown rapidly in recent decades. Originally introduced to physics via the concept of entropy by Clausius and Boltzmann, information theory now plays a fundamental role in a wide variety of subfields of physics. Examples include classical statistical physics in and out of equilibrium \cite{Jaynes1957a,Jaynes1957b,Jaynes1989,Floerchinger2020a,Dowling2020}, quantum mechanics, especially in connection with quantum computing, \cite{Nielsen2010,Wilde2013,Floerchinger2020b,Floerchinger2021b,Floerchinger2021c,Garttner2022}, the black hole information paradox \cite{Polchinski2017} and quantum field theory, whose entanglement properties remain an area of active research \cite{Bombelli1986,Srednicki1993,Callan1994,Calabrese2004,Calabrese2005a,Calabrese2005b,Calabrese2009a,Calabrese2009b,Calabrese2010,Audenaert2002,Plenio2005,Cramer2006,Eisert2010,Casini2009,Casini2022,Berges2017a,Berges2017b,Berges2018}.

For the specific case of continuum quantum field theories, the computation (and even the definition) of entropies is typically complicated due to the difficulties arising from the infinite number of degrees of freedom. A suitable notion of entropy for continuum field theories is the relative entropy (or Kullback-Leibler divergence) \cite{Kullback1951}, which quantifies the (information-theoretic) distinguishability of two states. Colloquially, since the relative entropy compares two states of the theory, the ultraviolet divergences present in both states cancel each other out, giving finite results even in the continuum, see, e.g., \cite{Floerchinger2021,Ditsch2023}. Relative entropies in relativistic quantum field theories can be rigorously defined using Tomita-Takesaki modular theory \cite{Araki1975,Araki1977,Haag2012,Witten2018,Hollands2018}.

In physics, relative entropy is an important quantity in quantum information theory \cite{Vedral2002,Nielsen2010} and has been used to study the information-theoretic properties of classical and quantum many-body systems \cite{Donald1987,Gaite1996,Zegarliski2000,Xu2019,Lashkari2014,Erdmenger2022}. It has also been used to formulate thermodynamics using a principle of minimum expected relative entropy \cite{Floerchinger2020a}, to study open system dynamics \cite{Mueller2015,Dowling2020} and to formulate uncertainty relations \cite{Rajagopal1986,Floerchinger2021}.

In this work, we are concerned with the properties of relative entropy in the context of \emph{classical} Gaussian statistical field theory\footnote{In the mathematical literature, this model is known as the Gaussian free field, see \cite{Werner2021,Sheffield2007}.}. Physically, this model can be interpreted either as a Gaussian approximation of a classical statistical system close to a second order phase transition or as the Euclidean version of the relativistic quantum field theory of a single massive scalar boson. In the functional integral formalism, this model is defined by a functional Gaussian probability measure and is thus a problem of classical probability theory. It is also the starting point for the construction of interacting theories \cite{Glimm2012,Simon1974,Guerra1975a,Guerra1975b}. Other recent works dealing with the information theory of Euclidean field theories include \cite{Floerchinger2023a,Floerchinger2023b,Cao2023}.

A statistical field theory shares an important property with quantum field theories, namely a continuum of degrees of freedom. Therefore, we need to use relative entropies to study the information-theoretic properties of this model. Besides physical considerations, Gaussian field theories are also interesting from a purely probability-theoretic point of view. In contrast to ``ordinary'' Gaussian measures on $\R^n$, which are specified by multivariate normal distributions, Gaussian measures on infinite dimensional spaces (which are needed in field theory because the number of degrees of freedom is infinite) are much more subtle to handle. In particular, the relative entropy between two Gaussian measures on infinite dimensional spaces can be infinite even if both measures are non-degenerate. One of the main challenges of this work will be to determine the conditions under which the relative entropy between two Gaussian field theories is finite.

Since we are working with relative entropies, we need to choose \emph{two} probability measures that we want to compare with each other. As we are considering only free theories (i.e., Gaussian measures), our options for choosing these measures are limited. The Gaussian statistical field theories we consider in this work will be assumed to have vanishing mean, and are therefore completely determined by a differential operator and a mass term in the (Euclidean) action, see the discussion in Chapter \ref{sec:func_integrals}. Furthermore, in order to render the relative entropies considered in this paper finite, we need to work on a finite volume, which requires the imposition of boundary conditions in the definition of the differential operator. The need for a finite volume is not surprising, since an entropic quantity like the relative entropy is expected to be extensive\footnote{However, see the discussion of $n$-particle states in \cite{Floerchinger2021} for an example of a relative entropy in field theory that is finite even in infinite volume.}.

Different states (or theories) with physical meaning, to be compared using relative entropy, can therefore be obtained by choosing different boundary conditions, different masses, or a combination of both\footnote{In principle, one could also allow for a space-dependent mass term and introduce external gauge fields or a Riemannian metric, which could be functions of the coordinates.}. In the main part of this paper we derive conditions under which the relative entropy between two Gaussian measures, corresponding to different choices of the above possibilities, is finite. We will see that the finiteness of the relative entropy between two field theories with different masses depends critically on the dimension $d$ of Euclidean spacetime. In particular, we find that $d=4$ is a critical dimension for the relative entropy between theories with different masses and the same (classical) boundary conditions\footnote{In the remainder of this work, we call Dirichlet, Neumann, Robin and periodic boundary conditions ``classical'' boundary conditions, see Section \ref{sec:classical_bcs}.}, in the sense that it is infinite at and above this dimension. Furthermore, we demonstrate that the relative entropy between two field theories over a bounded region with different boundary conditions can, in general, be infinite in all dimensions. For the special case of the relative entropy between two field theories with Robin boundary conditions, we find that it is finite precisely when $d < 3$. These results are summarized in the following two Theorems.
\begin{theorem*}[{Theorem \ref{thm:main_1}}]
    Let $\Omega \subset \R^d$ be open, bounded and with piecewise smooth boundary. The relative entropy $D_\mathrm{KL} (\mu_1 \| \mu_2)$ between two field theories on $\Omega$ with the same classical boundary conditions but different, non-zero masses, represented by Gaussian measures $\mu_1$ and $\mu_2$, respectively, is finite if and only if $d < 4$.
\end{theorem*}
\begin{theorem*}[{Theorems \ref{thm:main_2} and \ref{thm:main_3}}]
    Let $\Omega \subset \R^d$ be open, bounded and with $C^1$-boundary. The relative entropy $D_\mathrm{KL} (\mu_1 \| \mu_2)$ between two field theories on $\Omega$ with the same non-zero mass but different Robin boundary conditions (including the case of Neumann boundary conditions) is finite if and only if $d < 3$. Furthermore, the relative entropy between a field theory with Dirichlet boundary conditions and a field theory with Robin boundary conditions (including the case of Neumann boundary conditions) is infinite for all $d \in \N$.
\end{theorem*}

In addition to these main results, we show that the relative entropy between fields with different masses exhibits an ordering with respect to the choice of boundary conditions. In particular, we find that the relative entropy between fields with Dirichlet boundary conditions is smaller than the relative entropy between fields with periodic boundary conditions, which is in turn smaller than the relative entropy between fields with Neumann boundary conditions. This ordering is independent of the dimension of Euclidean spacetime. Moreover, we demonstrate that, while the relative entropy between fields with different masses is infinite for $\Omega = \R^d$, the infinite volume limit of the relative entropy \emph{density} exists for $d < 4$, and this limit is independent of the choice of Dirichlet, Robin or periodic boundary conditions. For the case $d=1$, we derive exact closed-form expressions for the relative entropy, both for fields with different masses and for fields with different boundary conditions.

Another quantity we will consider in this work is the mutual information between two disjoint regions of Euclidean spacetime. Interpreted as the average amount of information shared between the fields in each of the different regions, the mutual information provides an insight into the mutual dependence of these fields. Since mutual information takes into account non-linear relationships between random variables, it is generically more general than other correlation quantifiers that only consider linear relationships such as covariance \cite{Li1990,Dionisio2004}. Mutual information has been used in classical and quantum statistical physics, for example in the study of phase transitions \cite{Matsuda1996,Wicks2007,Wilms2011,Wilms2012,Lau2013,Iaconis2013,Stphan2014,Sriluckshmy2018}.

We will show that the properties of the mutual information between two disjoint regions are dictated by the Markov property of the scalar field \cite{Nelson1973a,Nelson1973b}. In particular, we show that the mutual information is finite if the regions are separated by a finite distance and construct an explicit example of ``touching'' rectangular regions for which the mutual information is infinite. This finding is summarized in the following Theorem.
\begin{theorem*}[{Theorems \ref{thm:equiv_mutual_info} and \ref{thm:main_5}}]
    Let $\Omega_A$ and $\Omega_B$ be disjoint bounded open subsets of $\R^d$. If $\Omega_A$ and $\Omega_B$ are separated by a finite distance, the mutual information $I (\Omega_A : \Omega_B)$ between the two fields on $\Omega_A$ and $\Omega_B$, respectively, is finite. Furthermore, if $\Omega_A$ and $\Omega_B$ are disjoint open $d$-rectangles which ``touch'' each other on one of their faces, see Fig. \ref{fig:rectangles}, then the mutual information $I (\Omega_A : \Omega_B)$ is infinite.
\end{theorem*}
Moreover, we argue (without providing a rigorous proof) that such a mutual information satisfies an \emph{area law} \cite{Wolf2008,Lau2013}, i.e., only the degrees of freedom at the boundaries of the regions contribute to the mutual information, reinforcing the notion that the Markov property plays a crucial role in the study of mutual information. Just as for the relative entropy, we derive exact closed-form expressions for the mutual information between two disjoint regions in $d=1$.

\emph{The remainder of this paper is organized as follows.} In Section \ref{sec:func_integrals}, we briefly recall the basic principles of statistical and Euclidean quantum field theories in order to provide a physical context for this work. We then discuss functional integrals and in particular the functional probability measures that define the theories we wish to consider. These probability measures are specified by a differential operator. Since we are working with field theories over a finite volume, we need to impose boundary conditions in the definition of this differential operator, and we devote Section \ref{sec:covariance_operators} to their study. A brief reminder of information theory, with special emphasis on relative entropy and mutual information, is given in Section \ref{sec:info_theory}. The main part of this work consists of Section \ref{sec:applications}. In Section \ref{sec:diff_masses}, we study the relative entropy between free scalar field theories with different masses. In particular, we derive conditions under which the relative entropy between two such theories is finite, and we discuss the dependence of the relative entropy on the choice of boundary conditions. In Section \ref{sec:diff_bcs}, we consider the relative entropy between field theories with different boundary conditions. Finally, we study the mutual information between two disjoint regions of Euclidean space in Section \ref{sec:mutual_info}. The Appendices \ref{app:covariance_operators} to \ref{app:1D_Dirichlet_RE} provide a collection of technical results used in the main part of the paper.

\emph{Notation and Definitions.} Elements $\boldsymbol{x}$ of Euclidean space $\R^d$ are written in boldface except when we explicitly work in $d = 1$. The expectation value of a random variable $X$ is denoted by $\EE [X]$ and the probability of an event $x$ is denoted by $\mathbb{P} (x)$.

Let $\mathcal{H}$ be a complex Hilbert space with inner product $\braket{.,.}_\mathcal{H}$. Throughout this paper, we use the physicist's definition of the inner product, i.e., $\braket{.,.}_\mathcal{H}$ is linear in the second argument and conjugate linear in the first. We denote by $\mathfrak{L} (\mathcal{H})$, $\mathfrak{B} (\mathcal{H})$, $\mathfrak{T} (\mathcal{H})$ and $\mathrm{HS} (\mathcal{H})$ the set of linear, bounded, trace class and Hilbert-Schmidt operators on $\mathcal{H}$, respectively. For any of the aforementioned sets, a superscript $+$ denotes the subset of positive operators. For example, $\mathfrak{B}^+ (\mathcal{H})$ denotes the set of all positive bounded operators on $\mathcal{H}$. An operator $T \in \mathfrak{B} (\mathcal{H})$ is called positive if $\braket{f,Tf}_\mathcal{H} \geq 0$ for all $f \in \mathcal{H}$. If $T$ is a positive operator and $\braket{f,Tf}_\mathcal{H} = 0$ if and only if $f = 0$, we call $T$ strictly positive. We call a symmetric operator $T$ with domain $\mathfrak{D} (T)$ bounded from below (or semi-bounded) if there exists a real number $c$ such that $\braket{\psi, T \psi} \geq c \| \psi \|^2$ for all $\psi \in \mathfrak{D} (T)$. The form domain of an operator $T$ is denoted by $\mathcal{Q} (T)$. The Fourier transform of a function $f$ on $\R^d$, denoted $\hat{f}$, and the inverse Fourier transform are, respectively, defined as
\begin{equation}
    \hat{f} (\boldsymbol{p}) = \int_{\R^d} f (\boldsymbol{x}) \; \E^{- \I \boldsymbol{p} \cdot \boldsymbol{x}} \; \D^d x \; , \qquad f (\boldsymbol{x}) = \int_{\R^d} \hat{f} (\boldsymbol{p}) \; \E^{\I \boldsymbol{p} \cdot \boldsymbol{x}} \; \frac{\D^d p}{(2 \pi)^d} \; .
\end{equation}

Let $\Omega \subset \R^d$ be open. The closure of $\Omega$ is written $\overline{\Omega}$. We denote by $C_0^k (\Omega)$ (including the case $C_0^k (\R^d)$) the set of compactly supported continuous functions on $\Omega$ with $k$ continuous derivatives, where $k = 0,1,2,\ldots$ or $k=\infty$. Recall that the support of a function $f : \Omega \to \C$, denoted $\mathrm{supp} \, f$, is the closure (in $\R^d$) of the set of points in $\Omega$ where $f$ is non-zero. In particular, a function on $\Omega \subset \R^d$ is said to be compactly supported if its support is a bounded subset of $\Omega$. A special and important class of compactly supported functions is the set of test functions $C_0^\infty (\Omega)$, which consists of smooth functions on $\Omega$ with compact support. Furthermore, we denote by $C^k (\overline{\Omega})$ the set of restrictions of functions in $C_0^k (\R^d)$ to $\Omega$. In general, functions in $C_0^k (\Omega)$ and $C^k (\overline{\Omega})$ are complex-valued. However, it is sometimes convenient to consider only real-valued such functions, in which case we explicitly state this in the text.

We say an open subset $\Omega \subset \R^d$ is $C^0$ (i.e., continuous), respectively Lipschitz, respectively $C^k$, $k \in \N$, if its boundary $\partial \Omega$ can locally be represented by the graph of a continuous, respectively Lipschitz continuous, respectively $k$ times continuously differentiable function from $\R^{d-1}$ to $\R$, such that $\Omega$ lies only on one side of the boundary $\partial \Omega$, see \cite[Sec.~1.1]{Chandler2017} or \cite[Def.~1.2.1.1]{Grisvard2011}. Recall that every $C^k$-region is also Lipschitz, but Lipschitz continuity is a stronger notion than continuity.

Throughout this paper, we work in natural units, i.e., $\hbar = c = k_\mathrm{B} = 1$.

\section{Functional Integrals}\label{sec:func_integrals}

In this Section, we briefly sketch how a non-interacting scalar field theory can be realized as a Gaussian measure on a suitable infinite dimensional vector space.  This formalizes the idea of the Euclidean path integral and makes the connection to probability theory concrete. In contrast to the case of a Gaussian measure on $\R^n$, significantly more subtleties arise in the infinite dimensional case, some of which will be important in the discussion of information-theoretic aspects in Sections \ref{sec:info_theory} and \ref{sec:applications}. We follow the procedure of constructing Gaussian measures on spaces of distributions (or generalized functions) as is common in the literature, see, e.g., \cite{Velo1973,Nelson1973a,Nelson1973b,Glimm2012,Dimock2011,Simon1974,Guerra1975a,Guerra1975b,Guerra1976}. As we wish to keep the discussion in this Section intuitive, much of the technical details are relegated to Appendix \ref{app:func_integrals}. We note that the contents of Appendix \ref{app:func_integrals} are not essential for understanding the results obtained in Section \ref{sec:applications}. However, it is the basis for the derivation of conditions for equivalence of Gaussian measures (cf. Appendix \ref{app:equivalence}), which will be used frequently in this work. The reader who is not interested in the derivation of these conditions but only in their consequences may skip the Appendices altogether.

Functional integrals arise, for example, in the description of second order phase transitions of systems such as uniaxial and isotropic (anti-)ferromagnets (like the Ising model), fluids, superfluids and superconductors \cite{Hohenberg2015}. Such a model is defined at a fundamental microscopic length scale $\Lambda^{-1}$, where $\Lambda$ is some high momentum (or ultraviolet) cut-off scale. In the case of a spin model, for example, this microscopic scale is the lattice constant of the atomic lattice. The interactions of the microscopic degrees of freedom are assumed to be short-range. The length scale $L$ at which one conducts experiments is called the macroscopic scale and typically one can assume $\Lambda^{-1} \ll L$.

In the vicinity of a critical point, the correlation length $\xi$ is much greater than the microscopic scale and even diverges as the system approaches the critical temperature. When the system is close to a critical point and the correlation length is much larger than the microscopic scale but still significantly smaller than the macroscopic scale, we introduce a mesoscale $\lambda \lesssim \xi$ such that the system exhibits a scale hierarchy $\Lambda^{-1} \ll \lambda \ll L$. At the mesoscopic scale the system is approximately homogeneous and fluctuations at scales between the microscopic and mesoscopic scales are small. It is then reasonable to average out these small scale fluctuations. One then postulates that, close to a critical point, the thermodynamic partition sum of the system can be approximated by a formal integration over all configurations $\eta (x)$ of the order parameter, fluctuating at scales between $\lambda$ and $\Lambda^{-1}$. More explicitly \cite{Goldenfeld1992,Cardy1996,Itzykson1991},
\begin{equation}
    Z \approx \int \mathcal{D} \eta \; \E^{- \beta H_\lambda [\eta]} \; ,
\end{equation}
where $H_\lambda$ is an effective Hamiltonian.

If we assume the effective Hamiltonian to be quadratic in the order parameter $\eta$, we arrive at the \emph{Gaussian model}, whose information-theoretic properties we study in this work. The Gaussian model is the first correction to Landau's mean field theory \cite[Ch.~XIV]{Landau1966}. It takes fluctuations of the order parameter into account but assumes that the fluctuations follow a normal distribution around some mean value. It is well known that the validity of the Gaussian approximation depends crucially on the dimension of space. For Ising-type systems, the space dimension above which mean field theory and the Gaussian approximation can reliably describe critical phenomena, called the upper critical dimension $d_\mathrm{c}$, is four. The origin of the specific value of the upper critical dimension can be explained via the renormalization group \cite{Wilson1971a,Wilson1971b}. For $d > d_\mathrm{c}$, all higher order couplings become irrelevant and the theory is trivial, while for $d < d_\mathrm{c}$ there are relevant higher order couplings which need to be taken into account nonperturbatively close to the critical point. In the edge case $d = d_\mathrm{c}$, higher order couplings can be shown to be marginally irrelevant. While for $d > d_\mathrm{c}$ the Gaussian approximation yields a reasonable description of the system everywhere in the phase diagram, for the physically relevant spatial dimensions $d \leq 3$ the Gaussian approximation breaks down around the critical point. However, at sufficiently high temperatures the Gaussian contributions dominate and the Gaussian model is valid.

Another important area of application of functional integrals are Euclidean quantum field theories, i.e., relativistic quantum field theories analytically continued to imaginary time. Euclidean quantum field theories have the same structure as theories in classical statistical mechanics \cite{Symanzik1966}. In particular, they can be seen as probability theories defined by a Gibbs-type measure on some infinite dimensional space of ``field configurations''. Under certain conditions \cite{Nelson1973a,Nelson1973b,Osterwalder1973,Osterwalder1974,Klein1981}, one can recover the relativistic theory from these probability theories. This provides a connection between functional probability measures and relativistic quantum field theories. The Gaussian field theory considered in this work, which corresponds to a single massive scalar boson field without any interactions, is one of the simplest quantum field theories one can consider. Nevertheless, a sound understanding of the Gaussian theory is necessary for the treatment of interacting theories. For example, at least for weak interactions, one can treat the interacting theory as a perturbation of the free theory, i.e., the theory without interactions \cite{Glimm2012,Guerra1975a,Guerra1975b,Simon1974}. The free theory thus provides a starting point for the study of more complicated theories. Furthermore, certain aspects of interacting theories are already captured in the corresponding free theory. In Section \ref{sec:mutual_info}, we study the mutual information between two regions in Euclidean space and argue that the observed behaviour is due to the Markov property of the scalar field. However, this Markov property stems from the kinetic, i.e., non-interacting, part of the action containing the differential operator, see also the discussion in Section \ref{sec:Gaussian_measures}. Therefore, we may conjecture that certain features of the mutual information in the free theory carry over to interacting theories.

At this point, we note that in the remainder of this work, we employ the terminology from Euclidean quantum field theory. In particular, we interpret $\R^d$ as Euclidean spacetime, i.e., Minkowski spacetime Wick rotated to imaginary time. Nevertheless, we remind the reader that the Gaussian field theories considered in the following may also describe first order approximations of classical statistical systems close to a second order phase transition. What corresponds to a mass in the relativistic context has then the significance of an inverse correlation length, $m = \xi^{-1}$.

We proceed as follows. In Section \ref{sec:Gaussian_measures}, we discuss Gaussian measures corresponding to free field theories. Starting from a massive non-interacting scalar field theory on a finite lattice with $N$ lattice points, which is described by an ordinary Gaussian measure on $\R^N$, we motivate the question on how to define the corresponding continuum theory as a Gaussian measure on a space of ``field configurations''. Section \ref{sec:covariance_operators} is devoted to the covariance operators used in this work, i.e., inverses of the differential operator $- \Delta + m^2$. Since we want to describe fields on some bounded subregion $\Omega$ of Euclidean space $\R^d$, we have to choose boundary conditions for the field configurations to satisfy. These boundary conditions are incorporated by choosing self-adjoint extensions of the differential operator $- \Delta + m^2$ defined on the space of test functions $C_0^\infty (\Omega)$. We will demonstrate the connection between boundary conditions and self-adjoint extensions of the differential operator $- \Delta + m^2$ by using the concept of quadratic forms.

We emphasize that no new results are obtained in this Section. The purpose of this Section, together with Appendix \ref{app:func_integrals}, is to make this paper sufficiently self-contained. Standard references for the construction of Gaussian measures describing Euclidean bosonic quantum field theories include \cite{Velo1973,Nelson1973a,Nelson1973b,Glimm2012,Dimock2011,Simon1974,Guerra1975a,Guerra1975b,Guerra1976}. For a treatment of measures, in particular Gaussian measures, on infinite dimensional spaces, see, e.g., \cite{Bogachev2014,Bogachev2015,Gelfand1964,Schwartz1973}. Our discussion of unbounded self-adjoint operators and quadratic forms, especially with regard to the Laplacian, is largely based on \cite{Reed1981,Reed1975,Reed1978,Faris1975,Kato1995,Robinson1971,Davies1996,Evan2010}. An exposition of Sobolov spaces, which we need in Section \ref{sec:covariance_operators}, can be found in \cite{McLean2000,Chandler2017,Mazya2011} (see also Appendix \ref{app:sobolev_spaces}).

\subsection{Free Scalar Field Theories}\label{sec:Gaussian_measures}

Recall from the beginning of Section \ref{sec:func_integrals} that in the functional integral formalism, a free scalar (statistical or Euclidean quantum) field theory is defined by a Gaussian probability measure. Formally, this measure is given by the expression
\begin{equation}
    \D \mu = \frac{1}{Z} \, \exp \left[ - S_\mathrm{E} [\varphi] \right] \, \mathcal{D} \varphi \; ,
\end{equation}
where $Z$ is a normalization constant, $\mathcal{D} \varphi$ is a formal Lebesgue measure on the space of field configurations and $S_\mathrm{E} [\varphi]$ is the Euclidean action functional given by
\begin{equation}
    S_\mathrm{E} [\varphi] = \frac{1}{2} \int_\Omega \varphi (x) \left( - \Delta + m^2 \right) \varphi (x) \; \D^d x \; .
\end{equation}
The aim of this Section (and of Appendix \ref{app:func_integrals}) is to give a more precise meaning to this formal object. Note that throughout this work, we will only consider theories with vanishing mean value.

It is instructive to first consider this field theory regularized on a finite spacetime lattice $\Gamma$, see, e.g., \cite[Sec.~1.3]{Salmhofer2007}. Assuming that this lattice has $N$ lattice sites, the Gaussian measure defining the theory is just a centred Gaussian measure on $\R^N$ given by
\begin{equation}\label{eq:lattice_measure}
    \D \mu_\Gamma (\boldsymbol{\varphi}) = \frac{1}{\sqrt{\det (2 \pi C)}} \; \exp \left[ - \frac{1}{2} \, \boldsymbol{\varphi}^{\mathsf{T}} C^{-1} \boldsymbol{\varphi} \right] \; \D^N \varphi \; ,
\end{equation}
where
\begin{equation}
    \boldsymbol{\varphi}^{\mathsf{T}} C^{-1} \boldsymbol{\varphi} \coloneqq \sum_{i,j=1}^N \varphi_i \left( C^{-1} \right)_{ij} \varphi_j \; .
\end{equation}
The field configurations are thus given by real $N$-component vectors $\boldsymbol{\varphi}$. The inverse covariance matrix in \eqref{eq:lattice_measure}, sometimes called the precision matrix, is given by
\begin{equation}
    C^{-1} \coloneqq - \Delta_\Gamma + m^2 \mathds{1}_N \; ,
\end{equation}
where $\mathds{1}_N$ is the $N \times N$ unit matrix and $- \Delta_\Gamma$ is the lattice Laplacian, see, e.g., \cite[Sec.~1.3]{Salmhofer2007}. As we work on a finite lattice, some boundary conditions have to be imposed in the definition of $- \Delta_\Gamma$. A typical choice would be periodic boundary conditions.

We conclude that a free scalar field theory (with vanishing mean) on a finite lattice is uniquely defined by specifying its covariance matrix, or, equivalently, its precision matrix. This precision matrix is given by a discretized version of the differential operator $- \Delta + m^2$. Different choices of the mass $m > 0$ lead to different theories. But the mass parameter is not the only thing we can vary in order to achieve different theories. As already mentioned above, on a finite lattice we have to impose boundary conditions in the definition of the lattice Laplacian $- \Delta_\Gamma$. From the definition of a Gaussian measure on $\R^n$ (cf. \eqref{eq:def_gaussian_Rn}), we see that the choice of boundary conditions must be such that the precision matrix $C^{-1}$ is symmetric. In the continuum theory, the analogue of this is the requirement that the boundary conditions of $- \Delta + m^2$ are chosen such that this operator is self-adjoint as an operator on an $L^2$-space, cf. Appendix \ref{app:func_integrals}. The connection between boundary conditions and self-adjoint extensions of the Laplacian is discussed in Section \ref{sec:covariance_operators}. Instead of specifying the precision matrix, we can also uniquely define the theory by providing a covariance matrix. Of course, assuming that the covariance matrix is invertible, this is the same thing as providing a precision matrix. However, when dealing with theories of particular interest, it is sometimes more convenient to provide one rather than the other. Since a field theory is usually defined in terms of an action, it is often more opportune to specify the precision matrix, e.g., when we wish to impose certain boundary conditions directly. In some cases, however, we may wish to specify the covariance matrix. One such case is discussed in Section \ref{sec:covariance_operators}, the case of free boundary conditions. In this case, the covariance describing these boundary conditions is obvious, while the boundary conditions themselves are complicated. Due to their importance, we devote Section \ref{sec:covariance_operators} to discuss the covariance operators considered in this work.

After this brief discussion of free scalar field theories on the lattice, we now study how to realise a free scalar field theory in the continuum. We can think of the continuum theory as the limit\footnote{For a detailed discussion see, e.g., \cite[Ch.~VIII]{Simon1974} as well as \cite{Angelis1979,Capitani1991}.} of the lattice theory when the lattice constant $\varepsilon$ is sent to zero, the number of lattice sites $N$ is sent to infinity and the system size $N \varepsilon^d$ is kept fixed. Heuristically, in this limit, the lattice Laplacian $- \Delta_\Gamma + m^2 \mathds{1}_N$ becomes the usual ``continuum'' Laplacian $- \Delta + m^2$ satisfying some boundary conditions, which may themselves follow from the continuum limit of the lattice formulation. In more precise terms, the covariance of the lattice theory, uniquely defined by the lattice Laplacian with some choice of boundary conditions, converges to the covariance of the continuum theory, uniquely defined by the continuum Laplacian with boundary conditions stemming from the lattice theory, see \cite[Ch.~VIII]{Simon1974}. Furthermore, we may expect the field configurations $\boldsymbol{\varphi}$ to become functions\footnote{As discussed in Appendix \ref{app:func_integrals}, we will actually realize the field theory as a Gaussian measure on a space of generalized functions or distributions. See also \cite[Sec.~\S I.2]{Simon1974} and \cite{Sheffield2007}.} $\varphi (x)$ defined on the Euclidean spacetime region $\Omega$. Assuming a vanishing mean field configuration, the Gaussian measure (and thus the field theory) is completely determined by a covariance operator $\hat{C} = (- \Delta + m^2)^{-1}$, where suitable boundary conditions are to be imposed in the precise definition of the differential operator $- \Delta + m^2$.

Since the theory is determined by (the inverse of) a differential operator, its properties depend non-trivially on $d$, the dimension of Euclidean spacetime. In particular, the dimension $d$ determines the asymptotic behaviour of the spectrum of the covariance operator, i.e., the ultraviolet (UV) properties of the covariance operator, see Section \ref{sec:covariance_operators}. For example, in $d = 1$, the theory is UV-finite, since the corresponding propagator is continuous. As a consequence, in this case we may realize the theory as a Gaussian measure on the Hilbert space $L^2 (\Omega) = L^2 (\Omega, \D x)$ of square-integrable \emph{functions}, cf. Mourier's Theorem \cite{Mourier1953} (Theorem \ref{th:Mourier}). This crucial dependence of the UV properties of the covariance operator on the Euclidean spacetime dimension also has consequences for the relative entropy. In Section \ref{sec:diff_masses}, we demonstrate that the relative entropy between two field theories on a bounded region $\Omega \subset \R^d$ is finite for $d < 4$ if they have different masses, and is finite for $d < 3$ if they have different Robin-type boundary conditions.

We conclude this Section by noting that one may also think of a free scalar field theory as a Gaussian random process \cite{Lifshits2012}. In particular, the field $\varphi$ (the integration variable in the functional integral) is interpreted as Gaussian random variable indexed by (test) functions. It was shown by Nelson \cite{Nelson1973a,Nelson1973b} that this random process has the Markov property and is thus a Markov random field \cite{Kindermann1980}. Heuristically, the Markov property of the free scalar field stems from the fact that the Euclidean action functional only contains ``nearest-neighbour'' interactions introduced by the Laplacian. In the lattice-regularized theory, the nearest-neighbour property becomes manifest, and the theory is essentially an Ising ferromagnet, see also \cite[Sec.~IV]{Guerra1975b} and \cite[Sec.~IX.1]{Guerra1976}.

\subsection{Boundary Conditions and Covariance Operators}\label{sec:covariance_operators}

In this Section we discuss the covariance operators needed for the rest of this work. Recall from the previous Sections that we can completely specify a free scalar field theory (with vanishing mean) by providing the differential operator $- \Delta + m^2$ in the Euclidean action. The correlations between different Euclidean spacetime points $\boldsymbol{x}$ and $\boldsymbol{y}$ are described the Green's function $G (\boldsymbol{x},\boldsymbol{y})$, which is the integral kernel of the inverse of the differential operator, and we have the identity
\begin{equation}\label{eq:DG=1}
    (- \Delta + m^2) \, G (\boldsymbol{x},\boldsymbol{y}) = \delta^{(d)} (\boldsymbol{x}-\boldsymbol{y}) \; .
\end{equation}
When we define the field theory on a bounded open subset $\Omega \subset \R^d$, we need to specify some boundary conditions in the definition of the operator $- \Delta + m^2$. The Green's function $G(\boldsymbol{x},\boldsymbol{y})$ depends on the choice of these boundary conditions, thus different boundary conditions describe different correlations of Euclidean spacetime points even within the region $\Omega$. The boundary conditions considered in this work are the classical choices of Dirichlet, Neumann, periodic, Robin and free boundary conditions, which are also studied in \cite{Guerra1976,Glimm2012}. In this Section, we explore how boundary conditions are related to self-adjoint realizations of the formal operator $- \Delta + m^2$ as operators on $L^2 (\Omega)$. We then state important properties of and between these self-adjoint operators.

\subsubsection{Covariance Operator of Massive Scalar Field Theory}

Before we study the covariance operators of field theories on a bounded subregion $\Omega$ (which we shall refer to as the local case), it is instructive to start with the covariance of a field defined on all of $\R^d$ (the global case). We demonstrate that in this case the question of a realization as a self-adjoint operator on the Hilbert space of square-integrable functions is greatly simplified. Throughout this Section, let $m > 0$. We denote\footnote{As stated in \cite[Sec.~3.5]{Davies1996}, the results presented in this paragraph also hold if we replace $C_0^\infty (\R^d)$ by $\mathscr{S} (\R^d)$, the Schwartz space of smooth functions of rapid decrease.} by $C_0^\infty (\R^d)$ the real vector space of real-valued smooth functions on $\R^d$ with compact support. We define an operator on $L^2 (\R^d)$, the real Hilbert space of real-valued square-integrable functions on $\R^d$, that acts on functions $f \in C_0^\infty (\R^d)$ as $f \mapsto (- \Delta + m^2) f$, where $\Delta$ is the $d$-dimensional Laplacian, and denote this operator by $(- \Delta + m^2)|_{C_0^\infty}$. It is a densely defined, strictly positive and symmetric operator on $L^2 (\R^d)$ that is bounded from below by $m^2$. Since $(- \Delta + m^2)|_{C_0^\infty}$ is symmetric, it is closable \cite[p.~255]{Reed1981}. Moreover, it can be shown that $(- \Delta + m^2)|_{C_0^\infty}$ is essentially self-adjoint as an operator on $L^2 (\R^d)$ \cite[Thm.~3.5.3]{Davies1996}. Therefore, $(- \Delta + m^2)|_{C_0^\infty}$ has one and only one self-adjoint extension \cite[p.~256]{Reed1981}, namely its closure
\begin{equation}\label{eq:D_global}
    D \coloneqq \overline{(- \Delta + m^2 )|_{C_0^\infty}}^{L^2 (\R^d)} \; .
\end{equation}
Note that $D$ is the operator whose graph, $\Gamma (D)$, is the closure of the set $\{ (f, (-\Delta + m^2) f) : f \in  C_0^\infty (\R^d) \}$ in $L^2 (\R^d) \oplus L^2 (\R^d)$. Hence, the domain of $D$, denoted $\mathfrak{D} (D)$, is the completion of $C_0^\infty (\R^d)$ in the norm $|\!|\!|f|\!|\!| = \|f\|_{L^2} + \| D f \|_{L^2}$, see \cite[Prob.~VIII.15]{Reed1981}. One can show that $\mathfrak{D} (D) = H^{+2} (\R^d)$, the Hilbert-Sobolev space of order $2$ over $\R^d$, for definitions see Appendix \ref{app:sobolev_spaces} and references therein. As $D$ is the closure of a strictly positive operator, it is also strictly positive and furthermore it is also bounded from below by $m^2$, see \cite[\S14]{Faris1975}.

The inverse of the differential operator in \eqref{eq:D_global} is the covariance operator of a free massive scalar field theory over Euclidean spacetime $\R^d$, see \cite{Glimm2012,Dimock2011}. Notice that, interpreting the field theory as a Euclidean \emph{quantum} field theory, this choice of covariance operator describes the vacuum representation of the corresponding relativistic quantum field theory with $d-1$ space and one time dimension. Alternatively, we could compactify Euclidean spacetime $\R^d$ along one coordinate direction (which is then interpreted as the Euclidean time direction) into the circle $\mathbb{S}_\beta$ of circumference $\beta$. The inverse of the unique self-adjoint extension of $(- \Delta + m^2)|_{C_0^\infty}$ on $\mathbb{S}_\beta \times \R^{d-1}$ is then the covariance operator of a Euclidean quantum field theory describing a relativistic quantum field theory in $d-1$ space dimensions at inverse temperature $\beta = 1 / T$, see \cite{Hoegh1974,Klein1981}. The vacuum theory can then formally be obtained from this theory by taking the limit $\beta \to +\infty$, i.e., the vacuum corresponds to vanishing temperature.

As $D$ is bounded from below by $m^2 > 0$, it is a bijection from its domain $H^{+2} (\R^d)$ onto $L^2 (\R^d)$, cf. Proposition \ref{prop:bounded_below_bijective}. Its inverse $D^{-1}$ is a strictly positive bounded self-adjoint pseudo-differential operator with $\|D^{-1}\| \leq m^{-2}$ (cf. Corollary \ref{cor:bounded_inverse}), called the Green's operator, acting on functions $f \in L^2 (\R^d)$ as $f \mapsto G * f$, where $*$ denotes convolution and $G$ is a symmetric integral kernel of positive type, called the fundamental solution of $D$, given by \cite[Sec.~1.5]{Salmhofer2007}
\begin{align*}
    G (\boldsymbol{x},\boldsymbol{y};m) &= \int_{\R^d} \E^{\I \boldsymbol{p} \cdot (\boldsymbol{x}-\boldsymbol{y})} \, (|\boldsymbol{p}|^2 + m^2)^{-1} \; \frac{\D^d p}{(2 \pi)^d} \\
    &= \frac{1}{(2 \pi)^{\frac{d}{2}}} \left( \frac{m}{|\boldsymbol{x}-\boldsymbol{y}|} \right)^{\frac{d}{2}-1} K_{\frac{d}{2}-1} (m |\boldsymbol{x}-\boldsymbol{y}|) \; , \numberthis \label{eq:fundamental_solution}
\end{align*}
where $\boldsymbol{x}, \boldsymbol{y} \in \R^d$ and $\boldsymbol{x} \neq \boldsymbol{y}$. Here, $K_\alpha (z)$ is the modified Bessel function of the second kind \cite{Watson1995}. For Euclidean spacetime dimensions $d \leq 3$, the fundamental solution can be written as
\begin{alignat}{2}
  &d = 1: \qquad G (x,y;m) &&= \frac{1}{2m} \, \E^{-m |x-y|} \; , \label{eq:fund_sol_1d} \\
  &d = 2: \qquad G (\boldsymbol{x},\boldsymbol{y};m) &&= \frac{1}{2\pi} \, K_0 (m |\boldsymbol{x}-\boldsymbol{y}|) \; , \\
  &d = 3: \qquad G (\boldsymbol{x},\boldsymbol{y};m) &&= \frac{1}{4\pi} \, \frac{\E^{-m |\boldsymbol{x}-\boldsymbol{y}|}}{|\boldsymbol{x}-\boldsymbol{y}|} \; .
\end{alignat}
Notice that in $d=1$, the fundamental solution is a continuous function on $\R$ and in particular, it is finite for coinciding points $x = y$. For $d \geq 2$, the fundamental solution diverges ``on the diagonal''. This divergence is called an ultraviolet (UV) divergence \cite{Gurau2014}. Thus, we may call a free scalar field theory in $d=1$ UV-finite. Another manifestation of this will be the fact that the covariance operator for a scalar field over a bounded domain will be of trace class in $d=1$, see the discussion below. We note that for small distances $\varepsilon \coloneqq |\boldsymbol{x}-\boldsymbol{y}| \ll m^{-1}$, the fundamental solution behaves like \cite[Lem.~1.10]{Salmhofer2007}
\begin{equation}\label{eq:greens_diagonal}
    G (\varepsilon;m) = \frac{\mathrm{const.}}{m^{2-d}} \begin{cases} 
      \log \varepsilon^{-1} \quad & \mathrm{for} \;\; d = 2 \\
      \varepsilon^{-(d-2)} \quad & \mathrm{for} \;\; d \geq 3
   \end{cases} \; .
\end{equation}
Thus, the theory becomes in this sense ``more singular'' in higher dimensions.

\subsubsection{Classical Boundary Conditions}\label{sec:classical_bcs}

We now focus on the covariance operators of a free field theory over a bounded region\footnote{In practice, we will only be interested in very well-behaved regions, such as spheres and cubes as well as unions of finite disjoint families thereof. Hence, we will not address possible complications that may arise for very pathological choices of $\Omega$.} $\Omega \subset \R^d$. Some properties of the global covariance discussed above, like the $d$-dependence of the character of the singularity around the diagonal, carry over to the local case. The most prominent difference between the global and the local case is the need to impose suitable boundary conditions. As argued above, there is one and only one self-adjoint extension of $(- \Delta + m^2)|_{C_0^\infty}$ when $\Omega = \R^d$. For bounded regions, however, there are uncountable infinitely many self-adjoint extensions. We demonstrate in the following that different self-adjoint extensions correspond to different boundary conditions. We first introduce boundary conditions via the variational principle, an overview of which can be found, e.g., in \cite{Tsang2000} and \cite{Olver1993}. The connection to self-adjoint extensions of $(- \Delta + m^2)|_{C_0^\infty}$ is then made via quadratic forms and a representation Theorem by Kato \cite{Kato1995}.

Let $\Omega \subset \R^d$ be open, bounded and with sufficiently smooth boundary $\partial \Omega$. For every $f \in C^\infty (\overline{\Omega})$, we define the Euclidean action functional
\begin{align*}
    S_\mathrm{E} [f; b] &= S^\Omega_\mathrm{E} [f] + S^{\partial \Omega}_\mathrm{E} [f;b] \\
    &= \frac{1}{2} \int_\Omega \left[ \left( \boldsymbol{\nabla} f (\boldsymbol{x}) \right)^2 + m^2 \left( f (\boldsymbol{x}) \right)^2 \right] \D^d x \\
    &+ \frac{1}{2} \iint_{\partial \Omega} b (\boldsymbol{x}, \boldsymbol{y}) f (\boldsymbol{x}) f (\boldsymbol{y}) \; \D S (\boldsymbol{x}) \, \D S (\boldsymbol{y}) \; , \label{eq:Eucl_action_bcs} \numberthis
\end{align*}
where $b$ is a symmetric (possibly formal) distributional kernel. To such an action we can define a corresponding bilinear form via\footnote{For a single Euclidean spacetime dimension $d=1$ we construct the boundary part of the action as
\begin{equation}
    S_\mathrm{E}^{\partial \Omega} [f;b] = \frac{1}{2} \sum_{x,y \in \partial \Omega} b (x,y) f (x) f (y) \; ,
\end{equation}
where the sum goes over the boundary points. The corresponding term in the bilinear form is constructed analogously.
}
\begin{align*}
    \mathfrak{q}_b (f,g) &= \int_\Omega \left[ \left( \boldsymbol{\nabla} f (\boldsymbol{x}) \right) \cdot \left( \boldsymbol{\nabla} g (\boldsymbol{x}) \right) + m^2  f (\boldsymbol{x}) g (\boldsymbol{x}) \right] \D^d x \\
    &+ \iint_{\partial \Omega} b (\boldsymbol{x}, \boldsymbol{y}) f (\boldsymbol{x}) g (\boldsymbol{y}) \; \D S (\boldsymbol{x}) \, \D S (\boldsymbol{y}) \; . \label{eq:Eucl_action_bilinear} \numberthis
\end{align*}
The kernel $b$ specifies the boundary conditions on $\partial \Omega$. For example, if we want to describe local boundary conditions, $b$ takes the form
\begin{equation}\label{eq:local_kernel_b}
    b (\boldsymbol{x}, \boldsymbol{y}) = \delta^{(d-1)} (\boldsymbol{x} - \boldsymbol{y}) \, \widetilde{b} (\boldsymbol{y}) \; ,
\end{equation}
where $\widetilde{b}$ is a function on $\partial \Omega$.

Upon varying the Euclidean action, we obtain
\begin{align*}
    \delta S_\mathrm{E} &= \int_\Omega \left( - \Delta f (\boldsymbol{x}) + m^2 f (\boldsymbol{x}) \right) \delta f (\boldsymbol{x}) \; \D^d x \\
    &+ \int_{\partial \Omega} \frac{\partial f}{\partial n} (\boldsymbol{x}) \, \delta f (\boldsymbol{x}) \; \D S (\boldsymbol{x}) + \iint_{\partial \Omega} b (\boldsymbol{x}, \boldsymbol{y}) f (\boldsymbol{x}) \, \delta f (\boldsymbol{y}) \; \D S (\boldsymbol{x}) \, \D S (\boldsymbol{y}) \; , \numberthis
\end{align*}
where $\partial / \partial n \coloneqq \boldsymbol{n} \cdot \boldsymbol{\nabla}$ and $\boldsymbol{n}$ is the outward pointing unit vector on $\partial \Omega$. The principle of stationary action, $\delta S_\mathrm{E} = 0$, then yields the following set of equations,
\begin{align}
    ( - \Delta + m^2 ) f (\boldsymbol{x}) &= 0 \; , \qquad \boldsymbol{x} \in \Omega \; , \\
    \frac{\partial f}{\partial n} (\boldsymbol{x}) + \int_{\partial \Omega} b (\boldsymbol{x},\boldsymbol{y}) f (\boldsymbol{y}) \; \D S (\boldsymbol{y}) &= 0 \; , \qquad \boldsymbol{x} \in \partial \Omega \; . \label{eq:boundary_condition_bkernel}
\end{align}
We thus see that the action \eqref{eq:Eucl_action_bcs} indeed yields the homogeneous boundary value problem described by the kernel $b$. In the following, we make the connection between actions of the form \eqref{eq:Eucl_action_bcs} and self-adjoint extensions of $(- \Delta + m^2)|_{C_0^\infty}$. We start with the classical choices of Dirichlet, Neumann and Robin boundary conditions. All these boundary conditions are local in the sense that the kernel $b$ is of the form \eqref{eq:local_kernel_b}. Thus, in this case we only need to specify the function $\widetilde{b}$. In addition, we study periodic boundary conditions, which can be written as a periodic sum of local boundary conditions of the form \eqref{eq:local_kernel_b}.

The first boundary conditions we consider are of Dirichlet type, which corresponds to $f \equiv 0$ on $\partial \Omega$. Formally, we can incorporate Dirichlet boundary conditions by choosing $\widetilde{b} \equiv + \infty$. To make the discussion more rigorous, instead of employing an infinite boundary term, we change the domain of the Dirichlet action to a subspace of functions which directly satisfy Dirichlet boundary conditions. Following the terminology of \cite{Tsang2000}, we call such boundary conditions essential. If the boundary conditions are such that the kernel $b$ can be chosen in a way that the corresponding action is well-defined, we call them natural. Examples of natural boundary conditions are Neumann and local Robin boundary conditions discussed below.

In this spirit, we define the Dirichlet form $\mathfrak{q}_\mathrm{D}$ on $C_0^\infty (\Omega) \times C_0^\infty (\Omega) \subset C^\infty (\overline{\Omega}) \times C^\infty (\overline{\Omega})$ via
\begin{equation}
    \mathfrak{q}_\mathrm{D} (f,g) = \int_\Omega \left( \boldsymbol{\nabla} f \cdot \boldsymbol{\nabla} g + m^2 f g \right) \; \D^d x \; .
\end{equation}
Clearly, $\mathfrak{q}_\mathrm{D}$ is densely defined (as a form on $L^2 (\Omega)$), symmetric and bounded from below by $m^2 > 0$. Furthermore, since the gradient $\boldsymbol{\nabla}$ defined on $C^\infty_0 (\Omega)$ is closable as an operator on $L^2 (\Omega)$ \cite[p.~14]{Robinson1971}, the Dirichlet form is closable \cite[Prop.~2.2]{Faris1975}. Therefore, by Theorem \ref{thm:second_rep_thm}, the closure $\overline{\mathfrak{q}_\mathrm{D}}$ determines a unique positive self-adjoint operator $- \Delta_\mathrm{D} + m^2$ via
\begin{equation}
    \overline{\mathfrak{q}_\mathrm{D}} (f,g) = \braket{f, (- \Delta_\mathrm{D} + m^2) g}_{L^2 (\Omega)} \; , \qquad f \in \mathcal{Q} (- \Delta_\mathrm{D}) \; , \;\; g \in \mathfrak{D} (- \Delta_\mathrm{D}) \; ,
\end{equation}
where $\mathcal{Q} (- \Delta_\mathrm{D})$ denotes the form domain of the operator $- \Delta_\mathrm{D}$, see Appendix \ref{app:quadratic_forms}. For obvious reasons, the operator $- \Delta_\mathrm{D}$ is called the Dirichlet Laplacian as, in a certain sense, the functions in its domain vanish on $\partial \Omega$\footnote{More precisely, if $\partial \Omega$ is $C^2$, it can be shown that $\mathfrak{D} (- \Delta_\mathrm{D}) = H^{+2} (\Omega) \cap H^{+1}_0 (\Omega)$, see \cite[Sec.~6.3.2]{Evan2010}. Thus, the domain of the Dirichlet Laplacian consists precisely of those functions in $H^{+2} (\Omega)$ with vanishing trace, see Appendix \ref{app:sobolev_spaces}.}.

The next kind of boundary conditions we discuss are Neumann boundary conditions, which correspond to $\partial f / \partial n = 0$ on $\partial \Omega$, i.e., $\widetilde{b} \equiv 0$. The Neumann form $\mathfrak{q}_\mathrm{N}$ is thus defined on $C^\infty (\overline{\Omega}) \times C^\infty (\overline{\Omega})$ via
\begin{equation}\label{eq:Neumann_form}
    \mathfrak{q}_\mathrm{N} (f,g) = \int_\Omega \left( \boldsymbol{\nabla} f \cdot \boldsymbol{\nabla} g + m^2 f g \right) \; \D^d x \; .
\end{equation}
Just like in the case of the Dirichlet form, the closure $\overline{\mathfrak{q}_\mathrm{N}}$ of the Neumann form uniquely determines a positive self-adjoint operator $- \Delta_\mathrm{N} + m^2$, where $- \Delta_\mathrm{N}$ is the Neumann Laplacian, cf. \cite[Sec.~XIII.15]{Reed1978}.

The choice $\widetilde{b} = \sigma$, where $\sigma$ is a continuous function on $\partial \Omega$, describes local Robin boundary conditions, i.e., $\partial f / \partial n = - \sigma f$ on $\partial \Omega$\footnote{Recall that we defined $\partial / \partial n$ to be the outward normal derivative. In contrast, \cite{Guerra1975a,Robinson1971} define $\partial / \partial n$ to be the inward normal derivative. Thus, there is an additional minus sign in the definition of Robin boundary conditions in our work compared to the above references.}. The Robin form is defined on $C^\infty (\overline{\Omega}) \times C^\infty (\overline{\Omega})$ by
\begin{equation}
    \mathfrak{q}_\sigma (f,g) = \int_\Omega \left( \boldsymbol{\nabla} f \cdot \boldsymbol{\nabla} g + m^2 f g \right) \; \D^d x + \int_{\partial \Omega} \sigma f g \; \D S \; .
\end{equation}
Note that the boundary form $\int_{\partial \Omega} \sigma f g \; \D S$ is infinitesimally relatively form bounded with respect to the Neumann form $\mathfrak{q}_\mathrm{N}$ \cite[p.~34]{Robinson1971}. The closure of $\mathfrak{q}_\sigma$ determines the self-adjoint operator $- \Delta_\sigma + m^2$ (cf. \cite[Sec.~\S2]{Robinson1971}), where $- \Delta_\sigma$ is the Robin Laplacian.

Finally, we introduce periodic boundary conditions, which are another example of essential boundary conditions. Let $\Omega = (0,L)^d$ be an open $d$-cube of edge length $L$. Periodic boundary conditions correspond to the choice
\begin{equation}
    b (\boldsymbol{x}, \boldsymbol{y}) = \lim_{\widetilde{b} \to + \infty} \widetilde{b} \left[ \sum_{r_j = \pm 1} \delta^{(d-1)} \left( \boldsymbol{x} - \boldsymbol{y} \right) - \delta^{(d-1)} \left( \boldsymbol{x} - \boldsymbol{y} + L \boldsymbol{r} \right) \right] \; .
\end{equation}
The boundary condition in \eqref{eq:boundary_condition_bkernel} becomes here
\begin{equation}
    \frac{\partial f}{\partial n} (\boldsymbol{x}) + \lim_{\widetilde{b} \to + \infty} \widetilde{b} \left[ f (\boldsymbol{x}) - f (\boldsymbol{x} + L \boldsymbol{r}) \right] = 0 \; ,
\end{equation}
where $\boldsymbol{x} + L \boldsymbol{r}$ is the boundary point opposite to $\boldsymbol{x}$. In the limit $\widetilde{b} \to + \infty$, this equation yields the usual periodic boundary conditions,
\begin{align}
    f (\boldsymbol{x}) - f (\boldsymbol{x} + L \boldsymbol{r}) &= 0 \; , \\
    \frac{\partial f}{\partial n} (\boldsymbol{x}) + \frac{\partial f}{\partial n} (\boldsymbol{x} + L \boldsymbol{r}) &= 0 \; .
\end{align}

Alternatively, just like for Dirichlet boundary conditions, we can incorporate periodic boundary conditions by choosing the domain of functions suitably. In particular, we define the periodic form
\begin{equation}
    \mathfrak{q}_\mathrm{P} (f,g) = \int_\Omega \left( \boldsymbol{\nabla} f \cdot \boldsymbol{\nabla} g + m^2 f g \right) \; \D^d x \; ,
\end{equation}
on $\mathfrak{D} (P_0) \times \mathfrak{D} (P_0)$, where\footnote{For a motivation of the notation, see \cite[Sec.~\S2]{Robinson1971}.}
\begin{equation}
    \mathfrak{D} (P_0) = \left\{ f \in C^\infty (\overline{\Omega}) \; : \; f |_{x_i = -L/2} = f |_{x_i = +L/2} \; , \quad i = 1, \ldots, d \; \right\} \subset C^\infty (\overline{\Omega}) \; .
\end{equation}
The self-adjoint operator associated with the closures of the periodic form will be denoted by $- \Delta_\mathrm{P} + m^2$.

The above discussion illustrates the connection between boundary conditions and self-adjoint realizations of the formal differential operator $- \Delta + m^2$ as an operator on $L^2 (\Omega)$. The approach of defining different self-adjoint realizations implicitly via forms is perhaps particularly natural from the point of view of a physicist working in field theory: We start by writing down an action which, by the principle of stationary action, yields the desired boundary value problem. The closure of the bilinear form associated with this action then corresponds to a unique self-adjoint realization of $- \Delta + m^2$ as an operator on $L^2 (\Omega)$.

The operators $- \Delta_\mathrm{X} + m^2$, $\mathrm{X} \in \{\mathrm{D}, \mathrm{N}, \mathrm{P}, \sigma \}$, introduced above are strictly positive and we denote their inverses by $\hat{G}_\mathrm{D}$, $\hat{G}_\mathrm{N}$, $\hat{G}_\mathrm{P}$ and $\hat{G}_\sigma$, respectively. These operators are compact integral operators on $L^2 (\Omega)$ with integral kernels given by the corresponding Green's functions. In Section \ref{sec:diff_masses}, the operators $\hat{G}_\mathrm{X}$ will serve as the covariance operators of free massive scalar field theories over a bounded region $\Omega$.

\subsubsection{Free Boundary Conditions}\label{sec:free_bcs}

So far, we have constructed self-adjoint realizations of $- \Delta + m^2$ by specifying bilinear forms. This is particularly convenient when we are interested in the ``classical'' boundary conditions $\{\mathrm{D}, \mathrm{N}, \mathrm{P}, \sigma \}$, as these are easier to implement via the differential operator. The covariance operators of the theories describing fields satisfying these boundary conditions then follow as the inverses of these self-adjoint operators. However, we may also wish to realize a theory with a particular covariance operator and not by imposing any specific boundary conditions. The boundary conditions then follow from the inverse of the prescribed covariance operator. In this case we thus obtain boundary conditions from the covariance and not the covariance from the boundary conditions. An example of this procedure, which we will use in Section \ref{sec:applications}, is presented below \cite{Guerra1975a,Guerra1976,Glimm2012}.

Let us start with a physical motivation. Suppose we are given a free scalar field of mass $m$ on Euclidean spacetime $\R^d$. From this we want to obtain a \emph{reduced} theory, i.e., a theory that fully describes the physics within a bounded region $\Omega$ but contains no information about the physics in the exterior $\R^d \setminus \Omega$. Heuristically, we can obtain such a \emph{local} theory from the \emph{global} theory by integrating out the degrees of freedom in the exterior region, i.e., marginalizing the global probability distribution. Then, all expectation values of observables supported in $\Omega$ can be computed with either the global or the local theory\footnote{This is similar to the concepts of partial trace and reduced density operator in quantum information theory, see, e.g., \cite[Sec.~2.4.3]{Nielsen2010}.}. Since the theory under consideration here is Gaussian with vanishing mean, we can construct such a reduced (or marginalized) theory simply by requiring that it describes the same correlations in the region $\Omega$ as the global theory. This motivates the following discussion\footnote{A similar line of reasoning has recently been employed in the context of local non-equilibrium dynamics of relativistic quantum field theories, leading to the discovery of a ``hyperbolic version'' of the free boundary conditions considered here \cite{Schroefl2024}.}.

Let $\Omega$ be a bounded open subset of $\R^d$. Recall the definition of the fundamental solution $G$ given in \eqref{eq:fundamental_solution}. We define the operator $\hat{G}_0$ on $L^2 (\Omega)$ by
\begin{equation}
    (\hat{G}_0 f) (\boldsymbol{x}) = \int_\Omega G (\boldsymbol{x},\boldsymbol{y};m) f (\boldsymbol{y}) \; \D^d y \; , \qquad f \in L^2 (\Omega) \; .
\end{equation}
From the properties of $D^{-1}$, we can see that $\hat{G}_0$ is a strictly positive, bounded and self-adjoint operator on $L^2 (\Omega)$ with $\|\hat{G}_0\| \leq m^{-2}$. The operator $\hat{G}_0$ is a kind of volume potential (see., e.g., \cite{Steinbach2010}). In particular, it is the Bessel potential, a comprehensive overview of which can be found in \cite{Aronszajn1961,Adams1967,Adams1969}.

It was shown in \cite[Thm.~II.6]{Guerra1976} that the inverse of $\hat{G}_0$ is a self-adjoint extension of $(- \Delta + m^2 )|_{C_0^\infty (\Omega)}$. We now derive the boundary conditions that are satisfied by functions in the domain of this operator. This question is treated in \cite[Sec.~II.2]{Guerra1976} (see also \cite{Kalmenov2009}). Fix $f \in \mathfrak{D} (\hat{G}_0^{-1}) \cap C^2 (\Omega) \cap C^1 (\overline{\Omega})$. Using Green's third identity \cite[Ch.~6]{McLean2000}, which can be derived from Green's second identity \cite[Ch.~4]{McLean2000} together with \eqref{eq:DG=1}, we have, for all $\boldsymbol{x} \in \Omega$,
\begin{align*}
    f (\boldsymbol{x}) &= ( \hat{G}_0 \hat{G}_0^{-1} f ) (\boldsymbol{x}) = \int_\Omega G (\boldsymbol{x},\boldsymbol{y};m) \, (- \Delta_y + m^2) f (\boldsymbol{y}) \; \D^d y \\
    &= f (\boldsymbol{x}) + \int_{\partial \Omega} \left[ \frac{\partial G (\boldsymbol{x},\boldsymbol{y};m)}{\partial n_y} f (\boldsymbol{y}) - \frac{\partial f (\boldsymbol{y})}{\partial n_y} G (\boldsymbol{x},\boldsymbol{y};m) \right] \D S (\boldsymbol{y}) \; , \numberthis
\end{align*}
which implies
\begin{equation}\label{eq:nonlocal_bcs_deriv1}
    \int_{\partial \Omega} \frac{\partial G (\boldsymbol{x},\boldsymbol{y};m)}{\partial n_y} f (\boldsymbol{y}) \; \D S (\boldsymbol{y}) = \int_{\partial \Omega} \frac{\partial f (\boldsymbol{y})}{\partial n_y} G (\boldsymbol{x},\boldsymbol{y};m) \; \D S (\boldsymbol{y}) \; , \qquad \boldsymbol{x} \in \Omega \; .
\end{equation}
It can be shown\footnote{The reader may also compare this with the theory of surface layer potentials, see, e.g., \cite{McLean2000,Steinbach2010,Hsiao2021}.} \cite[Sec.~II.2]{Guerra1976} that in the limit $\Omega \ni \boldsymbol{x}' \to \boldsymbol{x} \in \partial \Omega$ the left-hand side of the above equation can be written as
\begin{equation}
    \int_{\partial \Omega} \frac{\partial G (\boldsymbol{x},\boldsymbol{y};m)}{\partial n_y} f (\boldsymbol{y}) \; \D S (\boldsymbol{y}) = \iint_{\partial \Omega} G (\boldsymbol{x},\boldsymbol{y};m) k (\boldsymbol{y},\boldsymbol{z};m) f (\boldsymbol{z}) \; \D S (\boldsymbol{y}) \, \D S (\boldsymbol{z})
\end{equation}
for all $\boldsymbol{x} \in \partial \Omega$. Here, the kernel $k$ is given by
\begin{equation}
    k (\boldsymbol{x},\boldsymbol{y};m) = \frac{\partial^2}{\partial n_x \partial n_y} G_\mathrm{D}^\mathrm{ext.} (\boldsymbol{x},\boldsymbol{y};m) \; ,
\end{equation}
where $G_\mathrm{D}^\mathrm{ext.}$ is the Green's function of the Dirichlet problem in the exterior domain $\Omega^\mathrm{ext.} = \R^d \setminus \overline{\Omega}$. The condition \eqref{eq:nonlocal_bcs_deriv1} thus yields the boundary condition
\begin{equation}\label{eq:nonlocal_bcs}
    \frac{\partial f}{\partial n} (\boldsymbol{x}) = \int_{\partial \Omega} k (\boldsymbol{x},\boldsymbol{y};m) f (\boldsymbol{y}) \; \D S (\boldsymbol{y}) \; , \qquad \boldsymbol{x} \in \partial \Omega \; .
\end{equation}

Following \cite{Guerra1976}, we call these boundary conditions \emph{free} boundary conditions. We make two important observations. First, the boundary conditions in \eqref{eq:nonlocal_bcs} are non-local in the sense that they cannot be written in the form \eqref{eq:local_kernel_b}, i.e., the normal derivative at a point on the boundary depends on the value of the function at each point of the boundary. As we will discuss in Section \ref{sec:mutual_info}, this is necessary for two disjoint open regions $\Omega_A$ and $\Omega_B$ to ``communicate'' with each other. Secondly, the boundary condition (more precisely, the self-adjoint extension) depends on the mass as indicated by the mass dependence of the kernel $k(.,.;m)$ in \eqref{eq:nonlocal_bcs}.

Similar to the case of classical boundary conditions discussed in Section \ref{sec:classical_bcs}, we can associate a bilinear form (and thus an action) to the free boundary conditions. We define the free form $\mathfrak{q}_\mathrm{F}$ on $\mathcal{Q} (\mathfrak{q}_\mathrm{F}) = \mathfrak{D} (\hat{G}_0^{-1/2})$ by $\mathfrak{q}_\mathrm{F} (f,g) = \braket{\hat{G}_0^{-1/2} f, \hat{G}_0^{-1/2} g}_{L^2 (\Omega)}$. Clearly, $\mathfrak{q}_\mathrm{F}$ is symmetric, densely defined and positive. As $\hat{G}_0^{-1/2}$ is self-adjoint, it is closed and the same is true for $\mathfrak{q}_\mathrm{F}$ \cite[Ex.~VI.1.13]{Kato1995}. For any $f,g \in C^\infty (\overline{\Omega})$, the free form is given by
\begin{equation}\label{eq:free_form}
    \mathfrak{q}_\mathrm{F} (f,g) = \int_\Omega \left( \boldsymbol{\nabla} f \cdot \boldsymbol{\nabla} g + m^2 f g \right) \; \D^d x - \iint_{\partial \Omega} k (\boldsymbol{x}, \boldsymbol{y}; m) f (\boldsymbol{x}) g (\boldsymbol{y}) \; \D S (\boldsymbol{x}) \, \D S (\boldsymbol{y}) \; .
\end{equation}
Free boundary conditions thus correspond to the choice $b = - k$ in \eqref{eq:Eucl_action_bilinear}.

It is encouraging to see that all boundary conditions we investigate fall into the class of theories described by \eqref{eq:boundary_condition_bkernel} for a suitable choice of the kernel $b$. We conjecture that a very large class of physically interesting states in the region $\Omega$ can be described along these lines with a suitable boundary action $S_\mathrm{E}^{\partial \Omega}$, cf. \eqref{eq:Eucl_action_bcs}.

We finish this Section with a useful operator inequality. Let $T_1$ and $T_2$ be two bounded operators on a separable Hilbert space $\mathcal{H}$. If $\braket{h, T_1 h} \leq \braket{h, T_2 h}$ for all $h \in \mathcal{H}$, we write $T_1 \leq T_2$.
\begin{lemma}[{\cite[Sec.~III.1]{Guerra1976}, see also \cite[Sec.~7.7]{Glimm2012}}]\label{lem:operator_inequalities}
    Let $\Omega$ be a bounded open subset of $\R^d$ with boundary $\partial \Omega$. Then, the following operator inequality holds.
    \begin{equation}
        \hat{G}_\mathrm{D} \leq \hat{G}_0 \leq \hat{G}_\mathrm{N} \; .
    \end{equation}
    If $\Omega$ is a rectangular domain, we have the additional inequality
    \begin{equation}
        \hat{G}_\mathrm{D} \leq \hat{G}_\mathrm{P} \leq \hat{G}_\mathrm{N} \; .
    \end{equation}
\end{lemma}

\section{Information Theory and Relative Entropy}\label{sec:info_theory}

In this Section we introduce the relative entropy and discuss its properties. In particular, we focus on the relative entropy between two Gaussian probability measures, which correspond to free field theories. The concepts of equivalence and mutual singularity of two Gaussian measures play an important role in the discussion of the relative entropy. As we will see, the relative entropy between two Gaussian measures is finite if and only if the two measures are equivalent and it is infinite precisely when the two measures are mutually singular. In contrast to the case of Gaussian measures on $\R^n$, the question of equivalence is much more subtle when we consider Gaussian measures on infinite dimensional spaces. In Appendix \ref{app:equivalence}, we derive necessary and sufficient conditions for the equivalence of two Gaussian measures on infinite dimensional spaces describing free scalar field theories over a bounded subset $\Omega \subset \R^d$. We conclude this Section with a discussion of a special type of relative entropy, the mutual information.

\subsection{Equivalence and Mutual Singularity of Measures}

We start this Section by introducing the concepts of absolute continuity, equivalence and mutual singularity of probability measures. For a basic introduction to (Gaussian) measure theory, we refer the reader to Appendix \ref{app:func_integrals} and references therein.
\begin{definition}[{\cite[Def.~3.2.1]{Bogachev2007}}]
    Let $\mu$ and $\nu$ be two measures on a measurable space $(X,\mathcal{A})$.
    \begin{itemize}
        \item The measure $\nu$ is called absolutely continuous with respect to $\mu$, written $\nu \ll \mu$, if $\nu (A) = 0$ for every $A \in \mathcal{A}$ with $\mu (A) = 0$. If $\nu \ll \mu$ and $\mu \ll \nu$, then the measures $\mu$ and $\nu$ are called equivalent and we write $\mu \sim \nu$.
        \item The measure $\nu$ is called singular with respect to $\mu$, written $\nu \perp \mu$, if there exists a set $A \in \mathcal{A}$ such that $\mu (A) = 0$ and $\nu (X \setminus A) = 0$. One can show that $\nu \perp \mu$ implies $\mu \perp \nu$. Therefore, we call two such measures mutually singular.
    \end{itemize}
\end{definition}

Assume for a moment that $\mu$ and $\nu$ are probability measures. The null sets of $\mu$ and $\nu$, i.e., the sets of $\mu$- or $\nu$-measure zero, are those events that are impossible with respect to $\mu$ and $\nu$, respectively. In other words, if $S$ is a $\mu$-null set, then the probability of the event that the random variable $X$ whose law is given by $\mu$ is realized as an element in $S$ is zero. The probability measure $\nu$ is absolutely continuous with respect to the probability measure $\mu$ precisely when the family of $\mu$-null sets is a subset of the family of $\nu$-null sets. In particular, two probability measures are equivalent if and only if they have the same null sets, i.e., the same impossible events.

On the other hand, if $\mu$ and $\nu$ are mutually singular, then there exists an event that is impossible with respect to $\mu$ but certain with respect to $\nu$. Of course, the complementary event is certain with respect to $\mu$ but impossible with respect to $\nu$, reflecting the fact that $\nu \perp \mu$ implies $\mu \perp \nu$.

For the case where $\mu$ and $\nu$ are Gaussian measures, a Theorem by Feldman and H\'ajek (Theorem \ref{th:FeldmanHajek}) states that $\mu$ and $\nu$ are either equivalent or mutually singular. Recall that a Gaussian measure on $\R$ is called non-degenerate when the standard deviation $\sigma \in \R$ is positive and is called degenerate when it is the Dirac measure, i.e., when it assigns unit probability to its mean $a \in \R$, which corresponds formally to the choice $\sigma = 0$. Let $\mu$ and $\nu$ be centred Gaussian measures on $\R$ (i.e., Gaussian measures with mean zero) such that $\mu$ is non-degenerate and $\nu$ is degenerate. Then, the event that the random variable takes the value zero is certain with respect to $\nu$ but impossible with respect to $\mu$ as the set $\{0\}$ has Lebesgue measure zero. Similarly, the probability that the random variable takes any value in $\R \setminus \{0\}$ is certain with respect to $\mu$ but impossible with respect to the centred Dirac measure $\nu$. We conclude that $\mu \perp \nu$.

Let $\mu$ and $\nu$ again be (not necessarily centred) Gaussian measures on $\R$, but this time assume that they are both non-degenerate. The families of $\mu$- and $\nu$-null sets coincide in this case as they are both given by the subsets of $\R$ with Lebesgue measure zero. We therefore conclude that $\mu \sim \nu$. It is easy to see that this result generalizes to higher dimensions. In particular, any two non-degenerate\footnote{Recall that a Gaussian measure on $\R^n$ is called non-degenerate if the covariance matrix is full rank and degenerate otherwise. In particular, a Gaussian measure on $\R^n$ is degenerate precisely when it is supported on a lower-dimensional subspace of $\R^n$. Intuitively speaking, a degenerate Gaussian measure on $\R^n$ has vanishing variance in the eigendirections of the covariance matrix spanning the covariance matrix's kernel.} Gaussian measures on $\R^n$ are equivalent.

At this point, one may think that two non-degenerate\footnote{The notion of (non-) degeneracy of Gaussian measures on infinite dimensional spaces generalises from the corresponding notion of Gaussian measures on $\R^n$. In particular, a Gaussian measure on an infinite dimensional space is called non-degenerate precisely when its covariance form, see Appendix \ref{app:func_integrals} for definitions, is strictly positive \cite{Eldredge2016}, which implies that the measure is supported everywhere \cite[Sec.~3.6]{Bogachev2015}. For the Gaussian measures on infinite dimensional spaces considered in this work, this amounts to the covariance operator $\hat{C}$ being strictly positive.} Gaussian measures on \emph{infinite} dimensional spaces (which are needed for the description of non-interacting field theories in the continuum) are also always equivalent. However, as discussed in Appendix \ref{app:equivalence}, two non-degenerate Gaussian measures on infinite dimensional spaces may be mutually singular. In fact, it turns out that equivalence of Gaussian measures on infinite dimensional spaces is a rather strong requirement and establishing equivalence can be subtle. For example, a simple scaling of the covariance operator leads to mutually singular Gaussian measures: Two centred non-degenerate Gaussian measures $\mu_1$ and $\mu_2$ on a separable Hilbert space $\mathcal{H}$ with covariance operators $\hat{C}_1$ and $\hat{C}_2 = \gamma \hat{C}_1$, $\gamma > 1$, respectively, are mutually singular\footnote{To see this, take \cite[Thm.~2.25]{DaPrato2014} and notice that in this case $(\hat{C}_1^{-1/2} \hat{C}_2^{1/2})(\hat{C}_1^{-1/2} \hat{C}_2^{1/2})^* - I$ is proportional to the identity and therefore not a Hilbert-Schmidt operator on $\mathcal{H}$.}. This qualitative difference of Gaussian measures on $\R^n$ and infinite dimensional spaces is not only interesting from a purely mathematical point of view but has also important consequences for the study of the relative entropy between field theories. As we will argue in the next Section, the relative entropy between two Gaussian measures is finite if and only if the two Gaussian measures are equivalent. Thus, when one studies the relative entropy in continuum field theory, a key step is to prove (or disprove) the equivalence of the underlying probability measures. In Section \ref{sec:applications}, where we consider concrete examples of relative entropy in field theories, we address precisely this question. The necessary and sufficient conditions for equivalence of the Gaussian measures used in Section \ref{sec:applications} are derived in Appendix \ref{app:equivalence}.

\subsection{Relative Entropy}\label{sec:relative_entropy}

In this Section we introduce the relative entropy between two probability measures. Besides its properties as a divergence, we will also discuss its statistical interpretation. In particular, we argue that the relative entropy can be thought of as a measure of the distinguishability of two probability measures. From a mathematical point of view, the relative entropy allows for a meaningful generalization of the concept of an entropy to probability measures on infinite dimensional spaces. We start this Section with a short reminder of basic information theory.

Let $X$ be a random variable taking values in a finite set $\mathcal{X}$. Its probability mass function is defined by $p_X (x) = \mathbb{P} (X=x)$ for all $x \in \mathcal{X}$. The information content (or surprisal) of a realization $x \in \mathcal{X}$ of the random variable $X$ is defined as \cite{Borda2011}
\begin{equation}
    i_X (x) \coloneqq - \log p_X (x) = - \log \mathbb{P} (X=x) \; .
\end{equation}
The information content quantifies the amount of information we gain from observing a particular realization (or ``outcome'') $x \in \mathcal{X}$ of the random variable $X$. The average information content is of central importance in information theory and is called Shannon's entropy \cite{Shannon1948}. It is defined as
\begin{equation}
    S (X) \coloneqq \EE [i_X] = - \sum_{x \in \mathcal{X}} p_X (x) \log p_X (x) \; .
\end{equation}

Shannon's entropy is defined for probability distributions on finite or countable sets. If one considers na\"ive generalizations of Shannon's entropy to probability measures on other sets, important properties of Shannon's entropy are typically lost. For example, the differential entropy of probability distributions on $\R^n$ is neither non-negative nor invariant under a change of variables. An additional problem arises in the case of probability measures on infinite dimensional spaces, describing field theories in the continuum, where it is not clear how to make sense of the ``functional'' entropy
\begin{equation}\label{eq:func_entropy}
    S [p] = - \int p [\varphi] \log p [\varphi] \; \mathcal{D} \varphi \; .
\end{equation}
In practice, such a functional entropy generically suffers from UV-divergences, as can be seen from a simple one-loop calculation.

Following the ideas of Kullback and Leibler \cite{Kullback1951,Kullback1997}, one can extend the concept of entropy to a general measure theoretic framework\footnote{See also \cite{Masani1992}.}. Consider two probability spaces $(\Omega, \mathcal{A}, \mu_i)$, $i=1,2$. Furthermore, we shall assume that $\mu_1 \ll \mu_2$. Then, by the Radon-Nikodym Theorem \cite[Thm.~3.2.2]{Bogachev2007}, there exists a non-negative, real-valued, $\mu_2$-integrable function $f$ such that, for all measurable sets $A \subseteq \Omega$,
\begin{equation}
    \mu_1 (A) = \int_A f \; \D \mu_2 \; .
\end{equation}
The function $f$ is called the Radon-Nikodym derivative (or generalized density) of $\mu_1$ with respect to $\mu_2$ and is usually denoted by $f = \D \mu_1 / \D \mu_2$.

\begin{definition}[\cite{Kullback1951}]
    Let $(X,\mathcal{A},\mu_i)$, $i=1,2$, be two probability spaces and suppose that $\mu_1 \ll \mu_2$. The relative entropy (or Kullback–Leibler divergence) between $\mu_1$ and $\mu_2$ is defined by
    \begin{equation}
        D_\mathrm{KL} (\mu_1 \| \mu_2) = \int_X \frac{\D \mu_1}{\D \mu_2} (x) \, \log \left[ \frac{\D \mu_1}{\D \mu_2} (x) \right] \, \D \mu_2 (x) = \int_X \log \left[ \frac{\D \mu_1}{\D \mu_2} (x) \right] \, \D \mu_1 (x) \; .
    \end{equation}
\end{definition}
The relative entropy is non-negative and vanishes precisely when $\mu_1 = \mu_2$ \cite[Thm.~3.1]{Kullback1997}. It is, however, not a metric as it is not symmetric and, more importantly, does not satisfy the triangle inequality. Rather, it is an example of a (directed) statistical divergence, a concept that is central in the field of information geometry \cite{Amari2000,Amari2016,Ay2017}. The relative entropy is invariant under parameter transformations and additive for independent random variables (cf. \cite[Ch.~2]{Kullback1997}), as one would expect from an entropy. Finally, it satisfies a monotonicity property in the form of the data processing inequality \cite[Thm.~9]{VanErven2014}.

As already noted above, if $\mu_1$ and $\mu_2$ are Gaussian measures, then they are either equivalent or mutually singular. They are equivalent if and only if $\mu_1 \ll \mu_2$ and $\D \mu_1 / \D \mu_2 > 0$ almost everywhere with respect to $\mu_2$ \cite[Sec.~3.2]{Bogachev2007}. In this sense, the Radon-Nikodym derivative between two equivalent Gaussian measures is supported everywhere. If $\mu_1$ and $\mu_2$ are equivalent Gaussian measures, then the Radon-Nikodym derivative is given by $\D \mu_1 / \D \mu_2 = \exp F$, where $F$ is a second order polynomial \cite[Cor.~6.4.10]{Bogachev2015}. In particular, this means that $\log [ \D \mu_1 / \D \mu_2 ]$ is integrable with respect to $\mu_1$ and the relative entropy is finite. If $\mu_1$ and $\mu_2$ are mutually singular, the relative entropy is usually defined to be $+ \infty$. This definition can be motivated from the case of two probability distributions on finite sets whose support does not coincide and where one usually employs the definition $c \log \frac{c}{0} \coloneqq + \infty$. Therefore, for the case of two Gaussian measures $\mu_1$ and $\mu_2$ on a common measurable space $(X, \mathcal{A})$, the relative entropy is defined as
\begin{equation}
    D_\mathrm{KL} (\mu_1 \| \mu_2) = \begin{cases} 
        \begin{aligned}
            \int_X \frac{\D \mu_1}{\D \mu_2} (x) \log \, &\frac{\D \mu_1}{\D \mu_2} (x) \; \D \mu_2 (x) && \quad \text{ if } \; \mu_1 \sim \mu_2 \\[.9mm]
            + &\infty && \quad \text{ if } \; \mu_1 \perp \mu_2
        \end{aligned}
    \end{cases} .
\end{equation}

One often interprets the second argument of the relative entropy as a model (or reference, or approximation) for the true theory given in the first argument. In particular, the relative entropy $D_\mathrm{KL} (\mu_1 \| \mu_2)$ quantifies the average excess information content from using $\mu_2$ as a model when the true theory is given by $\mu_1$ \cite{Cover2005}. With this interpretation, the definition $D_\mathrm{KL} (\mu_1 \| \mu_2) = + \infty$ if $\mu_1 \perp \mu_2$ can be understood as the statement that the observation of an event that is impossible with respect to the model gives us infinite information about whether the model is equal to the true distribution, namely that it is not. In the following, we usually interpret the relative entropy as an information-theoretic measure of the distinguishability of two theories described by probability measures. The relative entropy is zero (and hence the theories are indistinguishable) precisely when $\mu_1 = \mu_2$. For $\mu_1 \neq \mu_2$, the distinguishability is positive and increases monotonically as the two theories become ``more different''\footnote{As an example, consider Section \ref{sec:diff_masses}, where the relative entropy increases as the difference in the masses increases.}. In the extreme case where the two theories are mutually singular, i.e., when the model predicts zero probability for an event that is certain with respect to the true distribution, we would intuitively say that the two theories can be perfectly distinguished from one another. This gives another motivation for the definition $D_\mathrm{KL} (\mu_1 \| \mu_2) = + \infty$ if $\mu_1 \perp \mu_2$.

\subsection{Mutual Information}

We now introduce a special kind of relative entropy, the mutual information. In Section \ref{sec:mutual_info}, we study the mutual information between two disjoint regions in Euclidean space, which we argue is the quantity that is closest to the notion of a relative entropy of entanglement as defined in \cite[Sec.~3.4]{Hollands2018} for relativistic quantum field theories, in classical statistical field theory. In the following, we discuss the properties of the mutual information for probability distributions over finite sets as well as multi-variate normal distributions.

The mutual information between two random variables $A$ and $B$ is defined as \cite{Shannon1948,Gelfand1959,Chiang1959}
\begin{align*}
    I (A:B) &\coloneqq D_\mathrm{KL} (\mu_{AB} \| \mu_A \otimes \mu_B) \\
    &= \int_{X \times Y} \frac{\D \mu_{AB}}{\D \mu_A \D \mu_B} (x,y) \; \log \left[ \frac{\D \mu_{AB}}{\D \mu_A \D \mu_B} (x,y) \right] \, \D \mu_A (x) \, \D \mu_B (y) \; , \label{eq:MI_def} \numberthis
\end{align*}
where $\mu_{AB}$ is a joint probability measure and $\mu_A$ and $\mu_B$ are marginals, and we assumed that $\mu_{AB} \ll \mu_A \otimes \mu_B$. For the particular case where we consider probability distributions on a finite set, we can write the mutual information as
\begin{align*}
    I (A : B) &= \sum_{i,j} p_{AB} (x_i, y_j) \, \log \frac{p_{AB} (x_i, y_j)}{p_A (x_i) \, p_B (y_j)} \\
    &= S (A) + S (B) - S (AB) \; , \label{eq:MI_finite_set} \numberthis
\end{align*}
where $p_A$ and $p_B$ are the marginal distributions given by
\begin{equation}
    p_A (x_i) = \sum_j p_{AB} (x_i, y_j) \; , \qquad p_B (y_j) = \sum_i p_{AB} (x_i, y_j) \; ,
\end{equation}
respectively. The second line in \eqref{eq:MI_finite_set} provides an interpretation of the mutual information in terms of information content: The mutual information is the average amount of information that is shared by (or is common to both) the random variables $A$ and $B$. Alternatively, the mutual information can be interpreted as the average amount of information that is gained about $A$ by observing $B$ and vice versa. In particular, the mutual information is zero precisely when $A$ and $B$ are independent, i.e., when $p_{AB} (x_i, y_j) = p_A (x_i) \, p_B (y_j)$, in which case $A$ and $B$ share no information. Thus, the mutual information quantifies how (in-)dependent two random variables are from each other.

It is instructive to consider the case where $A$ and $B$ are centred random variables following a Gaussian distribution over $\R^n$ and $\R^m$, respectively. Then, the joint distribution $p_{AB}$ is a multi-variate normal distribution with $(n+m) \times (n+m)$-dimensional covariance matrix
\begin{equation}
    \Sigma = \begin{pmatrix}
        \Sigma_A & \Sigma_{AB} \\
        \Sigma^\mathsf{T}_{AB} & \Sigma_B
    \end{pmatrix} \; ,
\end{equation}
where $\Sigma_{AB}$ is a $(n \times m)$-matrix, called the cross-covariance matrix of the joint distribution $p_{AB}$. The covariance matrices of the marginal distributions $p_A$ and $p_B$ are given by $\Sigma_A$ and $\Sigma_B$, respectively, and the covariance matrix of the product distribution $p_A \cdot p_B$ is given by $\Sigma_A \oplus \Sigma_B$. Using the definition of the mutual information in \eqref{eq:MI_def}, we see that the mutual information is the relative entropy between the ``true'' distribution containing all cross-correlations of the random variables $A$ and $B$ and the ``model'' which coincides with the true distribution except that it contains no cross-correlations. Using the expression for the relative entropy between multivariate Gaussian distributions, \eqref{eq:KL_multivariate_Gaussian}, the mutual information can in this case be written as
\begin{equation}
    I (A : B) = \frac{1}{2} \log \left[ \frac{\det \Sigma_A \, \det \Sigma_B}{\det \Sigma} \right] \; .
\end{equation}

\section{Applications of Relative Entropy in Statistical Field Theory}\label{sec:applications}

In this Section, we apply the results of Sections \ref{sec:func_integrals} and \ref{sec:info_theory} to statistical field theory. In Section \ref{sec:diff_masses}, we consider two free scalar field theories on a bounded domain with different mass parameters $m_1$ and $m_2$ but the same classical boundary conditions. We show that the equivalence of two such field theories depends non-trivially on $d$, the dimension of Euclidean space. For Dirichlet, Neumann and periodic boundary conditions on a cubic region, we then calculate the relative entropy of theories with different masses.

In Section \ref{sec:diff_bcs}, we discuss the properties of the relative entropy between two field theories on a bounded domain with different boundary conditions. While, in general, two field theories with different boundary conditions can be mutually singular in all Euclidean spacetime dimensions, we show that there exists an upper critical dimension for the relative entropy between two Robin field theories. We then calculate the relative entropy between two such theories in one Euclidean spacetime dimension.

Finally, in Section \ref{sec:mutual_info}, we discuss a special type of relative entropy, the mutual information. We show that a sufficient condition for the finiteness of the mutual information between two disjoint open regions $\Omega_A$ and $\Omega_B$ is that these regions are separated by a finite distance. We then give an explicit example where the regions ``touch'' each other and in which case the mutual information becomes infinite. Finally, we argue (without providing a rigorous proof) that the mutual information satisfies an area law.

Some findings presented in this Section, especially regarding the mutual information, depend heavily on results obtained in \cite{Guerra1975a,Guerra1975b,Guerra1976}. To make the exposition in this Section somewhat self-contained, we will restate some of these results. We emphasize, however, that a much more detailed discussion of these topics can be found in the aforementioned references.

\subsection{Field Theories with Different Masses}\label{sec:diff_masses}

In this Section, we study the relative entropy between two free scalar field theories with different masses on a bounded region of space $\Omega$. We derive conditions for the equivalence of two such field theories. Recall from Section \ref{sec:info_theory} that we call two free scalar field theories equivalent if the Gaussian measures describing them are equivalent, i.e., if they have the same null sets. We then discuss the properties of the relative entropy between two equivalent field theories.

Let $\Omega$ be some bounded open subset of $\R^d$ with boundary $\partial \Omega$. Consider centred Gaussian measures of the form $\mu_i = \mathcal{N} (0,\hat{G}_i)$ (for definitions and notation, see Appendix \ref{app:func_integrals}), where $\hat{G}_i = (- \Delta_\mathrm{X} + m_i^2 )^{-1}$, $m_i > 0$, is an operator on $L^2 (\Omega)$. Here, $- \Delta_\mathrm{X}$ is some self-adjoint extension of $- \Delta |_{C_0^\infty (\Omega)}$. As discussed in  Section \ref{sec:covariance_operators}, typical choices for this extension are the Dirichlet Laplacian, the Neumann Laplacian or the Laplacian describing periodic boundary conditions. For two such field theories, it is easy to derive a necessary and sufficient condition for equivalence in terms of the covariance operators.

\begin{lemma}\label{lem:equivalence_mass_independent}
    Let $m_1, m_2 > 0$ be such that $m_1 \neq m_2$. The centred Gaussian measures $\mu_1 = \mathcal{N} (0,\hat{G}_1)$ and $\mu_2 = \mathcal{N} (0,\hat{G}_2)$, where $\hat{G}_i = (- \Delta_\mathrm{X} + m_i^2 )^{-1}$, $i \in \{ 1,2 \}$, are equivalent if and only if $\hat{G}_2$ (and hence also $\hat{G}_1$) is a Hilbert-Schmidt operator on $L^2 (\Omega)$.
\end{lemma}
\begin{proof}
    Obviously $\hat{G}_1^{-1}$ and $\hat{G}_2^{-1}$ have the same form domain, namely that of $- \Delta_\mathrm{X}$. As $\hat{G}_1^{-1}$ and $\hat{G}_2^{-1}$ are bounded from below by a positive number, by Theorem \ref{thm:second_rep_thm}, the ranges of the square roots of their inverses thus coincide, i.e., $\hat{G}_1^{1/2} [L^2 (\Omega)] = \hat{G}_2^{1/2} [L^2 (\Omega)]$. Let $\{\lambda_n\}_{n=1}^\infty$ be the sequence of eigenvalues of $- \Delta_\mathrm{X}$, enumerated in non-decreasing order, and define $\hat{B} \coloneqq \hat{G}_1^{-1/2} \hat{G}_2^{1/2}$. Then,
    \begin{align*}
        \| \hat{B} \hat{B}^* - I \|_{\mathrm{HS}}^2 &= \sum_{n=1}^\infty \left( \frac{\lambda_n + m_1^2}{\lambda_n + m_2^2} - 1 \right)^2 \\
        &= \sum_{n=1}^\infty \left( \frac{m_1^2 - m_2^2}{\lambda_n + m_2^2} \right)^2 \\
        &= (m_1^2 - m_2^2)^2 \, \| \hat{G}_2 \|_{\mathrm{HS}}^2 \; .
    \end{align*}
    Therefore, by Theorem \ref{thm:equivalence}, $\mu_1 \sim \mu_2$ precisely when $\hat{G}_2 \in \mathrm{HS} (L^2 (\Omega))$. 
\end{proof}

If the covariance operators $\hat{G}_1$ and $\hat{G}_2$ are Hilbert-Schmidt, by \eqref{eq:KL_eigenvalues}, the relative entropy between $\mu_1$ and $\mu_2$ takes the simple form
\begin{equation}\label{eq:DKL_series}
    D_\mathrm{KL} (\mu_1 \| \mu_2) = \frac{1}{2} \sum_{n=1}^\infty \left[ \frac{m_1^2 - m_2^2}{\lambda_n + m_2^2} - \log \left( \frac{m_1^2 - m_2^2}{\lambda_n + m_2^2} + 1 \right) \right] \; .
\end{equation}
We can see that the above series is indeed convergent precisely when $\hat{G}_2$ is Hilbert-Schmidt. More specifically, define $x_n \coloneqq (m_1^2 - m_2^2) / (\lambda_n + m_2^2)$. As $n \to \infty$, $x_n \to 0$ and we can expand the logarithm around $x_n = 0$ for large $n$, i.e., $\log (x_n + 1) = x_n - \frac{1}{2} x_n^2 + \mathcal{O} (x_n^3)$. So for large $n$, the summands in the above series behave like $\frac{1}{2} x_n^2 + \mathcal{O} (x_n^3)$. We see that the additional term from the regularized Fredholm determinant (cf. \eqref{eq:reg_fredholm_det}) exactly cancels the problematic term $x_n$, which would lead to a convergent series if and only if $\hat{G}_2$ was also of trace class. Below, when we consider the infinite volume relative entropy density, it will become clear that this additional term from the regularized Fredholm determinant plays a role similar to a mass counterterm.

We do not yet know when the covariance operators are Hilbert-Schmidt. Recall that the covariance operators considered here are the inverses of the differential operator $- \Delta + m^2$. Thus, we expect the asymptotic behaviour of the eigenvalues of $\hat{G}_1$ and $\hat{G}_2$ to depend on $d$, the dimension of Euclidean space. For the case where $- \Delta_\mathrm{X}$ is the Dirichlet Laplacian $- \Delta_\mathrm{D}$, this is made precise by Weyl's law, see \cite{Weyl1912,Weyl1950}, \cite[Sec.~XIII.15]{Reed1978} and \cite[Sec.~VI.4]{Courant2008}. More precisely, let $\Omega \subset \R^d$ be a bounded domain with piecewise smooth boundary. Let $\{ \lambda_n \}_{n=1}^\infty$ be the sequence of eigenvalues of the Dirichlet Laplacian $- \Delta_\mathrm{D}$, enumerated in non-decreasing order. Then, as $n$ tends to infinity, the eigenvalues satisfy the asymptotic behaviour
\begin{equation}
    \lambda_n \sim \mathrm{const.} \times n^{2/d} \; ,
\end{equation}
where $\sim$ denotes asymptotic equivalence. Weyl's law generalizes to Neumann, Robin and periodic boundary conditions if the boundary of the region $\Omega$ is sufficiently regular \cite{Arendt2009}, as well as to the case of free boundary conditions \cite{Widom1963,Dostanic2014}.

Weyl's law implies that the covariance operators $\hat{G}_\mathrm{X}$, $\mathrm{X} \in \{ \mathrm{D}, \mathrm{N}, \mathrm{P}, \sigma \}$, are Hilbert-Schmidt\footnote{In particular, this means that the corresponding Green's function is square integrable \cite[Thm.~VI.23]{Reed1981},
\begin{equation}
    \int_\Omega \int_\Omega \left( G_\mathrm{X} (\boldsymbol{x},\boldsymbol{y}) \right)^2 \, \D^d x \, \D^d y < + \infty \; , \qquad d < 4 \; ,
\end{equation}
where $G_0 = G$. This once again reflects the fact that the field theory becomes more singular in higher dimensions as for $d \geq 4$ the singularity on the diagonal is no longer square integrable.} in Euclidean spacetime dimensions $d < 4$. If $d = 1$, they are also of trace class. For $d \geq 4$, two ``$\mathrm{X}$''-boundary condition field theories with different masses are therefore mutually singular. In particular, for $d = 4$, the Hilbert-Schmidt norm diverges logarithmically. We recall that $d = 4$ is exactly the upper critical dimension of a scalar field theory. It is no coincidence that the relative entropy is infinite at and above the upper critical dimension, as can be seen from the discussion of the relative entropy density at the end of this Section. Using Weyl's law together with Lemma \ref{lem:equivalence_mass_independent}, we arrive at the first main result of this work.

\begin{theorem}\label{thm:main_1}
    Let $\Omega \subset \R^d$ be open, bounded and with piecewise smooth boundary. Furthermore, let $m_1, m_2 > 0$ be such that $m_1 \neq m_2$ and denote by $\mu_1 = \mathcal{N} (0,\hat{G}_1)$ and $\mu_2 = \mathcal{N} (0,\hat{G}_2)$ two centred Gaussian measures, where $\hat{G}_i = (- \Delta_\mathrm{X} + m_i^2 )^{-1}$, $i \in \{ 1,2 \}$ and $\mathrm{X} \in \{ \mathrm{D}, \mathrm{N}, \mathrm{P}, \sigma \}$. The relative entropy $D_\mathrm{KL} (\mu_1 \| \mu_2)$ is finite if and only if $d < 4$.
\end{theorem}

\begin{figure}[t]
    \centering
    \includegraphics[width=0.80\textwidth]{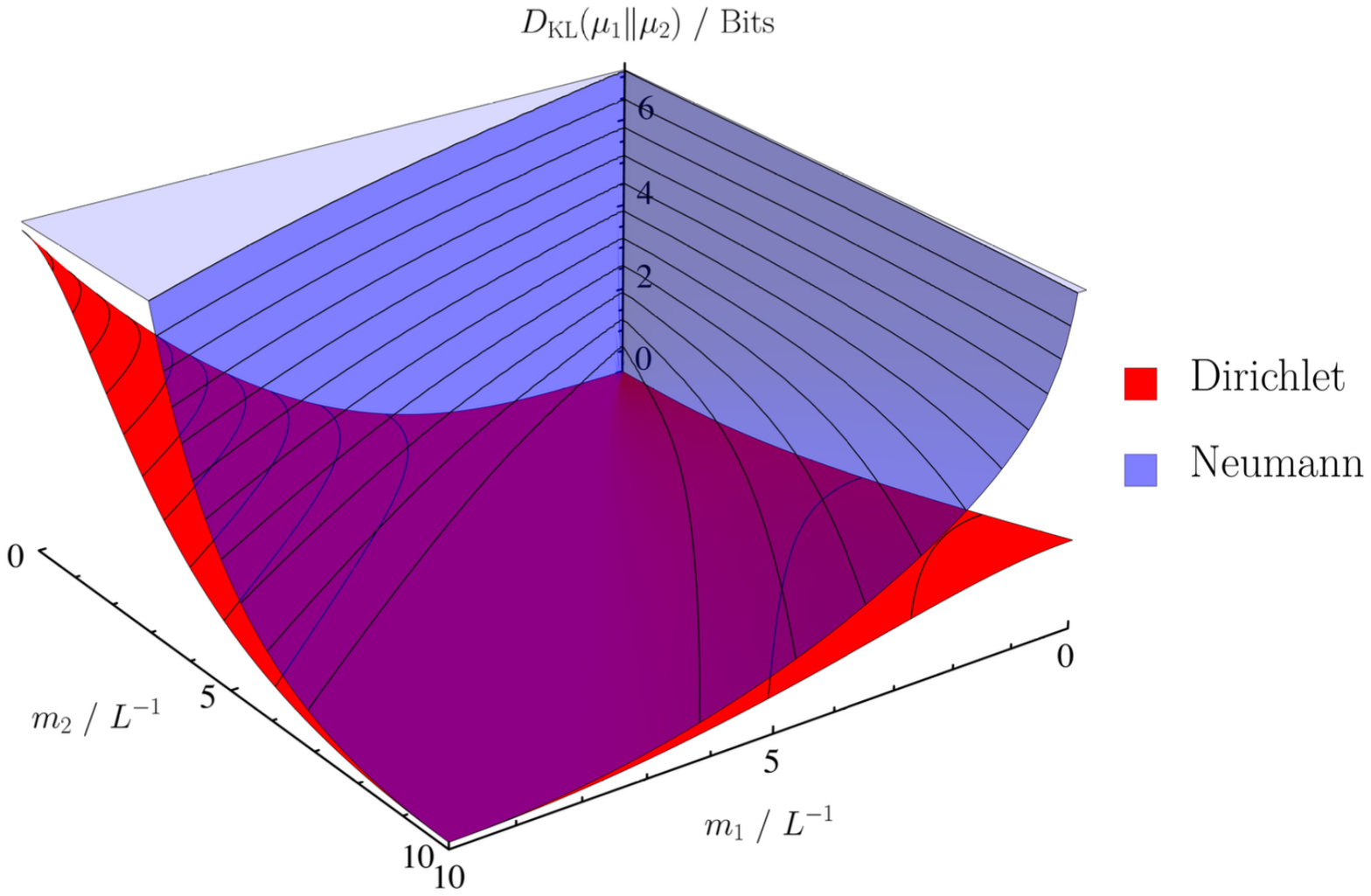}%
    \caption{The relative entropy between two field theories with masses $m_1$ and $m_2$, respectively, on an interval of length $L$. We consider Dirichlet and Neumann boundary conditions and plot the relative entropy in bits against the masses in units of the inverse interval length $L^{-1}$. Note that in the limit $m_i \to 0$ the relative entropy is finite for Dirichlet boundary conditions but diverges for Neumann boundary conditions. This is due to the zero mode of the Laplacian that is present when we choose Neumann boundary conditions.}
    \label{fig:DKL_1d_Plot3D}
\end{figure}

In the following, we restrict ourselves to the case where $\Omega$ is an open $d$-cube of edge length $L$. In this case, the Dirichlet, Neumann and periodic covariance operators ($\hat{G}_\mathrm{D}$, $\hat{G}_\mathrm{N}$ and $\hat{G}_\mathrm{P}$, respectively) are Hilbert-Schmidt precisely when $d < 4$. We start with the case $d = 1$, i.e., we consider a field theory on an interval of length $L$. The eigenvalues of the Dirichlet Laplacian $-\Delta_\mathrm{D}$ on an interval of length $L$ are given by $(n \pi)^2 / L^2$, $n \in \N$. As demonstrated in Appendix \ref{app:1D_Dirichlet_RE}, the relative entropy between two Dirichlet field theories with different masses on an interval of length $L$ admits the following closed-form expression,
\begin{equation}\label{eq:DKL_Dirichlet_1d}
    D^\mathrm{D}_\mathrm{KL} (\mu_1 \| \mu_2) = \frac{1}{4} \left[ \left( 1 - \frac{m_1^2}{m_2^2} \right) + \frac{L (m_1^2 - m_2^2)}{m_2 \tanh (L m_2)} - 2 \log \left( \frac{m_2 \sinh (L m_1)}{m_1 \sinh (L m_2)}  \right) \right] \; .
\end{equation}

We now discuss the properties of this quantity. For a fixed system size $L$, the relative entropy \eqref{eq:DKL_Dirichlet_1d} increases as we increase the absolute value of the mass difference $|m_1 - m_2|$ and is zero precisely when $m_1 = m_2$. This is consistent with our interpretation of the relative entropy as a measure of distinguishability: The greater the difference in masses, the more ``different'' the corresponding field theories are, and the better we can distinguish between them. If the masses are the same, then the theories are the same (remember that we have chosen the same boundary conditions for both fields) and there is no way to distinguish between them. So the relative entropy should be zero in this case.

\begin{figure}[t]
    \centering
    \includegraphics[width=0.475\textwidth]{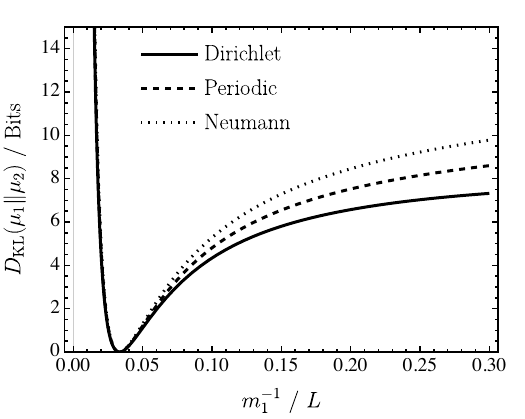}%
    \caption{The relative entropy between two field theories with masses $m_1$ and $m_2 = 30 L^{-1}$, respectively, on an interval of length $L$. We consider Dirichlet, periodic and Neumann boundary conditions and plot the relative entropy in units of bits against the inverse mass (or correlation length) $m_1^{-1}$ in units of the interval length $L$. We observe the ordering $D^\mathrm{D}_\mathrm{KL} \leq D^\mathrm{P}_\mathrm{KL} \leq D^\mathrm{N}_\mathrm{KL}$, with equality only for coinciding masses.}
    \label{fig:DKL_1d_3}
\end{figure}

The expression in \eqref{eq:DKL_Dirichlet_1d} can easily be adapted to other typical boundary conditions. The eigenvalues of the Neumann Laplacian $-\Delta_\mathrm{N}$ on an interval of length $L$ are given by $(n \pi)^2 / L^2$, $n \in \N_0$. Recalling the series representation of the relative entropy \eqref{eq:DKL_series}, the Neumann relative entropy then reads
\begin{equation}\label{eq:DKL_Neumann_1d}
    D^\mathrm{N}_\mathrm{KL} (\mu_1 \| \mu_2) = \frac{1}{2} \left[ \frac{m_1^2}{m_2^2} - \log \left( \frac{m_1^2}{m_2^2} \right) - 1 \right] + D^\mathrm{D}_\mathrm{KL} (\mu_1 \| \mu_2) \; .
\end{equation}
The function $f (x) = x - \log x - 1$ is non-negative for positive $x$ and zero precisely at $x = 1$. Thus, for $m_1 \neq m_2$, the Neumann relative entropy is strictly larger than the Dirichlet relative entropy, cf. Fig. \ref{fig:DKL_1d_Plot3D} and \ref{fig:DKL_1d_3}.

The eigenvalues of the Laplacian $-\Delta_\mathrm{P}$ with periodic boundary conditions on an interval of length $L$ are given by $(2 n \pi)^2 / L^2$, $n \in \N$, each of multiplicity two, together with the smallest eigenvalue $0$ of multiplicity one. The relative entropy for periodic boundary conditions then reads
\begin{equation}\label{eq:DKL_Periodic_1d}
    D^\mathrm{P}_\mathrm{KL} (\mu_1 \| \mu_2) = \frac{1}{2} \left[ \frac{m_1^2}{m_2^2} - \log \left( \frac{m_1^2}{m_2^2} \right) - 1 \right] + 2 D^\mathrm{D}_\mathrm{KL} (\mu_1 \| \mu_2) \bigg|_{L \to L/2} \; .
\end{equation}
Here, $D^\mathrm{D}_\mathrm{KL} (\mu_1 \| \mu_2) |_{L \to L/2}$ denotes the Dirichlet relative entropy given in \eqref{eq:DKL_Dirichlet_1d} but with $L$ replaced by $L/2$. As shown in Fig. \ref{fig:DKL_1d_3}, we observe the ordering $D^\mathrm{D}_\mathrm{KL} \leq D^\mathrm{P}_\mathrm{KL} \leq D^\mathrm{N}_\mathrm{KL}$, where equality holds only for equal masses.

\begin{figure}[t]
    \subfloat[\label{subfig:DKL_1d_a}]{%
        \includegraphics[width=0.475\textwidth]{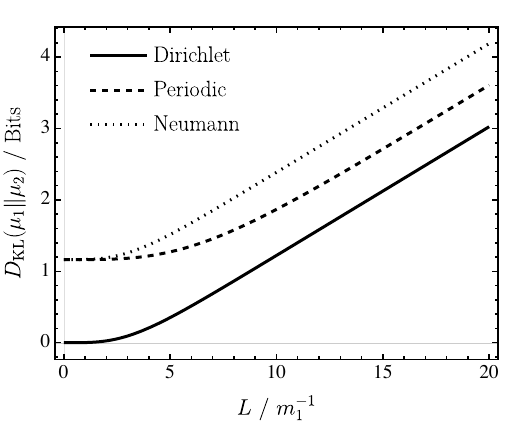}%
    }
    \hfill
    \subfloat[\label{subfig:DKL_1d_b}]{%
        \includegraphics[width=0.475\textwidth]{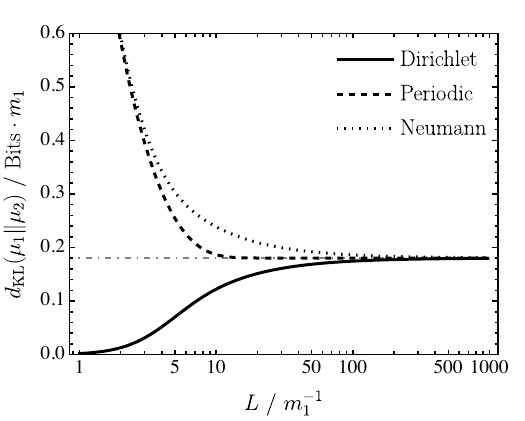}%
    }
    \caption{ \textbf{(a)} The relative entropy between two field theories with masses $m_1$ and $m_2 = \frac{1}{2} m_1$, respectively, on an interval of length $L$ for three different boundary conditions. The relative entropy in units of bits is plotted against the system size $L$ in units of the correlation length or inverse mass $m_1^{-1}$. We see that as soon as the system size is larger than the largest correlation length (in this case, $m_2^{-1} = 2 m_1^{-1}$), the relative entropy scales linearly with $L$. If $L$ is smaller than both correlation lengths, the relative entropy is close to zero for Dirichlet boundary conditions and attains a constant value for Neumann and periodic boundary conditions. \textbf{(b)} The relative entropy density as a function of the system size $L$. We see that in the infinite volume limit $L \to + \infty$, the relative entropy density converges for all three boundary conditions to the limit given in \eqref{eq:infinite_vol_density}, represented in the Figure as a grey dash-dotted line.}
    \label{fig:DKL_1d}
\end{figure}

Instead of changing the masses for some fixed system size, we can also keep the masses fixed and vary $L$. For large $L$, \eqref{eq:DKL_Dirichlet_1d} scales linearly in $L$, which is reminiscent of the extensive behaviour of an entropy. Thus, the distinguishability is large when both length scales set by the inverse masses are small compared to the size of the system. As shown in Fig. \ref{subfig:DKL_1d_a}, this is true for all three boundary conditions considered here. Conversely, if both inverse masses are large compared to the size of the system, distinguishability is low. As can be seen from from \eqref{eq:DKL_Neumann_1d} and \eqref{eq:DKL_Periodic_1d}, as well as from Fig. \ref{subfig:DKL_1d_a}, the additional $L$-independent term from the zero eigenvalue of $- \Delta_\mathrm{N}$ and $- \Delta_\mathrm{P}$ causes the Neumann and periodic relative entropies to attain a constant value in the limit $L \to 0$, while the Dirichlet relative entropy vanishes as the system size approaches zero. This difference in the behaviour is due to the absence of a zero mode of the Dirichlet Laplacian. More precisely, the limit $L \to 0$ corresponds to the limit $m_1 \to 0$ and the value of the relative entropy in this limit is determined by the zero mode of the Laplacian.

We conclude the discussion of the case $d = 1$ by considering a relative entropy density, cf. Fig. \ref{subfig:DKL_1d_b}. We define the relative entropy density  as $d^\mathrm{X}_\mathrm{KL} (\mu_1 \| \mu_2) \coloneqq L^{-1} D^\mathrm{X}_\mathrm{KL} (\mu_1 \| \mu_2)$, where $\mathrm{X} \in \{ \mathrm{D}, \mathrm{P}, \mathrm{N} \}$. In the infinite volume limit, i.e., in the limit $L \to + \infty$, we observe that the relative entropy density converges and the limit is independent of the boundary conditions considered here. More precisely, for all $\mathrm{X} \in \{ \mathrm{D}, \mathrm{P}, \mathrm{N} \}$, the infinite volume limit of the relative entropy density is given by
\begin{equation}\label{eq:infinite_vol_density}
    \lim_{L \to + \infty} d^\mathrm{X}_\mathrm{KL} (\mu_1 \| \mu_2) = \frac{(m_1 - m_2)^2}{4 m_2} \; .
\end{equation}

\begin{figure}[t]
    \subfloat[\label{subfig:DKL_2d_a}]{%
        \includegraphics[width=0.475\textwidth]{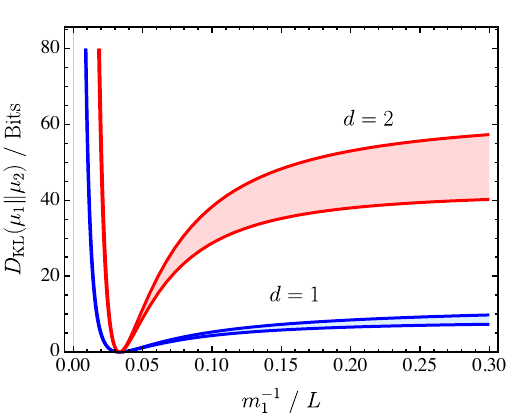}%
    }
    \hfill
    \subfloat[\label{subfig:DKL_2d_b}]{%
        \includegraphics[width=0.505\textwidth]{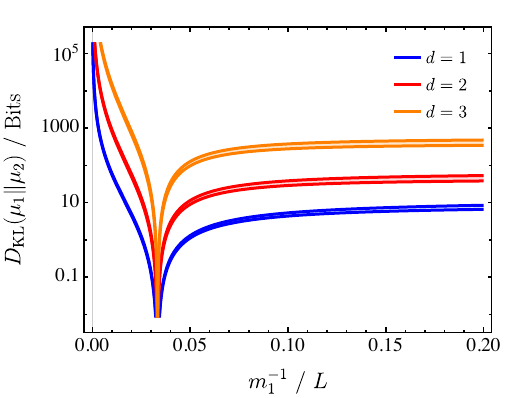}%
    }
    \caption{ \textbf{(a)} The relative entropy between two field theories with masses $m_1$ and $m_2 = 30 L^{-1}$ on a $d$-cube of edge length $L$ in dimensions $d=1$ (blue) and $d=2$ (red). We plot the relative entropy in units of bits against the correlation length or inverse mass $m_1^{-1}$ in units of the edge length $L$. The lower (upper) bound of each shaded region represents the Dirichlet (Neumann) relative entropy. We observe that for a fixed edge length $L$ the relative entropy is larger in higher dimensions. \textbf{(b)} Logarithmic plot of the relative entropy to include the case $d=3$ (orange). We use the same parameters as in (a).}
    \label{fig:DKL_2d}
\end{figure}

We now study field theories with different mass parameters in $d \geq 2$. Recall that we choose $\Omega$ to be an open $d$-cube of edge length $L$. In this case, the eigenvalues of the Laplacians $-\Delta_\mathrm{D}$, $-\Delta_\mathrm{N}$ and $- \Delta_\mathrm{P}$ in $d \geq 2$ can straight-forwardly obtained from the case $d=1$ discussed above. Unlike for $d=1$, we do not compute a closed-form expression for the Dirichlet relative entropy as we did in \eqref{eq:DKL_Dirichlet_1d}. Instead, we approximate the series \eqref{eq:DKL_series} numerically. In Fig. \ref{fig:DKL_2d}, we plot, for a given set of parameters, the relative entropy between two fields with different masses on a cubic region in dimensions $d \leq 3$. We observe that the relative entropy in dimensions $d=2$ and $d=3$ shows the same qualitative behaviour as in $d=1$. In particular, the Neumann relative entropy (top line of shaded region) is strictly greater (except at coinciding masses) than the Dirichlet relative entropy (bottom line of shaded region). Furthermore, we see that for a fixed edge length $L$ the relative entropy is larger in higher dimensions.

We continue with a calculation of the infinite volume relative entropy density. Let $\Omega$ again be the open $d$-cube of edge length $L$ in $\R^d$, $d < 4$. Just as in the case $d = 1$, this limit is independent of the choice of Dirichlet, Neumann or periodic boundary conditions. This follows from the same argument as in the one-dimensional case. As can be seen from \eqref{eq:DKL_Neumann_1d} and \eqref{eq:DKL_Periodic_1d}, the one-dimensional Neumann and periodic relative entropies differ from the Dirichlet relative entropy only by an additional term accounting for the zero mode of the respective Laplacian (and the scaling of the interval length in the periodic case). This additional term is independent of the system size and vanishes in the infinite volume limit when weighted with a factor $L^{-1}$. This argument generalizes to higher dimensions in a straightforward manner and thus establishes that the Dirichlet and Neumann relative entropy densities are equal in the infinite volume limit for $d < 4$. Furthermore, the rescaling $L \to L/2$ in the expression for the periodic relative entropy density causes a rescaling of the integral measure $\D^d p$ (see discussion below) by a factor of $2^{-d}$, which cancels the factor $2^d$ accounting for the multiplicity of the eigenvalues of the periodic Laplacian. We thus find that the infinite volume relative entropy density is independent of the boundary conditions considered here. For simplicity, we will calculate the infinite volume relative entropy density starting from the Dirichlet relative entropy density.

Using \eqref{eq:DKL_series}, we can write the Dirichlet relative entropy between two fields with masses $m_1$ and $m_2$, respectively, over $\Omega$ as
\begin{equation}
    D^\mathrm{D}_\mathrm{KL} (\mu_1 \| \mu_2) = \frac{1}{2} \sum_{\boldsymbol{n} \in \N^d} \left[ \frac{m_1^2 - m_2^2}{\frac{(|\boldsymbol{n}| \pi)^2}{L^2} + m_2^2} - \log \left( \frac{m_1^2 - m_2^2}{\frac{(|\boldsymbol{n}| \pi)^2}{L^2} + m_2^2} + 1 \right) \right] \; ,
\end{equation}
where $\boldsymbol{n} = (n_1, \ldots, n_d)$ and $|\boldsymbol{n}|^2 = n_1^2 + \ldots + n_d^2$. Upon defining $\frac{\Delta k}{2 \pi} \coloneqq L^{-1}$ and replacing the summation over $\N^d$ by a summation over $\Z^d$, the Dirichlet relative entropy density $d^\mathrm{D}_\mathrm{KL} (\mu_1 \| \mu_2) = L^{-d} D^\mathrm{D}_\mathrm{KL} (\mu_1 \| \mu_2)$ reads
\begin{equation}\label{eq:KL_density_Dirichlet_Zd}
    d^\mathrm{D}_\mathrm{KL} (\mu_1 \| \mu_2) = \frac{1}{2^{d+1}} \sum_{\boldsymbol{n} \in \Z^d} \frac{(\Delta k)^d}{(2 \pi)^d} \left[ \frac{m_1^2 - m_2^2}{\frac{(\Delta k |\boldsymbol{n}|)^2}{4} + m_2^2} - \log \left( \frac{m_1^2 - m_2^2}{\frac{(\Delta k |\boldsymbol{n}|)^2}{4} + m_2^2} + 1 \right) \right] + \ldots \; ,
\end{equation}
where the ellipsis denotes terms that take into account summands where one or more $n_i = 0$, i.e., the zero modes along one or more coordinate direction. As discussed above, these terms vanish in the limit $L \to + \infty$, so we will omit them in the remainder of the calculation. Since, as discussed above, the infinite volume limit of the relative entropy density is independent of the choice of boundary conditions, we define $d_\mathrm{KL} (\mu_1 \| \mu_2) \coloneqq \lim_{L \to +\infty} d^\mathrm{D}_\mathrm{KL} (\mu_1 \| \mu_2)$ to be \emph{the} (infinite volume) relative entropy density. Since the function appearing in the sum \eqref{eq:KL_density_Dirichlet_Zd} is Riemann integrable, we can take the limit $L \to + \infty$, which corresponds to $\Delta k \to 0$, and arrive at the convergent improper integral
\begin{equation}\label{eq:KLdensity_continuum}
    d_\mathrm{KL} (\mu_1 \| \mu_2) = \frac{1}{2} \int_{\R^d} \frac{\D^d p}{(2 \pi)^d} \left[ \frac{m_1^2 - m_2^2}{|\boldsymbol{p}|^2 + m_2^2} - \log \left( \frac{m_1^2 - m_2^2}{|\boldsymbol{p}|^2 + m_2^2} + 1 \right) \right] \; ,
\end{equation}
where we made the substitution $p_i \to 2 p_i$ to get rid of the factor $\frac{1}{4}$ in the denominator. Using $d$-dimensional spherical coordinates, the relative entropy density can be written as
\begin{equation}
    d_\mathrm{KL} (\mu_1 \| \mu_2) = \frac{\pi^{-\frac{d}{2}}}{2^d \, \Gamma (d/2)} \int_0^\infty q^{d-1} \left[ \frac{m_1^2 - m_2^2}{q^2 + m_2^2} - \log \left( \frac{m_1^2 - m_2^2}{q^2 + m_2^2} + 1 \right) \right] \D q \; .
\end{equation}
We can evaluate this integral by performing an integration by parts. More precisely, for $d < 4$,
\begin{align*}
    d_\mathrm{KL} (\mu_1 \| \mu_2) &= - \frac{\pi^{-\frac{d}{2}}}{d \, 2^d \, \Gamma (d/2)} \int_0^\infty q^{d} \; \frac{\D}{\D q} \left[ \frac{m_1^2 - m_2^2}{q^2 + m_2^2} - \log \left( \frac{m_1^2 - m_2^2}{q^2 + m_2^2} + 1 \right) \right] \D q \\
    &= \frac{2 \pi^{-\frac{d}{2}}}{d \, 2^d \, \Gamma (d/2)} \int_0^\infty \frac{q^{d+1} (m_1^2 - m_2^2)^2}{(q^2 + m_2^2)^2 \, (q^2 + m_1^2)} \; \D q \\
    &= \frac{\pi^{1-\frac{d}{2}}}{d \, 2^{d+1} \Gamma (d/2)} \, \frac{d m_1^2 m_2^{d-2} + 2 (m_2^d - m_1^d) - d m_2^{d}}{\sin \left( \frac{d \pi}{2} \right) } \; . \label{eq:rel_entropy_density_calc} \numberthis
\end{align*}
For integer dimensions relevant in this work, the infinite volume relative entropy density reads
\begin{alignat}{2}
  &d = 1: \qquad d_\mathrm{KL} (\mu_1 \| \mu_2) &&= \frac{(m_1 - m_2)^2}{4 m_2} \; , \label{eq:rel_entropy_density_1d} \\
  &d = 2: \qquad d_\mathrm{KL} (\mu_1 \| \mu_2) &&= \frac{1}{8 \pi} \left( m_2^2 - m_1^2 + m_1^2 \log \frac{m_1^2}{m_2^2} \right) \; , \\
  &d = 3: \qquad d_\mathrm{KL} (\mu_1 \| \mu_2) &&= \frac{\left( 2 m_1 + m_2 \right) \left( m_1 - m_2 \right)^2}{24 \pi} \; , \label{eq:rel_entropy_density_3d}
\end{alignat}
where we used L'H\^opital's rule for the case $d=2$. Note that the result for $d=1$ coincides with result obtained previously.

Finally, we discuss the dependence of the relative entropy density on the dimension of Euclidean spacetime. We define $g \coloneqq m_1^2 - m_2^2$ and interpret it as a coupling constant. For $m_1^2 < 2 m_2^2$, we can expand the logarithm in \eqref{eq:KLdensity_continuum}, yielding the expansion
\begin{align*}
    d_\mathrm{KL} (\mu_1 \| \mu_2) &= \frac{1}{2} \int_{\R^d} \frac{\D^d p}{(2 \pi)^d} \left[ \frac{1}{2} \frac{g^2}{(|\boldsymbol{p}|^2 + m_2^2)^2} - \frac{1}{3} \frac{g^3}{(|\boldsymbol{p}|^2 + m_2^2)^3} + \frac{1}{4} \frac{g^4}{(|\boldsymbol{p}|^2 + m_2^2)^4} - \ldots \right] \\
    &= \frac{1}{2} \left[ 
    \frac{1}{2} \; \vcenter{\hbox{\begin{tikzpicture}[baseline={(X.base)}]
        \node[circle,draw,thick,inner sep=12pt] (X) at (0,0) {};
        \foreach \X in {0,180}
            {\draw[thick] node[Vertex] at (\X:0.6) {};}
        \foreach \X in {0,180}
            {\draw[thick] (\X:0.6) -- (\X+10:0.75);
             \draw[thick] (\X:0.6) -- (\X-10:0.75);}
        \end{tikzpicture}}}
    - \frac{1}{3} \; \vcenter{\hbox{\begin{tikzpicture}[baseline={(X.base)}]
        \node[circle,draw,thick,inner sep=12pt] (X) at (0,0) {};
        \foreach \X in {-30,90,210}
            {\draw[thick] node[Vertex] at (\X:0.6) {};}
        \foreach \X in {-30,90,210}
            {\draw[thick] (\X:0.6) -- (\X+10:0.75);
             \draw[thick] (\X:0.6) -- (\X-10:0.75);}
        \end{tikzpicture}}}
    + \frac{1}{4} \; \vcenter{\hbox{\begin{tikzpicture}[baseline={(X.base)}]
        \node[circle,draw,thick,inner sep=12pt] (X) at (0,0) {};
        \foreach \X in {0,90,180,270}
            {\draw[thick] node[Vertex] at (\X:0.6) {};}
        \foreach \X in {0,90,180,270}
            {\draw[thick] (\X:0.6) -- (\X+10:0.75);
             \draw[thick] (\X:0.6) -- (\X-10:0.75);}
        \end{tikzpicture}}}
    - \frac{1}{5} \; \vcenter{\hbox{\begin{tikzpicture}[baseline={(X.base)}]
        \node[circle,draw,thick,inner sep=12pt] (X) at (0,0) {};
        \foreach \X in {0+90,72+90,144+90,216+90,288+90}
            {\draw[thick] node[Vertex] at (\X:0.6) {};}
        \foreach \X in {0+90,72+90,144+90,216+90,288+90}
            {\draw[thick] (\X:0.6) -- (\X+10:0.75);
             \draw[thick] (\X:0.6) -- (\X-10:0.75);}
        \end{tikzpicture}}}
    + \ldots \right] \; . \numberthis \label{eq:dKL_diagrams}
\end{align*}
In the diagrammatic expression in the second line, each vertex contributes a factor $g$ and all external momenta are set to zero. Note that the leading (order $g^2$) diagram,
$\vcenter{\hbox{\begin{tikzpicture}[baseline={(X.base)}]
    \node[circle,draw,thick,inner sep=4pt] (X) at (0,0) {};
    \foreach \X in {0,180}
        {\draw[thick] node[Vertex] at (\X:0.2) {};}
    \foreach \X in {0,180}
        {\draw[thick] (\X:0.2) -- (\X+20:0.4);
        \draw[thick] (\X:0.2) -- (\X-20:0.4);}
\end{tikzpicture}}}$
, is divergent for $d \geq 4$, reflecting the crucial dependence of the relative entropy on the Euclidean spacetime dimension.

We note that \eqref{eq:dKL_diagrams} is structurally equivalent to the renormalized one loop effective potential of a scalar field theory with quartic self-interaction in two or three Euclidean spacetime dimensions, see, e.g., \cite[Sec.~5.3.3]{Coleman1985}. In particular, the \emph{regularized} Fredholm determinant in the expression for the relative entropy between two Gaussian measures, \eqref{eq:DKL_logdet2}, provides precisely the mass counterterm that cancels the divergent diagram
$\vcenter{\hbox{\begin{tikzpicture}[baseline={(X.base)}]
    \node[circle,draw,thick,inner sep=4pt] (X) at (0,0) {};
    \draw[thick] node[Vertex] at (270:0.2) {};
    \draw[thick] (270:0.2) -- (270+60:0.4);
    \draw[thick] (270:0.2) -- (270-60:0.4);
\end{tikzpicture}}}$
. However, for $d \geq 4$, the diagram
$\vcenter{\hbox{\begin{tikzpicture}[baseline={(X.base)}]
    \node[circle,draw,thick,inner sep=4pt] (X) at (0,0) {};
    \foreach \X in {0,180}
        {\draw[thick] node[Vertex] at (\X:0.2) {};}
    \foreach \X in {0,180}
        {\draw[thick] (\X:0.2) -- (\X+20:0.4);
        \draw[thick] (\X:0.2) -- (\X-20:0.4);}
\end{tikzpicture}}}$
is also divergent and we would need an additional counterterm which the regularized Fredholm determinant does not provide. This is the origin of the divergence of the relative entropy (density) in dimensions $d \geq 4$.

Finally, we comment briefly on the limit where one of the fields becomes massless. For Dirichlet boundary conditions in $d < 4$, the limit $m_i \to 0$ gives a finite result, while for Neumann and periodic boundary conditions this limit does not exist (see also Fig. \ref{fig:DKL_1d_Plot3D}), which is due to the zero mode of the Neumann and periodic Laplacians. For the relative entropy densities given in \eqref{eq:rel_entropy_density_1d} to \eqref{eq:rel_entropy_density_3d} we observe that the limit $m_2 \to 0$ exists only for $d=3$ or, more precisely, for $2 < d < 4$, as can be seen from \eqref{eq:rel_entropy_density_calc}. This is of course due to the well-known infrared divergence of massless theories in $d \leq 2$. Note that the limit $m_1 \to 0$ yields a finite result for all $d < 4$.

\subsection{Field Theories with Different Boundary Conditions}\label{sec:diff_bcs}

In this Section, we study the relative entropy between two fields over a bounded region with different boundary conditions. For simplicity, we assume that both fields have the same mass, and note that different masses can be incorporated using the results of Section \ref{sec:diff_masses}.

In general, two field theories with different boundary conditions can be mutually singular in all Euclidean spacetime dimensions even if they have the same mass parameter. As a specific example of two such fields, we consider the relative entropy between a Dirichlet and a Robin field.
\begin{theorem}\label{thm:main_2}
    Let $\Omega \subset \R^d$ be open, bounded and with $C^1$-boundary. Furthermore, let $m > 0$ and denote by $\mu_\mathrm{D} = \mathcal{N} (0,\hat{G}_\mathrm{D})$ and $\mu_\sigma = \mathcal{N} (0,\hat{G}_\sigma)$ the centred Gaussian measures over $\Omega$ with covariance operators $\hat{G}_\mathrm{D} = (- \Delta_\mathrm{D} + m^2)^{-1}$ and $\hat{G}_\sigma = (- \Delta_\sigma + m^2)^{-1}$, respectively. Then, $D_\mathrm{KL} (\mu_\mathrm{D} \| \mu_\sigma) = D_\mathrm{KL} (\mu_\sigma \| \mu_\mathrm{D}) = + \infty$ for all $d \in \N$.
\end{theorem}
\begin{proof}
    Recall from Theorem \ref{thm:equivalence} that a necessary condition for the equivalence of two centred Gaussian measures $\mu_1 = \mathcal{N} (0,\hat{C}_1)$ and $\mu_2 = \mathcal{N} (0,\hat{C}_2)$ (and hence for the finiteness of the relative entropy between them) is that the form domains of the precision operators $\hat{C}^{-1}_i$ coincide. For the case where $\mu_1 = \mu_\mathrm{D}$ is the Dirichlet field and $\mu_2 = \mu_\sigma$ is a Robin field, this boils down to whether the form domains of $-\Delta_\mathrm{D}$ and $- \Delta_\sigma$ are equal. Since $\Omega$ is a bounded $C^1$-region in $\R^d$, the form domains of the Dirichlet and Robin Laplacians are given by\footnote{For definitions, see Appendix \ref{app:sobolev_spaces}.} $\mathcal{Q} (-\Delta_\mathrm{D}) = H^{+1}_0 (\Omega)$ and $\mathcal{Q} (-\Delta_\sigma) = H^{+1} (\Omega)$, cf. \cite[Ch.~XIII]{Reed1978} and the discussion below. It can be shown that $H_0^{+1} (\Omega) \subset H^{+1} (\Omega)$ \cite[p.~253]{Reed1978}. In general, however, $H_0^{+1} (\Omega)$ is a proper subset of $H^{+1} (\Omega)$. In particular, if $\Omega$ is $C^0$, then $H_0^{+1} (\Omega) \subsetneq H^{+1} (\Omega)$ \cite[Cor.~3.29(vii)]{Chandler2017}. Therefore, the form domains of $-\Delta_\mathrm{D}$ and $- \Delta_\sigma$ do not coincide, and the relative entropy between a Dirichlet and a Robin field is infinite.
\end{proof}
\begin{remark}
    For the specific case $d = 1$ and $\Omega = (a,b)$, a bounded open interval in $\R$, $H^{+1} ((a,b))$ consists of all absolutely continuous functions on $(a,b)$ whose distributional derivatives are in $L^2((a,b))$, while $H_0^{+1} ((a,b))$ consists precisely of those functions in $H^{+1} ((a,b))$ whose continuous extensions to $[a,b]$ vanish at the endpoints of the interval. Thus, a non-zero constant function is in $H^{+1} ((a,b))$ but not in $H_0^{+1} ((a,b))$.
\end{remark}

We now study the relative entropy between two Robin fields with equal non-zero masses. We first derive a necessary and sufficient condition for the finiteness of the relative entropy between a Robin field and a Neumann field, and the result for two Robin fields then follows from transitivity. Let $\Omega \subset \R^d$ be a bounded Lipschitz domain. Recall from Section \ref{sec:classical_bcs} that the Neumann and Robin forms are defined by
\begin{align}
    \mathfrak{q}_\mathrm{N} (f,g) &= \int_\Omega \left( \boldsymbol{\nabla} f \cdot \boldsymbol{\nabla} g + m^2 f g \right) \; \D^d x \; , \\
    \mathfrak{q}_\sigma (f,g) &= \int_\Omega \left( \boldsymbol{\nabla} f \cdot \boldsymbol{\nabla} g + m^2 f g \right) \; \D^d x + \int_{\partial \Omega} \sigma f g \; \D S \; , \label{eq:robin_form}
\end{align}
where $\sigma$ is a positive\footnote{If we choose $\sigma \geq 0$, then $\mathfrak{q}_\sigma$ is strictly positive, and we avoid complications from negative eigenvalues of $- \Delta_\sigma$, see \cite[Sec.~1.1.20]{Robinson1971} and the discussion of the one-dimensional case below.} continuous function on $\partial \Omega$. The Robin form can be written as $\mathfrak{q}_\sigma = \mathfrak{q}_\mathrm{N} + \mathfrak{b}_\sigma$, where $\mathfrak{b}_\sigma$ is the form
\begin{equation}
    \mathfrak{b}_\sigma (f,g) = \int_{\partial \Omega} \sigma (\boldsymbol{x}) f (\boldsymbol{x}) g (\boldsymbol{x}) \; \D S (\boldsymbol{x}) = \braket{\sqrt{\sigma} \gamma f, \sqrt{\sigma} \gamma g}_{L^2 (\partial \Omega)} \; ,
\end{equation}
where $\gamma : H^{+1} (\Omega) \to L^2 (\partial \Omega)$ is the trace operator defined in Appendix \ref{app:sobolev_spaces}. It is shown in \cite[p.~34]{Robinson1971} that the boundary form $\mathfrak{b}_\sigma$ is infinitesimally relatively form bounded with respect to the Neumann form $\mathfrak{q}_\mathrm{N}$. This means that the Neumann and Robin form domains coincide (cf. \cite[Thm.~VI.1.33]{Kato1995}), and in particular $\mathcal{Q} (- \Delta_\mathrm{N}) = \mathcal{Q} (- \Delta_\sigma) = H^{+1} (\Omega)$. Clearly $\mathfrak{b}_\sigma$ is symmetric, densely defined and positive. However, it is not closable\footnote{To see this, note that $\mathfrak{b}_\sigma$ is a generalization of the form considered in \cite[Ex.~VI.1.26]{Kato1995}, and we can use the same arguments as there.}.

As usual, we use the notation $\hat{G}_\mathrm{N} = (- \Delta_\mathrm{N} + m^2)^{-1}$ and $\hat{G}_\sigma = (- \Delta_\sigma + m^2)^{-1}$. Unlike the case of two fields with different masses but the same (local) boundary conditions, studied in Section \ref{sec:diff_masses}, the covariance operators $\hat{G}_\mathrm{N}$ and $\hat{G}_\sigma$ have no common eigenbasis. Therefore, we cannot use the simple condition for $\hat{B} \hat{B}^* - I$ to be a Hilbert-Schmidt operator on $L^2 (\Omega)$, which we used in Lemma \ref{lem:equivalence_mass_independent}. From another point of view, the additional boundary term in the Robin form \eqref{eq:robin_form} prevents us from simply adding and subtracting the precision operators $\hat{G}^{-1}_\mathrm{N}$ and $\hat{G}^{-1}_\sigma$. However, following an idea from \cite{Pinski2015}, we can add and subtract their corresponding bilinear forms.
\begin{lemma}\label{lem:N_vs_Robin}
    Let $\Omega \subset \R^d$ be open, bounded and with $C^1$-boundary. Furthermore, let $m > 0$, $\sigma$ a positive continuous function on $\partial \Omega$ and denote by $\mu_\mathrm{N} = \mathcal{N} (0,\hat{G}_\mathrm{N})$ and $\mu_\sigma = \mathcal{N} (0,\hat{G}_\sigma)$ the centred Gaussian measures over $\Omega$ with covariance operators $\hat{G}_\mathrm{N} = (- \Delta_\mathrm{N} + m^2)^{-1}$ and $\hat{G}_\sigma = (- \Delta_\sigma + m^2)^{-1}$, respectively. Then, $\mu_\mathrm{N} \sim \mu_\sigma$ precisely when $d < 3$.
\end{lemma}
\begin{proof}
    As by the above discussion, the form domains of $-\Delta_\mathrm{N}$ and $- \Delta_\sigma$ coincide. Using results from Appendix \ref{app:equivalence}, we then see that $\mu_\mathrm{N} \sim \mu_\sigma$ precisely when
    \begin{equation}
        \hat{B}^* \hat{B} - I = \left( \hat{G}_\sigma^{-1/2} \hat{G}_\mathrm{N}^{1/2} \right)^* \left( \hat{G}_\sigma^{-1/2} \hat{G}_\mathrm{N}^{1/2} \right) - I
    \end{equation}
    is a Hilbert-Schmidt operator on $L^2 (\Omega)$. This is equivalent to the requirement that
    \begin{equation}
        \sum_{n=1}^\infty \sum_{m=1}^\infty \left( \widetilde{\Delta \mathfrak{q}} (e_n, f_m) \right)^2 < + \infty \; ,
    \end{equation}
    where $\{e_n\}_{n=1}^\infty$ and $\{f_m\}_{m=1}^\infty$ are any two orthonormal bases in $L^2 (\Omega)$ and $\widetilde{\Delta \mathfrak{q}}$ is the bounded bilinear form on $L^2 (\Omega)$ defined by
    \begin{equation}\label{eq:def_tildedeltaq}
        \widetilde{\Delta \mathfrak{q}} (f,g) \coloneqq \mathfrak{b}_\sigma \left( \hat{G}_\mathrm{N}^{1/2} f, \hat{G}_\mathrm{N}^{1/2} g \right) = \braket{f, (\hat{B}^* \hat{B} - I) \, g}_{L^2 (\Omega)} \; .
    \end{equation}

    Let $R_\mathrm{N} (\boldsymbol{x},\boldsymbol{y})$ be the distributional kernel of $\hat{G}^{1/2}_\mathrm{N}$. We conclude from \eqref{eq:def_tildedeltaq} that the operator $\hat{B}^* \hat{B} - I$ has a kernel $\Psi_\sigma$ given by
    \begin{equation}
        \Psi_\sigma (\boldsymbol{x},\boldsymbol{y}) = \int_{\partial \Omega} \sigma (\boldsymbol{z}) \, R_\mathrm{N} (\boldsymbol{x},\boldsymbol{z}) R_\mathrm{N} (\boldsymbol{z},\boldsymbol{y}) \; \D S (\boldsymbol{z}) \; .
    \end{equation}
    Since $\hat{B}^* \hat{B} - I$ is an operator on $L^2 (\Omega)$, it is Hilbert-Schmidt if and only if it has a kernel that is square integrable. A quick calculation shows that this requirement can be written as
    \begin{equation}\label{eq:int_Psi_finite}
        \iint_\Omega \left| \Psi_\sigma (\boldsymbol{x},\boldsymbol{y})\right|^2 \, \D^d x \, \D^d y = \iint_{\partial \Omega} \sigma (\boldsymbol{z}) \sigma (\boldsymbol{z}') \left| G_\mathrm{N} (\boldsymbol{z},\boldsymbol{z}') \right|^2 \, \D S (\boldsymbol{z}) \, \D S (\boldsymbol{z}') < + \infty \; ,
    \end{equation}
    i.e., $\mu_\mathrm{N} \sim \mu_\sigma$ precisely when the Green's functions are square integrable on the boundary. This is certainly the case in $d=1$, since then the Green's functions are continuous on the diagonal and the integral over the boundary becomes a sum over a finite number of boundary points.

    For $d > 1$, we use the assumption that $\partial \Omega$ is sufficiently regular, i.e., $C^1$. We can then check whether the Green's function is square integrable by checking whether the singularity at $\boldsymbol{z} = \boldsymbol{z}'$ is square integrable. Since the behaviour of the Green's function for small distances $|\boldsymbol{x} - \boldsymbol{y}|$ is a UV property, it does not depend on the choice of boundary conditions, and we can use \eqref{eq:greens_diagonal}, the small distance behaviour of the fundamental solution, to estimate the singularity (see also \cite[Lem.~III.2]{Guerra1976}). In particular, for $|\boldsymbol{x} - \boldsymbol{y}| \ll m^{-1}$,
    \begin{equation}
        G_\mathrm{N} (\boldsymbol{x},\boldsymbol{y}) \sim \begin{cases} 
            \log \frac{1}{|\boldsymbol{x} - \boldsymbol{y}|} \quad & \mathrm{for} \;\; d = 2 \\
            \frac{1}{|\boldsymbol{x} - \boldsymbol{y}|^{d-2}} \quad & \mathrm{for} \;\; d > 2
        \end{cases} \; .
    \end{equation}
    Fix $\boldsymbol{z}' \in \partial \Omega$. Choose the coordinate system in $\R^d$ such that $\boldsymbol{z}' = \boldsymbol{0}$ and $\partial \Omega$ is flat near $\boldsymbol{z}'$ and lying in the plane $y_d = 0$\footnote{If $\partial \Omega$ is not flat near $\boldsymbol{z}'$, we can ``flatten out'' $\partial \Omega$ in the vicinity of $\boldsymbol{z}'$ using a continuously differentiable map, cf. \cite[Sec.~5.4]{Evan2010}.}. Then, the integral of $(G_\mathrm{N})^2$ near $\boldsymbol{z} = \boldsymbol{z}' = \boldsymbol{0}$, e.g., over a ball $\mathbb{B}_\varepsilon$ with radius $\varepsilon \ll m^{-1}$, is essentially given by
    \begin{equation}
        \int_{\mathbb{B}_\varepsilon} \frac{\D^{d-1} z}{|\boldsymbol{z}|^{2d-4}} \; \sim \; \int_0^\varepsilon \frac{\D r}{r^{d-2}}
    \end{equation}
    This integral is finite for $d<3$ and divergent otherwise (in particular, the integral diverges logarithmically for $d = 3$). Therefore, $\mu_\mathrm{N} \sim \mu_\sigma$ precisely when $d < 3$.
\end{proof}
\begin{remark}
    Note that the case where the Robin field is replaced by a Dirichlet field can be formally incorporated by choosing $\sigma (\boldsymbol{x}) \equiv \sigma_* > 0$ and then taking the limit $\sigma_* \to + \infty$. Then the condition in \eqref{eq:int_Psi_finite} is violated in all dimensions, which is consistent with our conclusion from the discussion of the Dirichlet and Neumann form domains above.
\end{remark}

We conclude from the above Lemma that a Neumann and a Robin field with equal masses over a bounded region $\Omega$ are equivalent in the sense of measures (and thus the relative entropy between them is finite) if and only if $d < 3$. As mentioned earlier, this result easily generalizes to the case of two Robin fields. By transitivity, Lemma \ref{lem:N_vs_Robin} implies that the relative entropy between two Robin fields is finite precisely when $d < 3$. We summarize this result in the following Theorem.
\begin{theorem}\label{thm:main_3}
    Let $\Omega \subset \R^d$ be open, bounded and with $C^1$-boundary. Furthermore, let $m > 0$, $\sigma_1$ and $\sigma_2$ two positive continuous functions on $\partial \Omega$ such that $\sigma_1 \neq \sigma_2$, and denote by $\mu_1 = \mathcal{N} (0,\hat{G}_{\sigma_1})$ and $\mu_2 = \mathcal{N} (0,\hat{G}_{\sigma_2})$ the centred Gaussian measures over $\Omega$ with covariance operators $\hat{G}_{\sigma_1} = (- \Delta_{\sigma_1} + m^2)^{-1}$ and $\hat{G}_{\sigma_2} = (- \Delta_{\sigma_2} + m^2)^{-1}$, respectively. Then, $D_\mathrm{KL} (\mu_1 \| \mu_2)$ is finite precisely when $d < 3$.
\end{theorem}

We now give a concrete calculation of the relative entropy between two Robin fields in $d=1$. As already mentioned (and also discussed in \cite{Guerra1975a}), the relative entropy between two such theories is finite. Consider two field theories on an interval with equal masses but different Robin boundary conditions of the form
\begin{equation}\label{eq:sigma_bcs}
    \frac{\partial f}{\partial n} = - \sigma f \; , \qquad \sigma \geq 0 \; .
\end{equation}
In particular, this class of boundary conditions includes Neumann ($\sigma = 0$), free ($\sigma = m$) and Dirichlet ($\sigma = +\infty$) boundary conditions.

In the following, let $\Omega = (0, L)$. The Radon-Nikodym derivative of the field theory $\mu_0$ with free boundary conditions with respect to a field theory $\mu_\sigma$ with boundary conditions \eqref{eq:sigma_bcs} can be shown to be (cf. \cite[Thm.~II.31]{Guerra1975a}, \cite[p.~263]{Guerra1976})
\begin{equation}\label{eq:density_sigma_bcs}
    \frac{\D \mu_0}{\D \mu_\sigma} (\varphi) = \E^{\frac{\sigma - m}{2} (\varphi(0)^2 + \varphi(L)^2)} \int \E^{- \frac{\sigma-m}{2} (\varphi(0)^2 + \varphi(L)^2)} \; \D \mu_0 (\varphi) \; .
\end{equation}
The normalization constant on the right-hand side can be computed, using the change of variables formula \cite[Eq.~(0.1)]{Bogachev2014}, to wit
\begin{align*}
    \int \E^{- \frac{\sigma-m}{2} (\varphi(0)^2 + \varphi(L)^2)} \; \D \mu_0 (\varphi) &= \frac{1}{2 \pi \sqrt{\det \Sigma}} \int_{\R^2} \E^{- \frac{\sigma-m}{2} \|\boldsymbol{x}\|^2} \, \E^{- \frac{1}{2} \boldsymbol{x}^\mathsf{T} \Sigma^{-1} \boldsymbol{x} } \; \D^2 x \\
    &= \sqrt{\det \left(\Sigma^{-1} \widetilde{\Sigma}\right)} \\
    &= \frac{2 m}{\sqrt{(\sigma + m)^2 - \E^{-2 m L} (\sigma - m)^2}} \; , \numberthis
\end{align*}
where $\Sigma$ and $\widetilde{\Sigma}$ are $2 \times 2$ matrices given by
\begin{equation}
    \Sigma = \begin{pmatrix}
        G (0,0;m) & G (0,L;m) \\
        G (L,0;m) & G (L,L;m)
    \end{pmatrix} = \frac{1}{2 m} \begin{pmatrix}
        1 & \E^{- m L} \\
        \E^{- m L} & 1
    \end{pmatrix} \; ,
\end{equation}
and $\widetilde{\Sigma} = \left( \Sigma^{-1} + (\sigma - m) \mathds{1}_2 \right)^{-1}$.

Given the density in \eqref{eq:density_sigma_bcs}, it is straightforward to calculate the relative entropy between $\mu_0$ and $\mu_\sigma$,
\begin{align*}
    D_\mathrm{KL} (\mu_0 \| \mu_\sigma) &= \int \log \left[ \frac{2 m}{\sqrt{(\sigma + m)^2 - \E^{-2 m L} (\sigma - m)^2}} \; \E^{\frac{\sigma - m}{2} (\varphi(0)^2 + \varphi(L)^2)} \right] \; \D \mu_0 (\varphi) \\
    &= \frac{1}{2} \log \left( \frac{4 m^2}{(\sigma + m)^2 - \E^{-2 m L} (\sigma - m)^2} \right) \\
    &+ \frac{\sigma - m}{2} \underbrace{\int (\varphi(0)^2 + \varphi(L)^2) \; \D \mu_0 (\varphi)}_{= \, G (0,0;m) + G (L,L;m) \, = \, 1/m} \\
    &= \frac{1}{2} \left[ \log \left( \frac{4 m^2}{(\sigma + m)^2 - \E^{-2 m L} (\sigma - m)^2} \right) + \frac{\sigma}{m} - 1 \right] \; . \label{eq:DKLsigma1d} \numberthis
\end{align*}

\begin{figure}[t]
    \centering
    \includegraphics[width=0.5\textwidth]{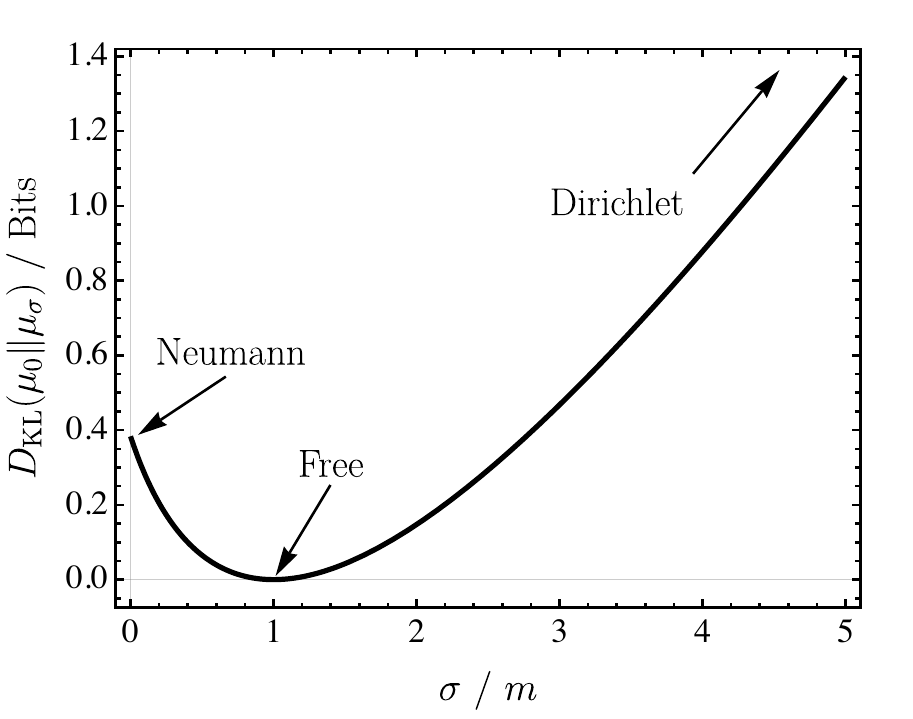}%
    \caption{The relative entropy between a one-dimensional field theory with free boundary conditions and a field theory with $\sigma$-boundary conditions given in \eqref{eq:sigma_bcs}, both with mass $m$, on an interval $\Omega = (0,1)$. As indicated by the arrows, the special values $\sigma = 0$, $\sigma = m$ and $\sigma = + \infty$ correspond to Neumann, free and Dirichlet boundary conditions, respectively. The relative entropy scales linearly in $\sigma$ for large $\sigma$, indicating that a $\sigma$-boundary condition field is mutually singular to a Dirichlet field.}
    \label{fig:sigma_bcs}
\end{figure}

We plot $D_\mathrm{KL} (\mu_0 \| \mu_\sigma)$ in Fig. \ref{fig:sigma_bcs} for $L=1$. We see that $D_\mathrm{KL} (\mu_0 \| \mu_\sigma)$ vanishes precisely when $\sigma = m$. It diverges when we take the limit $\sigma \to + \infty$, which is again consistent with our conclusion from the discussion of the Dirichlet and Neumann form domains. Note that the expression in \eqref{eq:DKLsigma1d} diverges in the limit
\begin{equation}\label{eq:robin_becomes_negative}
    \sigma \to - m \tanh \left( \frac{m L}{2} \right) < 0
\end{equation}
from above. This is due to the fact that for this value of $\sigma$ the operator $-\Delta_\sigma + m^2$ gets a zero eigenvalue. More precisely, for $\sigma < 0$, the Robin Laplacian $- \Delta_\sigma$ has at least one negative eigenvalue, with the smallest eigenvalue $\lambda_1 (\sigma) = - \mu^2$ satisfying $\lambda (\sigma) < - |\sigma|^2$, where $\mu$ is the smallest positive solution of \cite[Sec.~1.1.20]{Robinson1971}
\begin{equation}
    \tanh (\mu L) - \frac{2 |\sigma| \mu}{\mu^2 + \sigma^2} = 0 \; .
\end{equation}
If $\sigma$ is set to the value in \eqref{eq:robin_becomes_negative}, then $\mu = m$ is the smallest solution to the above equation and therefore the operator $- \Delta_\sigma + m^2$ is no longer strictly positive for $\sigma \leq - m \tanh \left( \frac{m L}{2} \right)$. This explains the divergence of \eqref{eq:DKLsigma1d} for too small values of $\sigma$.

\subsection{Mutual Information}\label{sec:mutual_info}

In this Section, we discuss a special kind of relative entropy, the mutual information between two disjoint regions $\Omega_A$ and $\Omega_B$. We show that the mutual information between two bounded regions is always finite if the these regions are separated by a finite distance. Furthermore, we give an example where the mutual information is infinite if the regions touch, i.e., if the separation distance is zero. We argue that the reason for this behaviour is the Markov property of a scalar Euclidean field theory \cite{Nelson1973a,Nelson1973b,Guerra1975a,Guerra1975b,Guerra1976}.

In the following, let $\Omega_A$ and $\Omega_B$ be disjoint open $C^0$ subsets of $\R^d$. Furthermore, let $L^2 (\Omega_A \cup \Omega_B)$ be the Hilbert space of square-integrable functions on $\Omega_A \cup \Omega_B$. Then, $L^2 (\Omega_A \cup \Omega_B) \simeq L^2(\Omega_A) \oplus L^2 (\Omega_B)$ where the isomorphism is given by the map $f \mapsto (f_A,f_B)$, where $f_A$ ($f_B$) is the restriction of $f$ to $\Omega_A$ ($\Omega_B$). We define the mutual information between the regions $\Omega_A$ and $\Omega_B$ to be the relative entropy
\begin{equation}\label{eq:mutual_info_regions}
    I(\Omega_A : \Omega_B) \coloneqq D_\mathrm{KL} ( \mu_{AB} \| \mu_A \otimes \mu_B ) \; .
\end{equation}
In the above expression, $\mu_{AB}$ is the free scalar field theory of mass $m > 0$ (and vanishing mean) over the region $\Omega_A \cup \Omega_B$ with free boundary conditions. Similarly, $\mu_A$ and $\mu_B$ are the free scalar field theories with the same mass $m$ (and vanishing means) over the regions $\Omega_A$ and $\Omega_B$, respectively, and $\mu_A \otimes \mu_B$ is the product measure.

These field theories can be interpreted as follows. The measure $\mu_{AB}$ describes a theory reduced to the region $\Omega_A \cup \Omega_B$, containing correlations between all spacetime points in $\Omega_A \cup \Omega_B$. In contrast, the field theory $\mu_A \otimes \mu_B$ only describes correlations between spacetime points within each regions, but it does not contain any cross-correlations between the regions. More precisely, if $\boldsymbol{x}_A \in \Omega_A$ and $\boldsymbol{x}_B \in \Omega_B$, then $\boldsymbol{x}_A$ and $\boldsymbol{x}_B$ are correlated in the theory $\mu_{AB}$ but uncorrelated in the theory $\mu_A \otimes \mu_B$. In other words, the field $\varphi_A$ in $\Omega_A$ and the field $\varphi_B$ in $\Omega_B$ are independent Gaussian random variables in the theory $\mu_A \otimes \mu_B$. The mutual information $I(\Omega_A : \Omega_B)$ can thus be interpreted as the distinguishability of a theory containing all cross-correlations and a theory containing no cross-correlations at all. We proceed by discussing the covariance operators of these two field theories.

The covariance operator of the centred Gaussian measure $\mu_{AB}$, denoted $\hat{G}_0$, is given by
\begin{equation}
    (\hat{G}_0 f) (\boldsymbol{x}) = \int_{\Omega_A \cup \Omega_B} G (\boldsymbol{x},\boldsymbol{y};m) f (\boldsymbol{y}) \; \D^n y \; , \qquad f \in L^2 (\Omega_A \cup \Omega_B) \; ,
\end{equation}
where $G$ is the fundamental solution of $- \Delta + m^2$ given in \eqref{eq:fundamental_solution}. Thus, we can interpret $\mu_{AB}$ as the free scalar field over $\Omega_A \cup \Omega_B$ with free boundary conditions, cf. Section \ref{sec:covariance_operators}. Recalling that $L^2 (\Omega_A \cup \Omega_B) \simeq L^2(\Omega_A) \oplus L^2 (\Omega_B)$, we can represent $\hat{G}_0$ as an operator on $L^2 (\Omega_A) \oplus L^2 (\Omega_B)$ by the matrix
\begin{equation}\label{eq:G0_matrix}
    \hat{G}_0 = \begin{pmatrix}
                    \hat{G}_A & \hat{G}_{AB} \\
                    \hat{G}^*_{AB} & \hat{G}_B
                \end{pmatrix} \; ,
\end{equation}
where the operators $\hat{G}_A : L^2 (\Omega_A) \to L^2 (\Omega_A)$, $\hat{G}_B : L^2 (\Omega_B) \to L^2 (\Omega_B)$ and $\hat{G}_{AB} : L^2 (\Omega_B) \to L^2 (\Omega_A)$ are defined as
\begin{alignat}{3}
    f_A &\mapsto (\hat{G}_A f_A) (\boldsymbol{x}) &&= \int_{\Omega_A} G (\boldsymbol{x},\boldsymbol{y};m) f_A (\boldsymbol{y}) \; \D^n y \; , \qquad &&\boldsymbol{x} \in \Omega_A \; , \\
    f_B &\mapsto (\hat{G}_B f_B) (\boldsymbol{x}) &&= \int_{\Omega_B} G (\boldsymbol{x},\boldsymbol{y};m) f_B (\boldsymbol{y}) \; \D^n y \; , \qquad &&\boldsymbol{x} \in \Omega_B \; , \\
    f_B &\mapsto (\hat{G}_{AB} f_B) (\boldsymbol{x}) &&= \int_{\Omega_B} G (\boldsymbol{x},\boldsymbol{y};m) f_B (\boldsymbol{y}) \; \D^n y \; , \qquad &&\boldsymbol{x} \in \Omega_A \; .
\end{alignat}
The corresponding covariance reads
\begin{align*}
    \mathrm{Cov} (f,g) &= \braket{f,\hat{G}_0 g}_{L^2 (\Omega_A \cup \Omega_B)} \\
    &= \braket{f_A,\hat{G}_A g_A}_{L^2 (\Omega_{A})} + \braket{f_B,\hat{G}_B g_B}_{L^2 (\Omega_{B})} + \braket{f_A,\hat{G}_{AB} g_B}_{L^2 (\Omega_{A})} \\
    &+ \braket{f_B,\hat{G}^*_{AB} g_A}_{L^2 (\Omega_{B})} \; . \numberthis
\end{align*}
From the above expression we can see that the operator $\hat{G}_{AB}$, together with its adjoint $\hat{G}^*_{AB}$, describes correlations across the two regions $\Omega_A$ and $\Omega_B$. Therefore, we call $\hat{G}_{AB}$ the cross-covariance operator of $\mu_{AB}$ \cite{Baker1970,Baker1973}.

For the product measure $\mu_A \otimes \mu_B$, the covariance operator $\hat{G}_\otimes$ is given by
\begin{equation}
    (\hat{G}_\otimes f) (\boldsymbol{x}) = \chi_{\Omega_A} (\boldsymbol{x}) \int_{\Omega_A} G (\boldsymbol{x},\boldsymbol{y};m) f_A (\boldsymbol{y}) \; \D^n y + \chi_{\Omega_B} (\boldsymbol{x}) \int_{\Omega_B} G (\boldsymbol{x},\boldsymbol{y};m) f_B (\boldsymbol{y}) \; \D^n y \; ,
\end{equation}
where $\chi_{\Omega_i}$ is the indicator function on $\Omega_i$. We can represent $\hat{G}_\otimes$ as an operator on $L^2 (\Omega_A) \oplus L^2 (\Omega_B)$ by the matrix
\begin{equation}
    \hat{G}_\otimes = \begin{pmatrix}
                    \hat{G}_A & 0 \\
                    0 & \hat{G}_B
                \end{pmatrix} \; .
\end{equation}
By construction $\mu_A \otimes \mu_B$ contains no cross-correlations between the regions $\Omega_A$ and $\Omega_B$.

From the definition of the mutual information in \eqref{eq:mutual_info_regions} and the discussion in Section \ref{sec:info_theory}, we see that $I (\Omega_A : \Omega_B)$ is finite precisely when $\mu_{AB} \sim \mu_A \otimes \mu_B$. In order to find conditions when this is the case, we first need some auxiliary results from the theory of Sobolev spaces, most of which are obtained in \cite{Guerra1975a,Guerra1976}. For a short collection of definitions and properties of Hilbert-Sobolev spaces see Appendix \ref{app:sobolev_spaces}.

Recall that $D$ is the unique self-adjoint extension of $(-\Delta + m^2) |_{C_0^\infty (\R^d)}$. The Hilbert-Sobolev space $H^{\pm 1} (\R^d)$ of order $\pm 1$ is the closure of $C_0^\infty (\R^d)$ in the inner product\footnote{Note that these inner products differ from those used in the definition of Sobolev spaces in Appendix \ref{app:sobolev_spaces} due to the general mass term. Nevertheless, the norms induced by these inner products and those used in Appendix \ref{app:sobolev_spaces} are equivalent.} $\braket{f,g}_{\pm 1} = \braket{f,D^{\pm 1} g}_{L^2}$. We recall the scale of Hilbert spaces
\begin{equation}
    H^{+1} (\R^d) \subset L^2 (\R^d) \subset H^{-1} (\R^d) \; .
\end{equation}
In particular, every $\varphi \in H^{-1} (\R^d)$ can be interpreted as a tempered distribution, i.e., an element in $\mathcal{S}^*$. For any closed subset $K$ of $\R^d$, the space of distributions with support in $K$,
\begin{equation}
    H^{-1}_K = \left\lbrace \varphi \in H^{-1} (\R^d) \; : \; \mathrm{supp} \, \varphi \subseteq K \right\rbrace
\end{equation}
is a closed subspace of $H^{-1} (\R^d)$. We denote by $e_K$ the orthogonal projection from $H^{-1} (\R^d)$ onto $H^{-1}_K$.

Let $\Omega \subset \R^d$ be an open subset, $H^{-1} (\Omega)$ the Hilbert-Sobolev space as defined in Appendix \ref{app:sobolev_spaces} and define $p_\Omega \coloneqq I - e_{\R^d \setminus \Omega}$, the projection onto $(H^{-1}_{\R^d \setminus \Omega})^\perp$. This is essentially a projection to distributions with support in $\Omega$.
\begin{lemma}[{\cite[Lem.~II.24]{Guerra1975a}}]\label{lem:Dirichlet_cov}
    Let $- \Delta_\mathrm{D}$ be the Dirichlet Laplacian on $L^2 (\Omega)$. For every $f \in H^{-1} (\Omega) \supset L^2 (\Omega)$,
    \begin{equation}
        (- \Delta_\mathrm{D} + m^2 )^{-1} f = \hat{G}_\mathrm{D} f = D^{-1} p_\Omega f \; .
    \end{equation}
\end{lemma}
Upon recalling that the inner product in $H^{-1} (\Omega)$ is given by $\braket{f,g}_{H^{-1} (\Omega)} = \braket{p_\Omega f, p_\Omega g}_{-1}$ and using \cite[Cor.~3.29(ii)]{Chandler2017}, we infer that $H^{-1} (\Omega)$ is the closure of $C_0^\infty (\Omega)$ with respect to the inner product
\begin{equation}\label{eq:inner_product_H-1}
    \braket{f,g}_{H^{-1} (\Omega)} = \braket{f, \hat{G}_\mathrm{D} g}_{L^2 (\Omega)} \; , \qquad f,g \in C_0^\infty (\Omega) \; .
\end{equation}

The proof that $\mu_{AB} \sim \mu_A \otimes \mu_B$ if the regions are separated relies on the following
\begin{lemma}[{\cite[Lem.~II.35]{Guerra1975a}, \cite[Thm.~III.16]{Simon1974}}]\label{lem:projection_traceclass}
    Let $\Omega_1$ and $\Omega_2$ be open subsets in $\R^d$ and denote their closures by $\Lambda_1$ and $\Lambda_2$, respectively. Suppose that $\Omega_1$ is bounded and $\Omega_1$ and $\Omega_2$ are separated by a finite distance. Then, the operator $\alpha = e_{\Lambda_1} e_{\Lambda_2} e_{\Lambda_1}$ is of trace class on $H^{-1} (\R^d)$ and $\|\alpha\| < 1$. Moreover, if $f$ is an eigenvector of $\alpha$ such that the corresponding eigenvalue is nonzero, then $f \in H^{-1} (\partial \Omega_1)$.
\end{lemma}
The following Lemma is useful for showing that the form domains of two precision operators coincide.
\begin{lemma}[{\cite[Prop.~B.1]{DaPrato2014}}]\label{lem:same_range}
    Let $T_1$ and $T_2$ be two bounded self-adjoint operators on a separable Hilbert space $\mathcal{H}$ with norm $\|.\|$. The ranges of $T_1$ and $T_2$ coincide, i.e., $T_1 [\mathcal{H}] = T_2 [\mathcal{H}]$, if and only if there exist constants $\gamma, \Gamma > 0$, $\gamma \leq \Gamma$, such that $\gamma \|T_1 f\| \leq \|T_2 f\| \leq \Gamma \|T_1 f\|$ for all $f \in \mathcal{H}$.
\end{lemma}

\begin{figure}[t]
    \centering
    \includegraphics[width=0.4\textwidth]{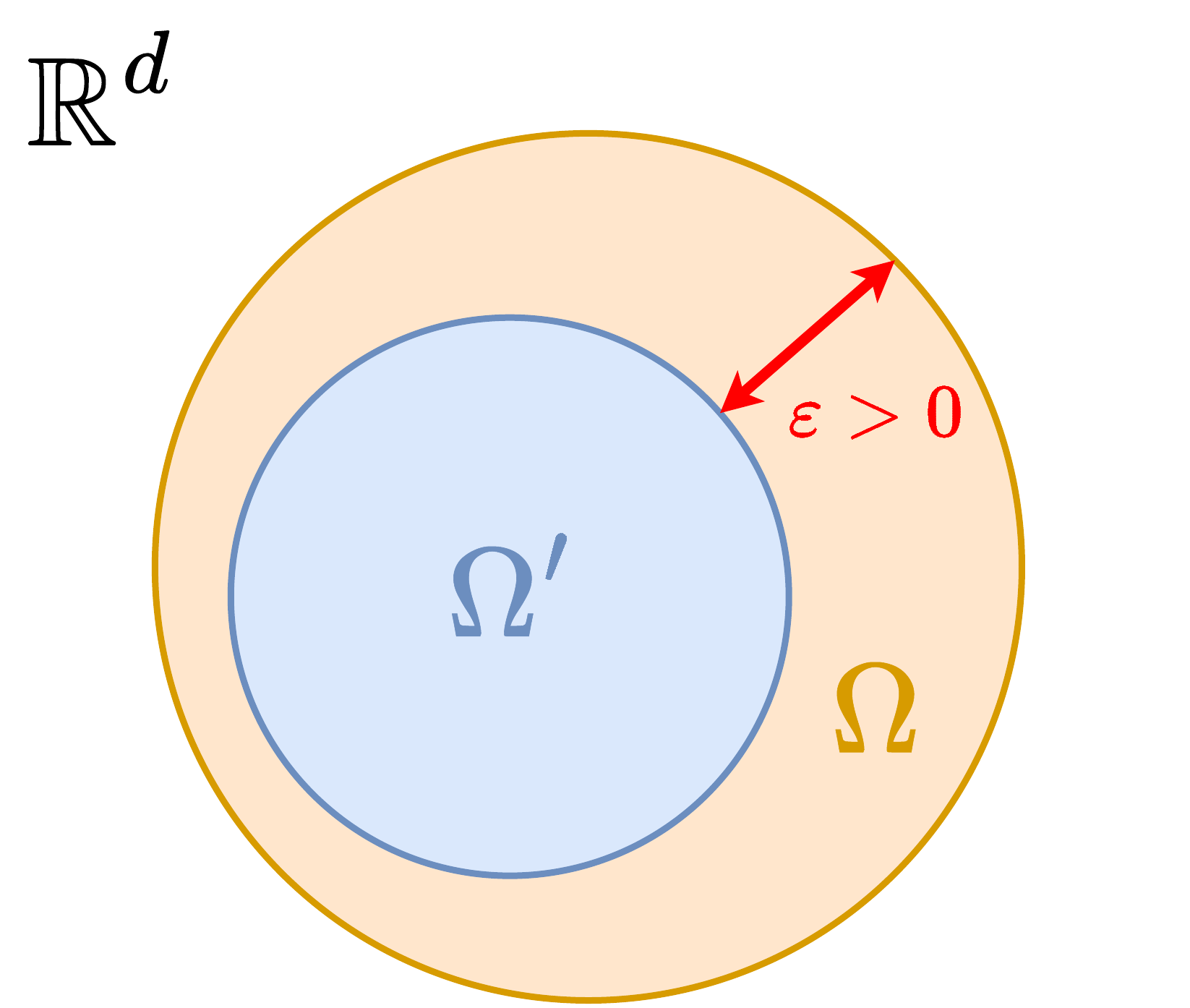}%
    \caption{Sketch of regions considered in Theorem \ref{thm:equiv_open_Dirichlet}. The bounded open sets $\Omega$ and $\Omega'$ are chosen such that $\overline{\Omega'}$ is a subset of $\Omega$ and $\mathrm{dist} (\Omega', \partial \Omega) > 0$. Two field theories with free and Dirichlet boundary conditions on $\partial \Omega$, respectively, will, in general, not be equivalent. However, they are equivalent if we restrict them to the smaller region $\Omega'$.}
    \label{fig:region1}
\end{figure}

The following Theorem states that while a Dirichlet field theory and a field theory with free boundary conditions over a bounded region $\Omega$ are generally mutually singular, if we restrict both field theories to a smaller subregion $\Omega'$, then they are in fact equivalent, see also Fig. \ref{fig:region1}.
\begin{theorem}[{\cite[Thm.~II.34]{Guerra1975a}}]\label{thm:equiv_open_Dirichlet}
    Let $\Omega$ and $\Omega'$ be bounded open regions in $\R^d$ such that $\Omega' \subset \Omega$ and $\mathrm{dist} (\Omega', \partial \Omega) > 0$. Let $\mu = \mathcal{N} (0, \hat{G}_0)$ and $\nu = \mathcal{N} (0, \hat{G}_\mathrm{D}^\Omega)$ be centred Gaussian measures over $\Omega'$, where $\hat{G}_\mathrm{D}^\Omega$ is the integral operator on $L^2 (\Omega')$ whose integral kernel is the restriction to $\Omega' \times \Omega'$ of the Dirichlet Green's function on $\Omega \times \Omega$. Then the measures $\mu$ and $\nu$ are equivalent, which implies that the relative entropy between them is finite.
\end{theorem}
\begin{remark}
    Note that this result also holds when we consider boundary conditions other than Dirichlet \cite[Thm.~II.2]{Guerra1976}. Specifically, it also holds for Neumann boundary conditions \cite[Thm.~III.6]{Guerra1976}.
\end{remark}
\begin{proof}
    Let $\Lambda^\mathrm{ext} = \R^d \setminus \Omega$. For all $f,g \in L^2 (\Omega')$, $\braket{f, \hat{G}_0 g}_{L^2 (\Omega')} = \braket{f,g}_{-1}$, and, by Lemma \ref{lem:Dirichlet_cov}, $\braket{f, \hat{G}_\mathrm{D}^\Omega g} = \braket{f, p_{\Omega} g}_{-1}$. Denote the closure $\overline{\Omega'}$ by $\Lambda'$. Using the self-adjointness of $e_{\Lambda'}$ with respect to $\braket{.,.}_{-1}$ and the fact that $e_{\Lambda'}$ acts as the identity on $H^{-1}_{\Lambda'} = \widetilde{H}^{-1} (\Omega')$, we can furthermore write
    \begin{equation}
        \braket{f, p_{\Omega} g}_{-1} = \braket{f, P g}_{-1} \; , \qquad f,g \in H^{-1}_{\Lambda'} \supset L^2 (\Omega) \; ,
    \end{equation}
    where we defined $P \coloneqq I - \alpha$ and $\alpha \coloneqq e_{\Lambda'} e_{\Lambda^\mathrm{ext}} e_{\Lambda'}$. By Lemma \ref{lem:projection_traceclass}, $\alpha \in \mathrm{HS} (H^{-1}_{\Lambda'})$ (see also \cite[p.~169]{Guerra1975a} and \cite[p.~217]{Simon1974}) and $\|\alpha\| < 1$. In particular, this means that $P = I-\alpha$ is strictly positive and boundedly invertible. Furthermore, $P$ is self-adjoint as an operator on $H^{-1} (\R^d)$, and it has a unique positive square root $P^{1/2}$.

    By Lemma \ref{lem:operator_inequalities}, $\braket{f, \hat{G}_\mathrm{D}^\Omega f}_{L^2 (\Omega')} \leq \braket{f, \hat{G}_0 f}_{L^2 (\Omega')}$ for all $f \in L^2 (\Omega')$. Furthermore,
    \begin{equation}
        \braket{f, \hat{G}_0 f}_{L^2 (\Omega')} = \braket{f, f}_{-1} \leq \| P^{-1} \| \, \braket{f, P f}_{-1} = \| P^{-1} \| \, \braket{f, \hat{G}_\mathrm{D}^\Omega f}_{L^2 (\Omega')} \; .
    \end{equation}
    Thus, by Lemma \ref{lem:same_range}, $\hat{G}_0^{1/2} [L^2 (\Omega')] = (\hat{G}_\mathrm{D}^\Omega)^{1/2} [L^2 (\Omega')]$.

    We define $\hat{B} \coloneqq \hat{G}_0^{-1/2} (\hat{G}_\mathrm{D}^\Omega)^{1/2}$ on $L^2 (\Omega')$. Let $\{\phi_n\}_{n=1}^\infty$ be an eigenbasis of $\hat{G}_0$ in $L^2 (\Omega')$. Then, $\{\psi_n\}_{n=1}^\infty$ where $\psi_n \coloneqq \hat{G}_0^{-1/2} \phi_n$ is an orthonormal basis in $H^{-1}_{\Lambda'} = \widetilde{H}^{-1} (\Omega')$. Then, as $\alpha \in \mathrm{HS} (H^{-1}_{\Lambda'})$,
    \begin{equation}
        \| \hat{B} \hat{B}^* - I \|^2_{\mathrm{HS} (L^2 (\Omega'))} = \sum_{n=1}^\infty \| ( P - I ) \psi_n \|^2_{-1} = \| \alpha \|^2_{\mathrm{HS} \left( H^{-1}_{\Lambda'} \right)} < + \infty \; .
    \end{equation}
    Therefore, by Theorem \ref{thm:equivalence}, the Gaussian measures $\mu = \mathcal{N} (0, \hat{G}_0)$ and $\nu = \mathcal{N} (0, \hat{G}_\mathrm{D}^\Omega)$ are equivalent.
\end{proof}

We will use Theorem \ref{thm:equiv_open_Dirichlet} to show that the mutual information between separated regions is finite. Another result that will turn out to be useful is that a Dirichlet field over $\Omega_A \cup \Omega_B$ (with $\Omega_A$ and $\Omega_B$ separated) factorises into a product measure. We first need the following
\begin{proposition}[{\cite[Prop.~XIII.3]{Reed1978}}]\label{prop:dirichlet_factorisation}
    Let $\Omega_A$ and $\Omega_B$ be disjoint open subsets of $\R^d$. Then, the Dirichlet Laplacian $-\Delta^{\Omega_A \cup \Omega_B}_\mathrm{D}$ on $L^2 (\Omega_A \cup \Omega_B)$ can be represented as an operator on $L^2 (\Omega_A) \oplus L^2 (\Omega_B)$ by the matrix
    \begin{equation}
    -\Delta^{\Omega_A \cup \Omega_B}_\mathrm{D} = \begin{pmatrix}
                    -\Delta^{\Omega_A}_\mathrm{D} & 0 \\
                    0 & -\Delta^{\Omega_B}_\mathrm{D}
                \end{pmatrix} \; ,
\end{equation}
where $-\Delta^{\Omega_i}_\mathrm{D}$ is the Dirichlet Laplacian on $L^2 (\Omega_i)$, $i \in \{A,B\}$.
\end{proposition}
From the above Proposition, we see that the Dirichlet covariance operator on $L^2 (\Omega_A \cup \Omega_B)$ can be represented as an operator on $L^2 (\Omega_A) \oplus L^2 (\Omega_B)$ by the matrix
    \begin{equation}
    \hat{G}_\mathrm{D}^{\Omega_A \cup \Omega_B} = \begin{pmatrix}
                    \left( -\Delta^{\Omega_A}_\mathrm{D} + m^2 \right)^{-1} & 0 \\
                    0 & \left( -\Delta^{\Omega_A}_\mathrm{D} + m^2 \right)^{-1}
                \end{pmatrix} = \begin{pmatrix}
                    \hat{G}^{\Omega_A}_\mathrm{D} & 0 \\
                    0 & \hat{G}^{\Omega_B}_\mathrm{D}
                \end{pmatrix} \; .
\end{equation}
Notice that a Dirichlet field theory over $\Omega_A \cup \Omega_B$ therefore contains no cross-correlations between the regions $\Omega_A$ and $\Omega_B$. Intuitively, we can think of Dirichlet boundary conditions as separating a bounded open region from the interior of its complement. The boundary $\partial \Omega$ acts as a ``wall'' between $\Omega$ and $\mathrm{int} (\R^d \setminus \Omega)$ as the fields are fixed to be zero on $\partial \Omega$. In other words, Dirichlet boundary conditions on $\partial \Omega$ cause the field in $\Omega$ to decouple from the field in $\mathrm{int} (\R^d \setminus \Omega)$ \cite[p.~120]{Guerra1975a}. Thus, there is no correlation between a point in $\Omega_A$ and a point in $\Omega_B$ as the field in $\Omega_A$ cannot influence the field in $\Omega_B$ and vice versa. This should be compared to free boundary conditions. As seen from \eqref{eq:nonlocal_bcs}, free boundary conditions are non-local and in particular, the value of the normal derivative at a point on the boundary is given by a surface integral over $\partial (\Omega_A \cup \Omega_B)$. For example, if $\Omega_A$ and $\Omega_B$ are separated by a finite distance, and we consider $- \Delta + m^2$ on $\Omega_A \cup \Omega_B$ with free boundary conditions, then the value of $\partial f / \partial n$ at a point $\boldsymbol{x} \in \partial \Omega_A$ depends on the value of $f$ on the \emph{whole} boundary $\partial (\Omega_A \cup \Omega_B)$. This is due to the fact that the kernel $k$ in \eqref{eq:nonlocal_bcs} is supported everywhere on the boundary, see also \cite{Guerra1976}. Therefore, for a field theory with free boundary conditions, the two regions $\Omega_A$ and $\Omega_B$ ``communicate'' via the non-local boundary conditions. Hence, the inverse of $- \Delta + m^2$ with free boundary conditions, $\hat{G}_0$, cannot be written as $\hat{G}_A \oplus \hat{G}_B$ as is the case for Dirichlet boundary conditions. Rather, $\hat{G}_0$ also contains components describing cross-correlations, see \eqref{eq:G0_matrix}.

\begin{figure}[t]
    \centering
    \includegraphics[width=0.65\textwidth]{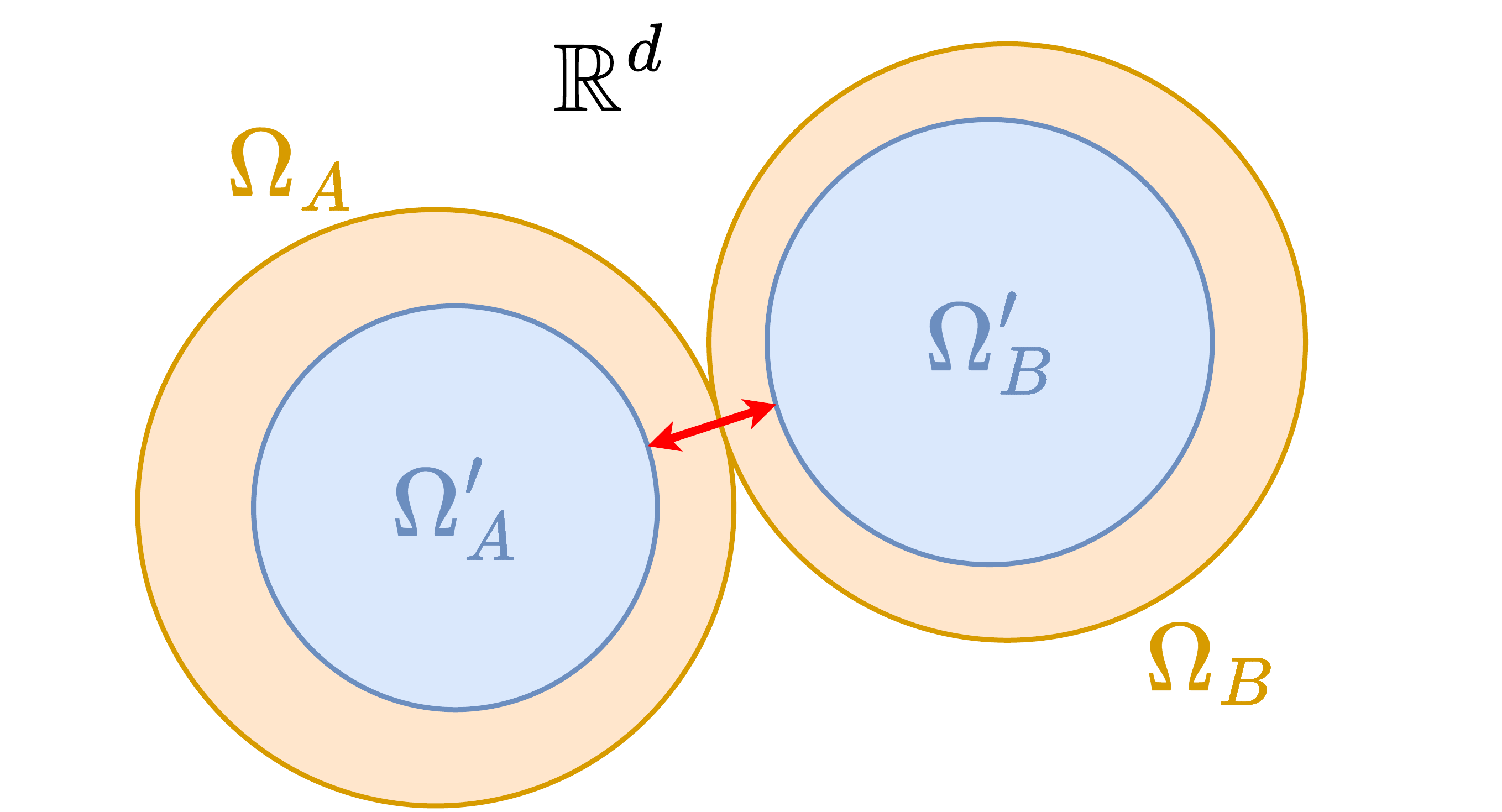}%
    \caption{Sketch of regions considered in Theorem \ref{thm:equiv_mutual_info}. The open subsets $\Omega'_A$ and $\Omega'_B$ are separated by a finite distance, indicated by the red double arrow. Thus, there exist two disjoint (possibly touching) open sets $\Omega_A$ and $\Omega_B$ such that $\mathrm{dist} (\Omega'_A, \partial \Omega_A) > 0$ and $\mathrm{dist} (\Omega'_B, \partial \Omega_B) > 0$. For such $\Omega'_A$ and $\Omega'_B$, the field theory $\mu_{AB}$ with free boundary conditions over $\Omega'_A \cup \Omega'_B$, containing all cross-correlations between $\Omega'_A$ and $\Omega'_B$, is equivalent to the theory $\mu_A \otimes \mu_B$. Therefore, the mutual information $I (\Omega'_A : \Omega'_B)$ for two such regions is finite.}
    \label{fig:region2}
\end{figure}

We can now present the main result of this Section, namely that the mutual information between two disjoint open regions in $\R^d$ is finite if those regions are separated by a finite distance.
\begin{theorem}\label{thm:equiv_mutual_info}
    Let $\Omega'_A$ and $\Omega'_B$ be disjoint bounded open subsets of $\R^d$ such that $\mathrm{dist} (\Omega'_A, \Omega'_B) > 0$. Then, $I (\Omega_A : \Omega_B)$ is finite.
\end{theorem}
\begin{proof}
    We show that the measures $\mu_{AB}$ and $\mu_A \otimes \mu_B$ are equivalent. Since $\mathrm{dist} (\Omega'_A, \Omega'_B) > 0$, there exist disjoint bounded open subsets $\Omega_A$ and $\Omega_B$ of $\R^d$ such that $\mathrm{dist} (\Omega'_A, \partial \Omega_A) > 0$ and $\mathrm{dist} (\Omega'_B, \partial \Omega_B) > 0$, cf. Figure \ref{fig:region2}. Let $\mu^\mathrm{D}_{AB}$ the field theory over $\Omega'_A \cup \Omega'_B$ that is the restriction of the Dirichlet theory over $\Omega_A \cup \Omega_B$ in the sense described in Theorem \ref{thm:equiv_open_Dirichlet}. By Theorem \ref{thm:equiv_open_Dirichlet}, $\mu_{AB} \sim \mu^\mathrm{D}_{AB}$.

    By Proposition \ref{prop:dirichlet_factorisation}, $\hat{G}_\mathrm{D}^{\Omega_A \cup \Omega_B} = \hat{G}_\mathrm{D}^{\Omega_A} \oplus \hat{G}_\mathrm{D}^{\Omega_B}$ and hence $\mu^\mathrm{D}_{AB} = \mu^\mathrm{D}_A \otimes \mu^\mathrm{D}_B$. By employing Theorem \ref{thm:equiv_open_Dirichlet} again, we see that $\mu_A \sim \mu^\mathrm{D}_A$ and $\mu_B \sim \mu^\mathrm{D}_B$, which implies that $\mu_A \otimes \mu_B \sim \mu^\mathrm{D}_A \otimes \mu^\mathrm{D}_B$. Thus $\mu_{AB} \sim \mu^\mathrm{D}_{AB} = \mu^\mathrm{D}_A \otimes \mu^\mathrm{D}_B \sim \mu_A \otimes \mu_B$ and therefore, by transitivity, $\mu_{AB} \sim \mu_A \otimes \mu_B$.
\end{proof}

\begin{figure}[t]
    \centering
    \includegraphics[width=0.60\textwidth]{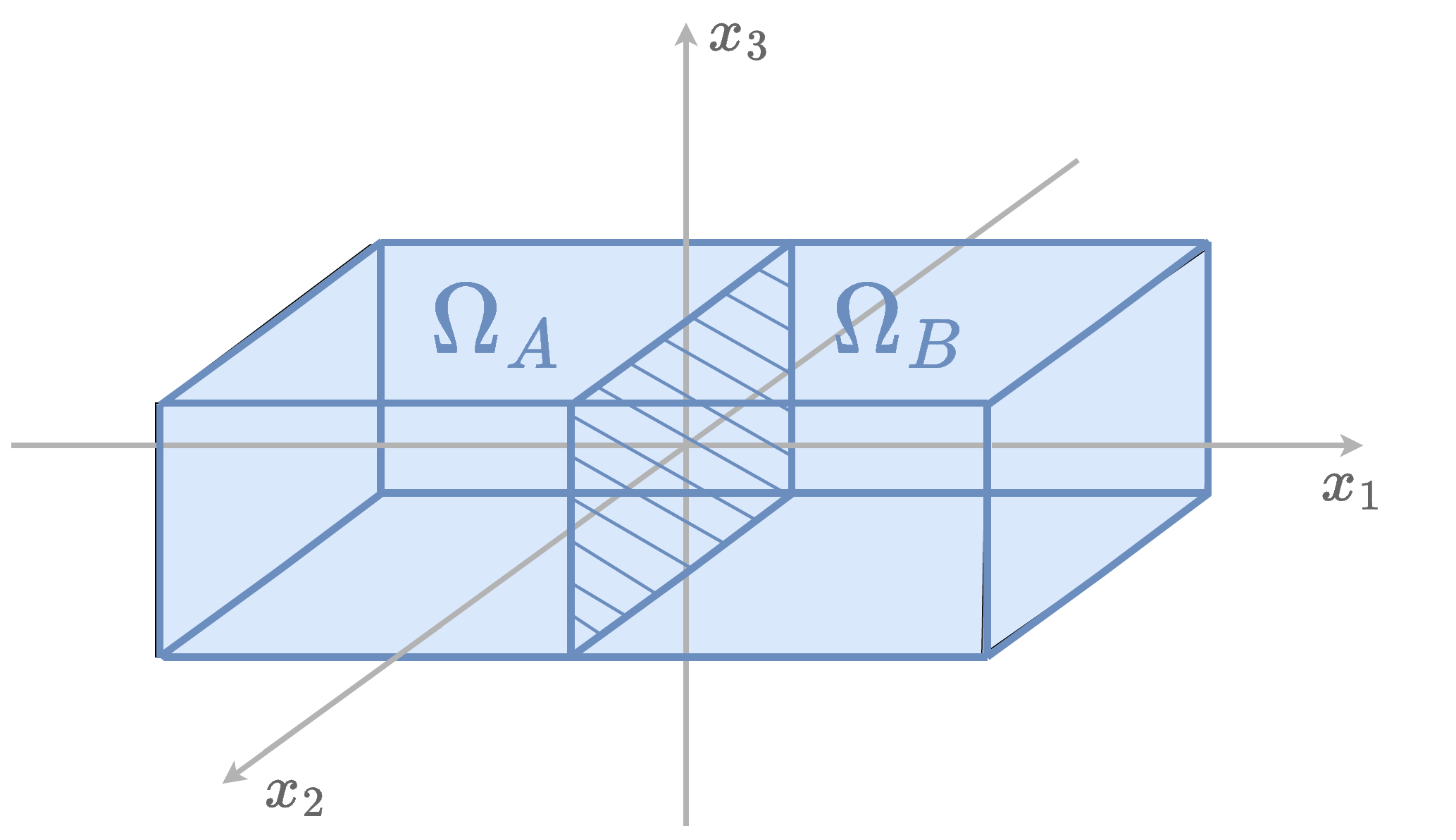}%
    \caption{Sketch of touching rectangular regions in $d = 3$. The two open rectangles $\Omega_A$ and $\Omega_B$ touch in the $x_1 = 0$ plane. Notice that for this configuration $\partial (\Omega_A \cup \Omega_B) \neq \partial ( \overline{\Omega_A \cup \Omega_B})$ as the rectangle in the $x_1 = 0$ plane (hatched region) is a subset of $\partial (\Omega_A \cup \Omega_B)$ but not of $\partial ( \overline{\Omega_A \cup \Omega_B})$. As discussed in the main text, one can think of $\Omega_A$ and $\Omega_B$ ``sharing'' this part of the boundary. Due to the Markov property of free scalar fields, the information on this surface is equal to the information in the whole region and thus the mutual information $I (\Omega_A : \Omega_B)$ is infinite.}
    \label{fig:rectangles}
\end{figure}

From the above Theorem we see that $\mathrm{dist} (\Omega_A, \Omega_B) > 0$ is a sufficient condition for the mutual information $I (\Omega_A : \Omega_B)$ to be finite. We now present an example where two regions $\Omega_A$ and $\Omega_B$ touch, i.e., where the separation distance is zero, and where the corresponding measures are mutually singular.
\begin{theorem}\label{thm:main_5}
    Let $\Omega_A, \Omega_B \subset \R^d$ be touching open $d$-rectangles as sketched in Fig. \ref{fig:rectangles} for $d = 3$. Then, $I (\Omega_A : \Omega_B)$ is infinite.
\end{theorem}
\begin{proof}
    We denote by $\Omega_{AB}$ the interior of $\overline{\Omega_A \cup \Omega_B}$. Notice that $\partial (\Omega_A \cup \Omega_B) \neq \partial \Omega_{AB}$. Let $\mu_{AB}$ be the field theory of mass $m>0$ with free boundary conditions over $\Omega_A \cup \Omega_B$. Recall that the reproducing kernel Hilbert space (RKHS) (cf. Appendix \ref{app:equivalence}) of $\mu_{AB}$ is defined as the closure of $C_0^\infty (\Omega_A \cup \Omega_B)$ in $\|.\|_{-1}$-norm, which is just the Hilbert-Sobolev space $\widetilde{H}^{-1} (\Omega_A \cup \Omega_B)$, for definitions see \cite[Ch.~3]{McLean2000} and \cite{Chandler2017} as well as Appendix \ref{app:sobolev_spaces}. For the specific open sets $\Omega_A$ and $\Omega_B$ considered here, namely the touching open rectangles, we furthermore have
    \begin{equation}
        \widetilde{H}^{-1} (\Omega_{AB}) = H^{-1} (\overline{\Omega_A \cup \Omega_B}) = \widetilde{H}^{-1} (\Omega_A \cup \Omega_B) \; .
    \end{equation}
    The first equality follows from \cite[Thm.~3.29(ii)]{McLean2000} and the fact that $\overline{\Omega_{AB}} = \overline{\Omega_A \cup \Omega_B}$ and the second equality follows from \cite[Lem.~3.17(v)]{Chandler2017} using \cite[Lem.~3.10(vii)]{Chandler2017}. The equality $\widetilde{H}^{-1} (\Omega_{AB}) = \widetilde{H}^{-1} (\Omega_A \cup \Omega_B)$ implies that we may equivalently think of $\mu_{AB}$ as the field theory with open boundary conditions over $\Omega_{AB}$ rather than over $\Omega_A \cup \Omega_B$. Informally speaking, we conclude that it does not matter whether we include $\Gamma = \Omega_{AB} \setminus (\Omega_A \cup \Omega_B)$ (the hatched region in Fig. \ref{fig:rectangles}) in the definition of the field theory.

    The Cameron-Martin space of $\mu_{AB}$, $\mathsf{H}_{\mu_{AB}}$, is given by $\hat{G}_0 [\widetilde{H}^{-1} (\Omega_{AB})]$, cf. \cite[Thm.~3.2.3]{Bogachev2015}. By noticing that we can write $\hat{G}_0 = r_{\Omega_{AB}} \circ D^{-1}$, where $r_{\Omega_{AB}}$ denotes the restriction of a function in $H^{+1} (\R^d)$ to $\Omega_{AB}$, \cite[Lem.~3.2 \& Eq.~(19)]{Chandler2017} imply that $\mathsf{H}_{\mu_{AB}} = H^{+1} (\Omega_{AB})$, see also \cite[Thm.~3.12(iii)]{Chandler2017}. As $\mathsf{H}_{\mu_{AB}}$ and the form domain of $\hat{G}^{-1}_0$ coincide as sets (cf. Lemma \ref{lem:CM_space}), we have $\mathcal{Q} (\hat{G}^{-1}_0) = H^{+1} (\Omega_{AB})$.

    On the other hand, using an analogous argument as above, the form domain of the precision operator $\hat{G}^{-1}_\otimes = \hat{G}^{-1}_A \oplus \hat{G}^{-1}_B$ of the product measure $\mu_A \otimes \mu_B$ is given by
    \begin{equation}
        \mathcal{Q} (\hat{G}^{-1}_\otimes) = \mathcal{Q} (\hat{G}^{-1}_A) \oplus \mathcal{Q} (\hat{G}^{-1}_B) = H^{+1} (\Omega_A) \oplus H^{+1} (\Omega_B) = H^{+1} (\Omega_A \cup \Omega_B) \; ,
    \end{equation}
    where we used \cite[p.~268]{Reed1978}. Suppose for a moment that $\Omega_A$ and $\Omega_B$ do not touch, i.e., there is a finite distance $\varepsilon > 0$ between these rectangles. Then, $\Omega_{AB} = \Omega_A \cup \Omega_B$ and $\mathcal{Q} (\hat{G}^{-1}_0) = \mathcal{Q} (\hat{G}^{-1}_\otimes)$, which is a necessary condition for the equivalence of $\mu_{AB}$ and $\mu_A \otimes \mu_B$, cf. Theorem \ref{thm:equivalence}. This is of course in accordance with the result obtained in Theorem \ref{thm:equiv_mutual_info}, which implies that $\mathcal{Q} (\hat{G}^{-1}_0) = \mathcal{Q} (\hat{G}^{-1}_\otimes)$ if $\Omega_A$ and $\Omega_B$ are separated. If $\Omega_A$ and $\Omega_B$ do touch (in the sense described above), however, then $\Omega_{AB} \neq \Omega_A \cup \Omega_B$ and we need to check whether $H^{+1} (\Omega_{AB})$ and $H^{+1} (\Omega_A \cup \Omega_B)$ coincide as sets.

    We recall that for $\Omega$ an open subset of $\R^d$, every $f \in H^{+1} (\Omega)$ is absolutely continuous on almost all straight lines that are parallel to the coordinate axes, cf. \cite[Thm.~1.1.3/1 \& p.~7]{Mazya2011}. Consider the function that is zero on $\Omega_A$ and one on $\Omega_B$. This function is in $H^{+1} (\Omega_A \cup \Omega_B)$ but not in $H^{+1} (\Omega_{AB})$, as it is not continuous along the $x_1$-direction. Therefore, $H^{+1} (\Omega_{AB}) \neq H^{+1} (\Omega_A \cup \Omega_B)$ and, by Theorem \ref{thm:equivalence}, $\mu_{AB} \perp \mu_A \otimes \mu_B$.
\end{proof}
\begin{remark}
    We emphasize that this result also holds in $d=1$. In particular, let $a,b,c \in \R$ such that $a < b < c$ and let $\Omega_A = (a,b)$, $\Omega_B = (b,c)$ and $\Omega_{AB} = (a,c)$. The Hilbert-Sobolev spaces $H^{+1} ((a,c))$ and $H^{+1} ((a,b) \cup (b,c))$ do not coincide as sets, see also \cite[p.~343]{Arendt2003}. More precisely, a function that is absolutely continuous on $(a,c)$ except at the point $b$ is an element of $H^{+1} (\Omega_A \cup \Omega_B)$ but not of $H^{+1} (\Omega_{AB})$. Therefore, for $\Omega_A = (a,b)$ and $\Omega_B = (b,c)$ we have $I (\Omega_A : \Omega_B) = +\infty$.
\end{remark}

One may be tempted to explain the divergence of the mutual information between two touching regions by the singularity of the fundamental solution $G$ at coinciding Euclidean spacetime points. More precisely, if the regions touch, the correlations between points close to the touching surface are arbitrarily large and this may seem like the cause of the divergence of the mutual information. However, in $d=1$, the fundamental solution is continuous and bounded, cf. \eqref{eq:fund_sol_1d}. Yet, as demonstrated above, the mutual information between two touching open intervals is infinite. Therefore, the analytic properties of the fundamental solution around the diagonal does not seem to be the cause of the divergence of the mutual information. We note that a similar observation has been made in \cite{Nesterov2010} in the context of entanglement entropy of relativistic quantum field theories. In particular, there the authors show that no matter how regular the Green's function is on the diagonal, the entanglement entropy in the corresponding quantum field theory is always UV divergent.

We argue that the divergence of the mutual information is due to the Markov property of the free Euclidean scalar field rather than due to the analytic properties of the Green's function in the vicinity of coinciding Euclidean spacetime points. We will show this in the case $d = 1$, where we can, at least formally, calculate the mutual information explicitly. As a corollary, we find that the mutual information satisfies an \emph{area law} \cite{Wolf2008,Lau2013}, i.e., only the degrees of freedom at the boundaries of the regions contribute to the mutual information. Note, however, that we will not provide a rigorous proof of the area law in this work.

Let $\Omega \subset \R^d$ be open and denote by $\varphi$ the field with free boundary conditions. This random variable can be decomposed as \cite[p.~242]{Guerra1976}
\begin{equation}\label{eq:field_decomposition}
    \varphi (f) = \varphi_\mathrm{D} (f) + \varphi_{\partial} (e_{\partial \Omega} f) \; ,
\end{equation}
where $\varphi_\mathrm{D}$ is the Dirichlet field over $\Omega$ and $\varphi_\partial$ is the field on the boundary $\partial \Omega$ with covariance operator $\hat{G}_0$, called the boundary field. The expression \eqref{eq:field_decomposition} has to be understood as the independent sum of the Gaussian random variables $\varphi_\mathrm{D}$ and $\varphi_\partial$, i.e., the law of $\varphi$ is a product measure of the form $\mu_\mathrm{D} \otimes \mu_\partial$, see \cite[Sec.~III.2]{Guerra1976} and \cite[Prop.~I.7]{Simon1974}.

Note that the decomposition in \eqref{eq:field_decomposition} is a consequence of the Markov property of the free Euclidean scalar field. More precisely, we have, for every $f \in C_0^\infty (\Omega)$ understood as a distribution,
\begin{equation}
    f = \left( p_\Omega + e_{\R^d \setminus \Omega} \right) f = \left( p_\Omega + e_{\overline{\Omega}} \, e_{\R^d \setminus \Omega} \, e_{\overline{\Omega}} \right) f = p_\Omega f + e_{\partial \Omega} f \; ,
\end{equation}
see also \cite[Sec.~II.1]{Guerra1976}. For the second equality, we used the same argument as in the proof of Theorem \ref{thm:equiv_open_Dirichlet}, and for the last equality we used the one-particle (or pre-) Markov property of the Euclidean scalar field, see \cite[Lem.~III.10]{Simon1974} and \cite[Prop.~II.3(i)]{Guerra1975a}. The decomposition in \eqref{eq:field_decomposition} then follows from \cite[Prop.~I.7]{Simon1974}, and the identification of $\varphi (p_\Omega f)$ with the Dirichlet field $\varphi_\mathrm{D} (f)$ follows from Lemma \ref{lem:Dirichlet_cov} and \eqref{eq:inner_product_H-1}.

Let $\Omega_A$ and $\Omega_B$ be open subsets of $\R^d$ separated by a finite distance. We denote by $\varphi^i_\mathrm{D}$ and $\varphi_{\partial i}$ the Dirichlet and boundary field over $\Omega_i$ and $\partial \Omega_i$, respectively, $i \in \{A,B\}$, and by $\varphi^{AB}_\mathrm{D}$ and $\varphi_{\partial AB}$ the Dirichlet and boundary field over $\Omega_A \cup \Omega_B$ and $\partial ( \Omega_A \cup \Omega_B )$. Recall from Section \ref{sec:info_theory} that we can write the mutual information between the two separated regions $\Omega_A$ and $\Omega_B$ as
\begin{equation}\label{eq:MI_entropies_1}
    I (\Omega_A : \Omega_B) = S (\varphi_A) + S (\varphi_B) - S (\varphi_{AB}) \; ,
\end{equation}
where $S (X)$ denotes the entropy of the random variable $X$. Note that this expression is purely formal as the entropies appearing on the right-hand side are not defined when the random variable $X$ is a statistical field $\varphi$. Upon splitting the random variables $\varphi_i$ according to \eqref{eq:field_decomposition} and noting that the entropy is additive for independent random variables, we can formally write
\begin{equation}\label{eq:MI_entropies_2}
    I (\Omega_A : \Omega_B) = S (\varphi^A_\mathrm{D}) + S (\varphi_{\partial A}) + S (\varphi^B_\mathrm{D}) + S (\varphi_{\partial B}) - S (\varphi^{AB}_\mathrm{D}) - S (\varphi_{\partial AB}) \; .
\end{equation}
However, the Dirichlet field $\varphi^{AB}_\mathrm{D}$ factorizes also into two independent Gaussian variables, see Proposition \ref{prop:dirichlet_factorisation} and the subsequent discussion. In particular, $\varphi^{AB}_\mathrm{D} = \varphi^A_\mathrm{D} + \varphi^B_\mathrm{D}$. Thus, the contributions from the Dirichlet fields cancel in \eqref{eq:MI_entropies_2}, and we conclude that for separated regions $\Omega_A$ and $\Omega_B$
\begin{equation}\label{eq:area_law}
    I (\Omega_A : \Omega_B) = S (\varphi_{\partial A}) + S (\varphi_{\partial B})  - S (\varphi_{\partial AB}) = I (\partial \Omega_A : \partial \Omega_B) \; .
\end{equation}
We see that only the boundaries of the two regions $\Omega_A$ and $\Omega_B$ matter for the mutual information, and we call \eqref{eq:area_law} the \emph{area law} for the mutual information, cf. \cite{Wolf2008,Lau2013}. Note that the area law is a direct consequence of the decomposition of the field in \eqref{eq:field_decomposition}, and therefore of the Markov property of scalar Euclidean field.

We now explicitly consider the case $d=1$. Let $a,b,c,d \in \R$ such that $a < b < c < d$ and define $\Omega_A = (a,b)$ and $\Omega_B = (c,d)$. In this case, the laws of the random variables $\varphi_{\partial A}$, $\varphi_{\partial B}$ and $\varphi_{\partial A B}$ are ordinary multivariate Gaussian distributions with vanishing mean and covariance matrices
\begin{equation}
    \Sigma_A = \frac{1}{2 m} \begin{pmatrix}
                    1 & \E^{- m|a-b|} \\
                    \E^{- m|a-b|} & 1
                \end{pmatrix} \; , \qquad \Sigma_B = \frac{1}{2 m} \begin{pmatrix}
                    1 & \E^{- m|c-d|} \\
                    \E^{- m|c-d|} & 1
                \end{pmatrix}
\end{equation}
and
\begin{equation}
    \Sigma_{AB} = \frac{1}{2 m} \begin{pmatrix}
                    1 & \E^{- m|a-b|} & \E^{- m|a-c|} & \E^{- m|a-d|} \\
                    \E^{- m|a-b|} & 1 & \E^{- m|b-c|} & \E^{- m|b-d|} \\
                    \E^{- m|a-c|} & \E^{- m|b-c|} & 1 & \E^{- m|c-d|} \\
                    \E^{- m|a-d|} & \E^{- m|b-d|} & \E^{- m|c-d|} & 1
                \end{pmatrix} \; ,
\end{equation}
respectively. Here, we used the expression of the fundamental solution in $d=1$, \eqref{eq:fund_sol_1d}. Using the expression for the relative entropy between multivariate Gaussian distributions, \eqref{eq:KL_multivariate_Gaussian}, the mutual information between $\Omega_A$ and $\Omega_B$ can be written as
\begin{equation}
    I ((a,b) : (c,d)) = \frac{1}{2} \log \left[ \frac{\det \Sigma_A \, \det \Sigma_B}{\det \Sigma_{AB}} \right] \; .
\end{equation}
The determinants of the covariance matrices can be computed explicitly and we arrive at
\begin{equation}
    I ((a,b) : (c,d)) = - \frac{1}{2} \log \left( 1 - \E^{-2m(c-b)} \right) \; .
\end{equation}

We observe that only the boundary points $b$ and $c$, i.e., the boundary points ``facing each other'', appear in the expression for the mutual information. This once again reflects the Markov property of the scalar field\footnote{Note that in $d=1$ we can interpret a scalar field as a continuous time Markov process and the Markov property implies that ``the future depends on the past only through the present'' \cite{Simon1979}.}, and in particular it is consistent with \cite[Prop.~II.3(ii)]{Guerra1975a}. It also means that the mutual information between two intervals is independent of the length of each of the intervals and only depends on the separation distance $\delta \coloneqq c - b$ between the two intervals. For small $\delta$, the mutual information behaves like $- \frac{1}{2} \log (2 m \delta)$, and we see that the mutual information diverges as $\delta \to 0^+$, which is consistent with our previous result that the mutual information has to be set to $+ \infty$ when the two intervals touch.

Furthermore, we note that, for fixed separation distance $\delta$, the mutual information decreases as the mass $m$ is increased. This is consistent with the interpretation of the mass as the inverse of the correlation length. In particular, as we decrease the correlation length, the mutual information between two finitely separated regions decreases. On the other hand, we see that the mutual information diverges in the limit $m \to 0$. Interpreting the mass again as the inverse of the correlation length, the limit of vanishing mass corresponds to the limit of infinite correlation length, i.e., the system approaches criticality. The divergence of the mutual information in this limit is consistent with previous works on mutual information in classical spin systems, where the divergence of the mutual information at the critical point is used to characterize second order phase transitions \cite{Matsuda1996,Wicks2007,Wilms2011,Wilms2012,Lau2013,Iaconis2013,Stphan2014,Sriluckshmy2018}.

Finally, the need for a finite separation distance between $\Omega_A$ and $\Omega_B$ for finite values of the mutual information is consistent with results obtained in the context of quantum mutual information in relativistic quantum field theories, see, e.g., \cite{Casini2004,Casini2006,Casini2015}.

\section{Conclusion}

In summary, we have studied the properties of relative entropy in the context of Gaussian statistical field theory. More specifically, we have considered the relative entropy between theories with different masses or boundary conditions. In the case of different masses (but equal classical boundary conditions), we showed a crucial dependence on the dimension of Euclidean spacetime: The relative entropy is finite in dimensions $d < 4$ and infinite in higher dimensions. We also showed that in dimensions $d < 4$ the relative entropy behaves in a way that is consistent with its interpretation in terms of the distinguishability of two theories. For situations where the spectrum of eigenvalues of the Laplacian is known, one can write relative entropies as infinite sums involving these eigenvalues. These sums can be conducted analytically or numerically.

In the case of fields with different boundary conditions, we showed that the relative entropy between a Dirichlet and a Neumann field (or, more generally, a Robin field) is always infinite. Whether the relative entropy between two Robin fields is finite depends again on the dimension $d$ of Euclidean spacetime. More specifically, we find that in this case the relative entropy is finite for $d < 3$. We explicitly calculated the relative entropy between a Robin field and a field with free boundary conditions in $d=1$.

Furthermore, we discussed the mutual information between two disjoint regions in Euclidean space. We showed that the mutual information between two such regions is finite if the regions are separated by a finite distance and constructed an example of touching regions between which the mutual information is infinite. We argue that this is due to the Markov property of the scalar field and therefore only the degrees of freedom at the boundary of the region contribute to the mutual information, i.e., the mutual information satisfies an area law. This idea is supported by an explicit calculation of the mutual information in $d = 1$.

An interesting observation concerns the importance of boundary terms for information-theoretic aspects of field theories. For example, the mutual information between degrees of freedom in non-overlapping regions would vanish if Dirichlet boundary conditions (or, more generally, any other local boundary conditions) were chosen everywhere on the boundaries. It is only through less restrictive choices, such as free boundary conditions, that information can be shared between regions.

For all boundary conditions we have investigated, it has been possible to introduce a suitable boundary term to the Euclidean field-theoretic action that realizes them. This boundary term depends only on fields on the boundary, but otherwise contributes to the variational principle and as a weight in the Gaussian measure like the other (volume) terms in the action. At the Gaussian level, this boundary term is specified by a bilinear form with a (possibly non-local) distributional kernel. Many of the information-theoretic properties are determined by the form of this kernel. It would be highly interesting to investigate in more detail how different states are fixed through such boundary terms also beyond the classical boundary conditions we have investigated so far.

We suggest a number of directions for further investigation. An obvious extension of the study of relative entropy in classical field theory presented here is to allow for more general choices of theories. For example, one can consider space-dependent masses, which would lead to the theory of Schr\"{o}dinger operators and is closely related to the possible inclusion of gauge fields and a space-dependent Riemannian metric. Finally, the generalization to interacting theories, i.e. non-Gaussian theories, immediately comes to mind when studying this manuscript. In a functional formulation of Euclidean field theory this is possible, but requires further regularization (in the ultraviolet) and renormalization. For the information-theoretic properties and the propagation of information between regions, the quadratic sector of the theory seems to play the most important role, and we therefore believe that many of the insights gained here will persist, at least on the conceptual and qualitative level, also for interacting theories.

Beyond classical field theory, relative entropies can also be applied to \emph{quantum} field theories. As mentioned in the introduction, this has already been done successfully in the framework of algebraic quantum field theory \cite{Araki1975,Araki1977,Haag2012,Witten2018,Hollands2018,Longo2018,Longo2019,Bostelmann2021,Ciolli2021,Galanda2023}. It would now be interesting to investigate how the functional integral approach can be combined with the algebraic formalism. Functional methods have great practical advantages over the algebraic approach in the sense that they allow one to treat (phenomenologically interesting) interacting theories non-perturbatively, although at the cost of mathematical rigour. For quantum field theories, but also for non-equilibrium statistical field theories, it would furthermore be interesting to study the time evolution of information in more detail. This requires an extension of our setup to a real-time formalism such as the Schwinger-Keldysh double time path or its classical analogue. Boundaries can then be imposed in space and in time. This could be used, for example, to study dynamical processes such as local thermalization in out-of-equilibrium quantum field theories \cite{Dowling2020,Calzetta2023}. In this sense, the present work can also be seen as a step towards a functional treatment of information theory in continuum quantum field theories.

\backmatter

\bmhead{Acknowledgments}

We acknowledge valuable discussions with Razvan Gurau, Tobias Haas, Christian Schmidt, Tim St\"otzel, Mohammadamin Tajik, Andreas Wipf and Jobst Ziebell and are grateful to Tim St\"otzel for helpful suggestions regarding the computation of the closed-form expression of the Dirichlet relative entropy. This work is supported by the Deutsche Forschungsgemeinschaft (DFG, German Research Foundation) under 273811115 -- SFB 1225 ISOQUANT.

\begin{appendices}

    \section{Compact Operators}\label{app:covariance_operators}
    
    In this Section, we collect miscellaneous results used in Section \ref{sec:covariance_operators} for the description of the covariance operators used in this work, which are inverses of self-adjoint realizations of the differential operator $- \Delta + m^2$ \cite[Sec.~4.1]{Grubb2009}. These covariance operators are compact (at least when $\Omega \subset \R^d$ is bounded) self-adjoint operators, and we begin this Section with a brief summary of the properties of such operators. We then state two results used in Section \ref{sec:covariance_operators}.
    
    Let $\mathcal{H}$ be a real or complex separable Hilbert space with inner product $\braket{.,.}$ and norm $\|.\|$. We can make some useful statements about the spectrum of self-adjoint compact operators. The Riesz-Schauder theorem \cite[Thm.~VI.15]{Reed1981} states that the spectrum $\sigma(T)$ of any compact operator $T$ on $\mathcal{H}$ is a discrete set with either no limit points or $\lambda = 0$ as its only limit point. Any non-zero $\lambda \in \sigma (T)$ is an eigenvalue of $T$ of finite multiplicity. If $T$ is, in addition, self-adjoint, then, by the Hilbert-Schmidt theorem \cite[Thm.~VI.16]{Reed1981}, there exists a complete orthonormal basis $\{\phi_n\}_{n=1}^\infty$ of $\mathcal{H}$ such that $T \phi_n = \lambda_n \phi_n$ and $\{\lambda_n\}_{n=1}^\infty$ is a null sequence in $\R$. This also implies that a self-adjoint and compact $T$ is strictly positive precisely when all its eigenvalues are positive. Furthermore, $\min \sigma (T) = 0$ and $\max \sigma (T) = \|T\|$, see \cite[Thm.~VI.6]{Reed1981}. As such a $T$ is both self-adjoint and injective, its inverse $T^{-1}$ is an unbounded densely defined self-adjoint operator on $\mathcal{H}$, see \cite[Prop.~A.8.2]{Taylor2010}.
    
    Two important special cases of compact operators on $\mathcal{H}$ are operators of trace class and Hilbert-Schmidt operators. For any $T \in \mathfrak{B}^+ (\mathcal{H})$, its trace is defined as $\mathrm{Tr} (T) = \sum_{n=1}^\infty \braket{\phi_n, T \phi_n}$, where $\{\phi_n\}_{n=1}^\infty$ is some orthonormal basis of $\mathcal{H}$. An operator $T \in \mathfrak{B} (\mathcal{H})$ is said to be of trace class if and only if $\mathrm{Tr} (|T|) < + \infty$, where $|T| = \sqrt{T^* T}$. We denote the family of all trace class operators on $\mathcal{H}$ by $\mathfrak{T} (\mathcal{H})$. If $T \in \mathfrak{T} (\mathcal{H})$ is self-adjoint, then $\mathrm{Tr} (T) = \sum_{n=1}^\infty \lambda_n$, where $\{\lambda_n\}_{n=1}^\infty$ is the sequence of eigenvalues of $T$. An operator $T \in \mathfrak{B} (\mathcal{H})$ is called Hilbert-Schmidt if and only if $\mathrm{Tr} (T^* T) < + \infty$. We denote the family of all Hilbert-Schmidt operators on $\mathcal{H}$ by $\mathrm{HS} (\mathcal{H})$. Hilbert-Schmidt operators on $L^2$-spaces have a special form. An operator $T \in \mathfrak{B} (L^2 (M, \mu))$, where $(M, \mu)$ is a measure space, is Hilbert-Schmidt precisely when there exists a function $K \in L^2 (M \times M, \mu \otimes \mu)$, called the kernel of $T$, such that $(Tf) (x) = \int K (x,y) f (y) \, \D \mu (y)$ $\mu$-a.e. for all $f \in L^2 (M, \mu)$. In other words, an operator on $L^2 (M, \mu)$ is Hilbert-Schmidt if and only if its kernel is square integrable in the sense that
    \begin{equation}
        \iint_M \left| K (x,y) \right|^2 \, \D\mu (x) \, \D \mu(y) < + \infty \; .
    \end{equation}
    For a useful collection of properties of bounded integral operators on $L^2$-spaces, see, e.g., \cite{Halmos1978}.
    
    As already stated, the covariance operators used in this work are the inverses of self-adjoint extensions of $(- \Delta + m^2)|_{C_0^\infty (\Omega)}$. The following Proposition states that such an operator has an inverse that is defined on all of $\mathcal{H}$.
    \begin{proposition}\label{prop:bounded_below_bijective}
        Let $A \, : \, \mathfrak{D} (A) \to \mathcal{H}$ be a self-adjoint operator that is bounded from below by $c > 0$. Then, $A$ is bijective.
    \end{proposition}
    \begin{proof}
        As $A$ is bounded from below by a positive number, it is strictly positive and hence possesses a unique strictly positive self-adjoint square root \cite{Wouk1966, Bernau1968}, which we denote by $B$. Then, for all $f \in \mathfrak{D} (A)$, $0 \leq \braket{f, A f} = \| B f \|^2$, where equality holds if and only if $f = 0$. This implies that $B f = 0$ if and only if $f = 0$, which means $B$ is injective. As the composition of two injective maps is injective, $A = B^2$ is injective.
        
        Since $\mathrm{ran} (T)^\perp = \mathrm{ker} (T^*)$ for any densely defined $T$ \cite[Prop.~1.6]{Faris1975}, the fact that $A$ is injective implies
        \begin{equation}
            \overline{\mathrm{ran} (A)} = \{ 0 \}^\perp = \mathcal{H} \; ,
        \end{equation}
        i.e., the range of $A$ is dense in $\mathcal{H}$. As $A$ is bounded from below by some $c > 0$, we have $c \|f\|^2 \leq \braket{f, A f} \leq \|f\| \, \|A f\|$, i.e., $\| A f \| \geq c \, \| f \|$ for all $f \in \mathfrak{D} (A)$. Fix $\varphi \in \mathcal{H}$. Then, there exists $\{\varphi_n\}_{n=1}^\infty \subset \mathrm{ran} (A)$ such that $\varphi_n \to \varphi$ in $\mathcal{H}$. Let $\{f_n\}_{n=1}^\infty \subset \mathfrak{D} (A)$ be such that $\varphi_n = A f_n$ for all $n \in \N$. By the semi-boundedness of $A$,
        \begin{equation}
            \| f_n - f_m \| \leq c^{-1} \| A (f_n - f_m) \| = c^{-1} \| \varphi_n - \varphi_m \| \; ,
        \end{equation}
        which implies that $\{f_n\}_{n=1}^\infty$ is Cauchy. We denote the limit of $\{f_n\}_{n=1}^\infty$ by $f$. As $A$ is self-adjoint, it is closed. Thus, its graph $\Gamma (A) = \{ (\psi, A \psi) \, : \, \psi \in \mathfrak{D} (A) \}$ is closed in $\mathcal{H} \oplus \mathcal{H}$ (and hence complete) and $\lim_{n \to \infty} (f_n, A f_n) = (f, \varphi) = (f, Af) \in \Gamma (A)$, which implies that $\varphi \in \mathrm{ran} (A)$. Therefore, $\mathrm{ran} (A) = \mathcal{H}$.
    \end{proof}
    Furthermore, this inverse is bounded.
    \begin{corollary}\label{cor:bounded_inverse}
        Let $A \, : \, \mathfrak{D} (A) \to \mathcal{H}$ be a self-adjoint operator that is bounded from below by $c > 0$. Then, the inverse of $A$ is a strictly positive bounded operator on $\mathcal{H}$ with $\|A^{-1}\| \leq c^{-1}$.
    \end{corollary}
    \begin{proof}
        As $A$ is self-adjoint, $A^{-1} : \mathcal{H} \to \mathfrak{D} (A)$ is also self-adjoint. Then, the Hellinger-Toeplitz Theorem \cite[p.~84]{Reed1981} implies that $A^{-1}$ is bounded. Fix $\psi \in \mathcal{H}$, which can be written as $\psi = A \varphi$ as $\mathcal{H} = \mathrm{ran} (A)$. Then,
        \begin{equation}
            \braket{\psi, A^{-1} \psi} = \braket{A \varphi, \varphi} \geq c \| \varphi \|^2 \geq 0 \; ,
        \end{equation}
        where equality holds if and only if $\varphi = 0$. As $A$ is injective, $\varphi = 0$ if and only if $\psi = 0$. Thus, $A^{-1}$ is strictly positive. Lastly,
        \begin{equation}
            \|A^{-1}\| = \sup_{\psi \neq 0} \frac{\|A^{-1} \psi\|}{\|\psi\|} = \sup_{\varphi \neq 0} \frac{\|\varphi\|}{\|A \varphi\|} = \sup_{\|\varphi\|=1} \left( \|A \varphi\| \right)^{-1} \leq c^{-1} \; .
        \end{equation}
    \end{proof}
    
    \section{Sobolev Spaces}\label{app:sobolev_spaces}

    In this Section, we define the Hilbert-Sobolev spaces needed in the main part of this work. For the most part, we follow the definitions and notations of \cite[Ch.~3]{McLean2000}.

    In the following, let $s \in \R$, $d \in \N$ and $m > 0$. We denote by $\mathscr{S} (\R^d)$ the Schwartz space of functions of rapid decrease. This space equipped with its usual Fr\'echet topology is denoted by $\mathcal{S} (\R^d)$, see, e.g., \cite{Treves2006}. The operator $\mathcal{J}^s : \mathcal{S} (\R^d) \to \mathcal{S} (\R^d)$, called the Bessel potential of order $s$ and mass $m$, defined by
    \begin{equation}
        (\mathcal{J}^s f) (x) = \int_{\R^d} (|\boldsymbol{p}|^2 + m^2)^{\frac{s}{2}} \hat{f} (\boldsymbol{p}) \; \E^{\I \boldsymbol{p} \cdot \boldsymbol{x}} \; \frac{\D^d p}{(2 \pi)^d} \; ,
    \end{equation}
    where $\hat{f}$ denotes the Fourier transform of $f$, is continuous and can be  continuously extended (with respect to the weak-* topology on $\mathcal{S}^* (\R^d))$ to an operator acting on tempered distributions, $\mathcal{J}^s : \mathcal{S}^* (\R^d) \to \mathcal{S}^* (\R^d)$. The Hilbert-Sobolev space $H^s (\R^d)$ of order $s$ on $\R^d$ is defined as the completion of $\mathscr{S} (\R^d)$ with respect to the norm induced by the inner product
    \begin{equation}
        \braket{f,g}_s \coloneqq \braket{\mathcal{J}^s f, \mathcal{J}^s g}_{L^2 (\R^d)} \; .
    \end{equation}
    We have the following scale of Hilbert spaces,
    \begin{equation}
        \ldots \subset H^{+2} (\R^d) \subset H^{+1} (\R^d) \subset H^0 (\R^d) = L^2 (\R^d) \subset H^{-1} (\R^d) \subset H^{-2} (\R^d) \subset \ldots \; ,
    \end{equation}
    where each inclusion $H^s (\R^d) \hookrightarrow H^{s-1} (\R^d)$ is continuous with dense image. We thus see that $H^s (\R^d)$ for $s > 0$ is a subset of $L^2 (\R^d)$ and consists precisely of those square integrable functions that fulfil additional regularity requirements, see the discussion of the spaces $H^s (\Omega)$, $\Omega \subset \R^d$ open, below.  On the other hand, for $s < 0$, elements in $H^s (\R^d)$ need not be square integrable and furthermore need not be functions at all. For example, the Dirac delta distribution is an element of $H^s (\R^d)$ if (and only if) $s < -d / 2$, as can be seen from the definition of the norm in $H^s (\R^d)$ and the Fourier transform of the Dirac delta.

    Let $K$ be a closed subset of $\R^d$. The linear space of distributions with support in $K$,
    \begin{equation}
        H^s_K = \left\lbrace \varphi \in H^s (\R^d) \; : \; \mathrm{supp} \, \varphi \subseteq K \right\rbrace \; ,
    \end{equation}
    is a closed subspace of $H^s (\R^d)$ and a Hilbert space when equipped with the inner product of $H^s (\R^d)$. We denote by $e_K$ the orthogonal projection from $H^s (\R^d)$ onto $H^s_K$.

    Let $\Omega \subset \R^d$ be open. The set of restrictions to $\Omega$ (in the sense of distributions) of elements in $H^s (\R^d)$,
    \begin{equation}
        H^s (\Omega) = \left\lbrace f \in \mathcal{D}^* (\Omega) \, : \, f = F|_\Omega \, , \; F \in H^s (\R^d) \right\rbrace \; ,
    \end{equation}
    is a Hilbert space when equipped with the inner product
    \begin{equation}\label{eq:H(Omega)_inner_product}
        \braket{f,g}_{H^s (\Omega)} \coloneqq \braket{p_\Omega F, p_\Omega G}_s \; ,
    \end{equation}
    where $F,G \in H^s (\R^d)$ such that $f = F|_\Omega$ and $g = G|_\Omega$ and $p_\Omega \coloneqq I - e_{\R^d \setminus \Omega}$. The space $H^s (\Omega)$ is furthermore the completion of $C^\infty (\overline{\Omega})$, the space of restrictions to $\Omega$ of functions in $C^\infty (\R^d)$, in the norm induced by the inner product \eqref{eq:H(Omega)_inner_product}.
    
    Let $\alpha \in \N_0^d$ be a multi-index. We define $\mathsf{D}^\alpha \coloneqq \partial_{x_1}^{\alpha_1} \ldots \partial_{x_d}^{\alpha_d}$, $\partial_{x_j}^{\alpha_j} \coloneqq (\partial / \partial x_j)^{\alpha_j}$, a partial differential operator of order $|\alpha| = \alpha_1 + \ldots + \alpha_d$. If $s \in \N_0$ and $\Omega$ is a Lipschitz domain (see, e.g., \cite[Def.~3.28]{McLean2000}), an equivalent norm on $H^s (\Omega)$ arises from the inner product
    \begin{equation}
        \sum_{|\alpha| \leq s} \braket{\mathsf{D}^\alpha f, \mathsf{D}^\alpha g}_{L^2 (\Omega)} \; ,
    \end{equation}
    cf. \cite[Thm.~3.30]{McLean2000} and \cite[Sec.~1.1.1]{Chandler2017}. In particular, in this case a function $f$ is in $H^s (\Omega)$ if and only if it and all its weak derivatives up to order $s$ are square integrable.
    
    More generally, Hilbert-Sobolev spaces on an open subset $\Omega$ can be defined either \emph{\`a la} Slobodeckij-Gagliardo (with the corresponding space denoted by $W_2^s (\Omega)$) or \emph{\`a la} Bessel-Fourier (with the corresponding space denoted by $H^s (\Omega)$). Whether or not these two definitions lead to the same space depends on $s$ and the regularity of $\Omega$, see \cite[Sec.~1.1.1]{Chandler2017} and references therein. We principally work with spaces $H^s (\Omega)$ but occasionally (when $H^s (\Omega) = W^s_2 (\Omega)$) make use of results from the theory of $W^s_2 (\Omega)$ spaces. Notice that for the case $\Omega = \R^d$, $H^s (\R^d) = W^s_2 (\R^d)$ for all $s \geq 0$ (in the sense that $H^s (\R^d)$ and $W^s_2 (\R^d)$ are norm-equivalent Hilbert spaces), see \cite[Thm.~3.16]{McLean2000}.

    Finally, we define, for every non-empty open $\Omega \subset \R^d$ and $s \in \R$, the space $\widetilde{H}^s (\Omega)$ as the closure of $C^\infty_0 (\Omega)$ in $H^s (\R^d)$ and the space $H^s_0 (\Omega)$ as the closure of $C^\infty_0 (\Omega)$ in $H^s (\Omega)$. The spaces $\widetilde{H}^s (\Omega)$ and $H^s_0 (\Omega)$ are Hilbert spaces when equipped with the inner products of $H^s (\R^d)$ and $H^s (\Omega)$, respectively.

    We conclude this Section with a brief discussion of the trace operator. Let $\Omega$ be a Lipschitz domain with boundary $\partial \Omega$. For every $f \in C^\infty (\overline{\Omega})$ we define the operator $\widetilde{\gamma}$ to be the restriction of $f$ to the boundary, i.e., $\widetilde{\gamma} f = f|_{\partial \Omega}$. This operator can be continuously extended to an operator $\gamma \, : \, H^{+1} (\Omega) \to H^{+ \frac{1}{2}} (\partial \Omega)$, see \cite[Thm.~3.37]{McLean2000}. The operator $\gamma$ is called the trace operator. For definitions of the trace spaces $H^s (\partial \Omega)$ see, e.g., \cite[Sec.~4.2]{Hsiao2021}.

    \section{Quadratic Forms}\label{app:quadratic_forms}
    
    In this Section we provide a brief summary of quadratic forms on Hilbert spaces. We closely follow \cite{Reed1981,Reed1975,Kato1995,Robinson1971}. Let $\mathcal{H}$ be a complex separable Hilbert space with inner product $\braket{.,.}$ and norm $\|.\|$ and $\mathcal{D}$ a dense linear subspace of $\mathcal{H}$. Let $\mathfrak{q} : \mathcal{D} \times \mathcal{D} \to \C$ be a (not necessarily bounded) sesquilinear form. Then $\mathcal{D}$ is called the form domain of $\mathfrak{q}$, and we write $\mathcal{Q} (\mathfrak{q}) = \mathcal{D}$. The sesquilinear form $\mathfrak{q}$ defines a quadratic form from $\mathcal{Q} (\mathfrak{q})$ into $\C$ via $f \mapsto \mathfrak{q} (f) \coloneqq \mathfrak{q} (f,f)$. On the other hand, a quadratic form on a complex Hilbert space defines a sesquilinear form via polarization. We call a sesquilinear form $\mathfrak{q}$ symmetric if $\mathfrak{q}(f,g) = \overline{\mathfrak{q}(g,f)}$ for all $f,g \in \mathcal{Q}(\mathfrak{q})$. Clearly, the quadratic form associated with a symmetric form is real-valued. A symmetric form $\mathfrak{q}$ is called bounded from below if there exists a real number $c$ such that $\mathfrak{q} (f) \geq c \|f\|^2$ for all $f \in \mathcal{Q} (\mathfrak{q})$, in which case we simply write $\mathfrak{q} \geq c$. The largest such number $c$ is called the lower bound of $\mathfrak{q}$. The symmetric form $\mathfrak{q}$ is called positive if $\mathfrak{q} \geq 0$.
    
    A form $\mathfrak{q}$, bounded from below by $c \in \R$, is called closed if the form domain $\mathcal{Q} (\mathfrak{q})$ is complete with respect to the norm
    \begin{equation}
        |\!|\!|f|\!|\!|^2 = \mathfrak{q} (f) + (1-c) \|f\|^2 \; .
    \end{equation}
    If $\mathfrak{q}$ is closed and $S \subset \mathcal{Q} (\mathfrak{q})$ is $|\!|\!|.|\!|\!|$-dense in $\mathcal{Q}(\mathfrak{q})$, we say that $S$ is a form core of $\mathfrak{q}$. A form $\mathfrak{q}$ is called closable if it has a closed extension. The following Theorem captures an important connection between closed forms that are bounded from below and self-adjoint operators.
    \begin{theorem}[{\cite[Thm.~VIII.15]{Reed1981}, \cite[Thm.~VI.2.1 \& Thm.~VI.2.6]{Kato1995}}]\label{thm:first_rep_thm}
        Let $\mathfrak{q}$ be a densely defined closed symmetric form bounded from below with form domain $\mathcal{Q} (\mathfrak{q})$. Then, there exists a unique densely defined self-adjoint operator $T$, bounded from below with the same lower bound as $\mathfrak{q}$, with domain $\mathfrak{D} (T) \subset \mathcal{H}$ such that
        \begin{itemize}
            \item $\mathfrak{D} (T) \subset \mathcal{Q} (\mathfrak{q})$ and $\mathfrak{q} (f,g) = \braket{T f, g}$ for all $f \in \mathfrak{D} (T)$ and $g \in \mathcal{Q} (\mathfrak{q})$,
            \item $\mathfrak{D} (T)$ is a form core of $\mathfrak{q}$,
            \item if $f \in \mathcal{Q} (\mathfrak{q})$, $h \in \mathcal{H}$ and $\mathfrak{q}(f,g) = \braket{h,g}$ for every $g$ belonging to a core of $\mathfrak{q}$, then $f \in \mathfrak{D} (T)$ and $T f = h$.
        \end{itemize}
    \end{theorem}
    We call $T$ the operator associated with the form $\mathfrak{q}$ and call $\mathcal{Q} (T) \coloneqq \mathcal{Q} (\mathfrak{q})$ the form domain of the operator $T$. The case where $\mathfrak{q}$ is positive is of special importance in this work, namely when we consider self-adjoint realizations of the differential operator $- \Delta + m^2$ in Section \ref{sec:covariance_operators}.
    \begin{theorem}[{\cite[Thm.~VI.2.23]{Kato1995}}]\label{thm:second_rep_thm}
        Let $\mathfrak{q}$ be a densely defined, closed, symmetric and positive quadratic form and let $T$ be the positive self-adjoint operator associated with $\mathfrak{q}$. Then, $\mathcal{Q} (T) = \mathfrak{D} (T^{1/2})$ and $\mathfrak{q} (f,g) = \braket{T^{1/2} f, T^{1/2} g}$ for all $f,g \in \mathcal{Q} (\mathfrak{q})$.
    \end{theorem}

    \section{Gaussian Functional Integrals}\label{app:func_integrals}

    This Appendix contains details of the discussion of functional integrals omitted in Section \ref{sec:func_integrals}. We discuss Gaussian measures on infinite dimensional spaces, particularly on Hilbert and locally convex spaces and recall Minlos' Theorem \cite{Minlos1959}, which ensures the existence of Gaussian measures on spaces of distributions with covariance operators relevant for this work.
    
    A free scalar field theory can be defined by the formal expression
    \begin{equation}
        \D \mu = \frac{1}{Z} \, \exp \left[ - S_\mathrm{E} [\varphi] \right] \, \mathcal{D} \varphi \; ,
    \end{equation}
    where $S_\mathrm{E} [\varphi]$ is the Euclidean action functional. As this measure is Gaussian, it is completely specified by a covariance operator and a mean. In statistical field theory, the covariance operator is the inverse of the elliptic differential operator appearing in the free Euclidean action, and its integral kernel is usually referred to as the Green's function, propagator or two-point function. A detailed discussion of the covariance operators needed in this work can be found in Section \ref{sec:covariance_operators}. The mean of the Gaussian measure is interpreted as the expected field configuration. The exposition in this Section is taken primarily from \cite{Bogachev2007,Bogachev2014,Bogachev2015}. We start with some very basic results from probability theory.

    Before we continue, we note that, throughout this Section, every time we consider a Gaussian measure on a locally convex space $X$ it is useful to keep the specific choice employed in the main part in mind, namely a Gaussian measure on the space of distributions $\mathcal{D}^*_\beta (\Omega)$, the strong topological dual of $\mathcal{D} (\Omega)$, the space of test functions equipped with the usual $LF$-topology. For definitions of these spaces see, e.g., \cite{Treves2006} as well as \cite[Ch.~V]{Reed1981}. A similar remark holds for the discussion of nuclear spaces and Minlos' Theorem at the end of this Section. More precisely, the choices $X = \mathcal{D}^*_\beta (\Omega)$, $X^* = \mathcal{D} (\Omega)$, $\mathfrak{X} = \mathcal{D} (\Omega)$ and $\mathfrak{X}^*_\beta = \mathcal{D}^*_\beta (\Omega)$ correspond to the specific spaces used in the main body of this work.
    
    A measure space is a triple $(X, \mathcal{A}, \mu)$, where $\mathcal{A}$ is a $\sigma$-algebra of subsets of a set $X$ and $\mu$ is a measure on $\mathcal{A}$. In case that $\mu$ is a probability measure, i.e., $\mu : \mathcal{A} \to [0,1]$ with $\mu (X) = 1$, the triple $(X, \mathcal{A}, \mu)$ is called a probability space. Let $(X, \mathcal{A}, \mu)$ be a measure space and $(Y,\mathcal{A}')$ a measurable space. A function $f : X \to Y$ is called measurable with respect to the pair $(\mathcal{A}, \mathcal{A}')$ if, for all $A \in \mathcal{A}'$, $f^{-1} (A) \in \mathcal{A}$. It is furthermore called $\mu$-measurable if it is measurable with respect to the pair $(\mathcal{A}, \mathcal{A}')$ and $\mathcal{A}$ is complete with respect to $\mu$. A $\mu$-measurable function induces a measure on $Y$, called the push-forward measure, $f_* \mu$, defined by
    \begin{equation}
        (f_* \mu) (A) = (\mu \circ f^{-1}) (A) = \mu (f^{-1}(A))
    \end{equation}
    for all $A \in \mathcal{A}'$. Measurable functions from a probability space $(X, \mathcal{A}, \mu)$ to some measurable space $(Y,\mathcal{A}')$ are called random variables.
    
    We are interested in $\sigma$-algebras generated by families of sets or functions on $X$. If $S$ is some family of subsets of $X$, then there exists a smallest $\sigma$-algebra in $X$, denoted $\sigma(S)$, containing $S$. Given a family $F$ of functions on $X$, there exists a smallest $\sigma$-algebra of subsets of $X$, denoted $\mathcal{E} (X,F)$, with respect to which all $f \in F$ are measurable.
    
    We now give some important examples of $\sigma$-algebras which we will need in this work. Suppose $X$ is a topological space and let $\mathscr{T}_X$ denote the topology of $X$. Then $\sigma(\mathscr{T}_X) \eqqcolon \mathcal{B} (X)$ is called the Borel $\sigma$-algebra of $X$. A measure that is defined on $\mathcal{B} (X)$ is called a Borel measure on $X$. Now suppose that $X$ is a locally convex space. The cylindrical $\sigma$-algebra generated by the topological dual $X^*$, denoted $\mathcal{E} (X) \coloneqq \mathcal{E} (X, X^*)$, is the smallest $\sigma$-field, with respect to which all continuous linear functionals on $X$ are measurable. As $\mathcal{E} (X)$ is generated by continuous functions, $\mathcal{E} (X) \subseteq \mathcal{B} (X)$ and in some important cases equality holds, see \cite[Thm.~A.3.7 \& Thm.~A.3.8]{Bogachev2015}.
    
    Given a topological space $X$, we can generate a $\sigma$-algebra on it via its topology, which yields the Borel $\sigma$-algebra $\mathcal{B} (X)$. Radon measures are Borel measures that satisfy a collection of natural properties one expects from a measure on a topological space. For the definition of a Radon measure used in this work, see \cite[Def.~A.3.10]{Bogachev2015}. A Radon measure on a locally convex space $X$ is uniquely determined by its restriction to $\mathcal{E} (X)$ \cite[Prop.~A.3.12]{Bogachev2015}. All measures considered in this work will be Radon measures on locally convex spaces, hence many theorems in the remainder of this section will be formulated in terms of Radon measures.
    
    We now consider a special class of probability measures, namely Gaussian measures. We start by recalling some basic facts of Gaussian measures on $\R^n$. A Gaussian measure on $\R$ is a Borel probability measure such that it is either the Dirac measure or its Radon-Nikodym derivative (or density) with respect to the Lebesgue measure is given by
    \begin{equation}
        p (x) = \frac{1}{\sqrt{2 \pi \sigma^2}} \exp \left[ - \frac{(x-a)^2}{2 \sigma^2} \right] \; ,
    \end{equation}
    with $a \in \R$ the mean and $\sigma^2 > 0$ the variance. The case $\sigma^2 = 0$ corresponds formally to the Dirac measure, in which case the measure is also called a degenerate Gaussian measure. If the mean vanishes, we call the Gaussian measure centred. Similarly, a (non-degenerate) Gaussian measure on $\R^n$ is defined as a Borel measure on $\R^n$ with density
    \begin{equation}\label{eq:def_gaussian_Rn}
        p (\boldsymbol{x}) = \frac{1}{\sqrt{\det ( 2 \pi A )}} \exp \left[ - \frac{1}{2} (\boldsymbol{x}-\boldsymbol{a})^\mathsf{T} A^{-1} (\boldsymbol{x}-\boldsymbol{a}) \right] \; ,
    \end{equation}
    where $A$ is a strictly positive symmetric matrix, called the covariance matrix of the Gaussian measure, and $\boldsymbol{a} \in \R^n$ is the mean vector.
    
    So far, we defined Gaussian measures as Borel measures on $\R^n$ whose densities with respect to the Lebesgue measure are (normalized) Gaussian distributions. We can also use the following equivalent definition of Gaussian measures on $\R^n$, which is particularly suitable for generalizations to infinite dimensional spaces. We define the family of Gaussian measures on $\R^n$ to be precisely the family of measures on $\R^n$ such that, for every linear functional\footnote{We identify the dual of $\R^n$ with $\R^n$.} $f$ on $\R^n$, the push-forward measure $f_* \mu$ is a Gaussian measure on $\R$ \cite[Def.~1.2.1]{Bogachev2015}. A (real-valued) random variable $\xi$ on a probability space $(X,\mathcal{A},\mu)$ is called Gaussian if the push-forward measure $\xi_* \mu$ is a Gaussian measure on $\mathcal{B} (\R)$. This leads to the following
    \begin{definition}
        A Borel probability measure $\mu$ on $\R^n$ is called Gaussian if every linear functional on $\R^n$ is a Gaussian random variable with respect to $\mu$.
    \end{definition}
    
    Intuitively, we can think of Gaussian measures on $\R^n$ to be Gaussian ``in every direction''. We can now use the above definition of Gaussian measures on $\R^n$ to generalize the concept of a Gaussian measure to arbitrary linear spaces, including infinite dimensional ones. This general definition does not require any topology on the linear space, see \cite{Bogachev2014,Bogachev2015}. Nevertheless, we will specify on the case of Gaussian measures on locally convex spaces\footnote{In fact, we can always trace the discussion back to Gaussian measures on locally convex spaces by introducing suitable locally convex topologies, see \cite[p.~42]{Bogachev2015}.}. Let $X$ be a locally convex space and $X^*$ its topological dual.
    \begin{definition}[{\cite[Def.~2.2.1]{Bogachev2015}}]
        Let $X$ be a locally convex space. A probability measure $\mu$ on the measurable space $(X,\mathcal{E}(X))$ is called Gaussian if every $f \in X^*$ is a Gaussian random variable. It is called a centred Gaussian measure if all $f \in X^*$ are centred Gaussian random variables.
    \end{definition}
    This definition can be restated in the language of random processes, see, e.g., \cite{Lifshits2012}. A probability measure $\mu$ on $(X,\mathcal{E} (X))$ is Gaussian if the random process on $(X, \mathcal{E} (X), \mu)$ indexed by $X^*$, $\{\varphi_f\}_{f \in X^*}$, where $\varphi_f = f (\varphi)$, is Gaussian \cite{Bogachev2014}. Recalling that a Radon measure is uniquely determined by its restriction to $\mathcal{E} (X)$, a Radon measure $\mu$ on a locally convex space $X$ is called Gaussian if its restriction to $\mathcal{E} (X)$ is Gaussian \cite[Def.~3.1.1]{Bogachev2015}.
    \begin{definition}[{\cite[Def.~2.2.7]{Bogachev2015}, see also \cite[Thm.~3.2.3]{Bogachev2015}}]\label{def:mean_covariance}
        Let $\mu$ be a Radon Gaussian measure on a locally convex space $X$. The mean of $\mu$, denoted $a_\mu$, is an element of $X$ and is defined by
        \begin{equation}
            a_\mu (f) = \EE [f] = \int_X f (\varphi) \; \D \mu (\varphi) \; , \qquad f \in X^* \; .
        \end{equation}
        The covariance operator of $\mu$, denoted $R_\mu$, is a linear map $R_\mu : X^* \to X$, defined by
        \begin{align*}
            R_\mu (f) (g) &= \EE [(f-a_\mu(f))(g-a_\mu(g))] \\
            &= \int_X \left( f (\varphi) - a_\mu (f) \right) \left( g (\varphi) - a_\mu (g) \right) \D \mu (\varphi) \; , \qquad f, g \in X^* \; . \numberthis
        \end{align*}
        The covariance operator $R_\mu$ induces the symmetric bilinear form $\mathrm{Cov}$ on $X^* \times X^*$ via $\mathrm{Cov} (f,g) \coloneqq R_\mu (f) (g)$ for all $f,g \in X^*$. The corresponding quadratic form, called the covariance of $\mu$, is positive, i.e., $\mathrm{Cov} (f,f) \geq 0$ for all $f \in X^*$.
    \end{definition}
    
    A Borel measure $\mu$ on $\R^n$ is Gaussian precisely when its Fourier transform is given by \cite[Prop.~1.2.2]{Bogachev2015}
    \begin{equation}
        \hat{\mu} (\boldsymbol{p}) \coloneqq \int_{\R^n} \exp \left[ \I \; \boldsymbol{p} \cdot \boldsymbol{x} \right] \; \D \mu (\boldsymbol{x}) = \exp \left[ - \frac{1}{2} \, \boldsymbol{p}^\mathsf{T} A \boldsymbol{p} + \I \; \boldsymbol{a} \cdot \boldsymbol{p} \right] \; ,
    \end{equation}
    with $\boldsymbol{a} \in \R^n$ the mean and $A$ the positive symmetric covariance matrix. An analogous statement also holds for measures on infinite dimensional spaces. First, however, we need to define the Fourier transform of such measures. Let $X$ be a locally convex space and let $\mu$ be a measure on $\mathcal{E} (X)$. Following \cite[Def.~A.3.17]{Bogachev2015}, the Fourier transform of $\mu$, denoted $\hat{\mu}$, is a map $\hat{\mu} \, : \, X^* \to \C$ defined by
        \begin{equation}
            \hat{\mu} (f) = \int_X \exp \left[ \I f (\varphi) \right] \; \D \mu (\varphi) \; , \qquad f \in X^* \; .
        \end{equation}
    A Gaussian measure on a locally convex space $X$ is fully characterised by its Fourier transform.
    \begin{theorem}[{\cite[Thm.~2.2.4]{Bogachev2015}}]\label{th:FT_Gaussian_measure}
        A Radon measure $\mu$ on a locally convex space $X$ is Gaussian precisely when its Fourier transform is given by
        \begin{equation}\label{eq:FT_Gaussian_measure}
            \hat{\mu} (f) = \exp \left[ - \frac{1}{2} B (f,f) + \I m (f) \right] \; .
        \end{equation}
        Here, $m$ is a linear functional on $X^*$ and $B$ is a symmetric bilinear form on $X^* \times X^*$ such that the corresponding quadratic form is positive, i.e., $B (f,f) \geq 0$ for all $f \in X^*$.
    \end{theorem}
    \begin{remark}
        Following Definition \ref{def:mean_covariance}, we identify $B$ with the covariance $\mathrm{Cov}$ and $m$ with the mean $a_\mu$, see also \cite[Prop.~2.2]{Bogachev2014} and \cite[Thm.~2.2.4]{Bogachev2015}.
    \end{remark}
    
    As an important special case, let us consider Gaussian measures on separable Hilbert spaces, $X = \mathcal{H}$. The following theorem establishes a one-to-one correspondence between pairs $(m,\hat{C})$, $m \in \mathcal{H}$, $\hat{C} \in \mathfrak{T} (\mathcal{H})$ self-adjoint and positive, and Gaussian measures on $\mathcal{H}$.
    \begin{theorem}[Mourier \cite{Mourier1953}]\label{th:Mourier}
        Let $\mu$ be a Gaussian measure on a separable Hilbert space $\mathcal{H}$. Then, there exists an element $m \in \mathcal{H}$, called the mean of $\mu$, and a positive self-adjoint trace class operator $\hat{C} : \mathcal{H} \to \mathcal{H}$, called the covariance operator of $\mu$, such that, for all $f,g \in \mathcal{H}$,
        \begin{align}
            m (f) &= \braket{m,f}_\mathcal{H} = \int_\mathcal{H} \varphi(f) \; \D \mu (\varphi) \; , \\
            \mathrm{Cov} (f,g) &= \braket{f, \hat{C} g}_\mathcal{H} = \int_\mathcal{H} \left( \varphi (f) - m (f) \right) \, \left( \varphi (g) - m (g) \right) \; \D \mu (\varphi) \; .
        \end{align}
        Furthermore, $\mu$ is uniquely determined by $m$ and $\hat{C}$.
        
        Conversely, for every pair of $m \in \mathcal{H}$ and positive self-adjoint trace class operator $\hat{C}$, there exists a unique Gaussian measure $\mu$ on $\mathcal{H}$ such that $m$ is the mean of $\mu$ and $\hat{C}$ is its covariance operator.
    \end{theorem}
    For a proof of this theorem, see, e.g., \cite[Theorem~2.3.1]{Bogachev2015} or \cite[Theorem~IV.2.4]{Vakhania1987}, as well as \cite[Sec.~2.3]{DaPrato2014}.
    
    The characterisation of a Gaussian measure in terms of its Fourier transform, as described in Theorem \ref{th:FT_Gaussian_measure} seems practical, and it would be convenient to define our field theory by writing down the Fourier transform of a Gaussian measure. However, there is, \emph{a priori}, no guarantee that a functional of the form \eqref{eq:FT_Gaussian_measure} is the Fourier transform of a Gaussian measure. That this is indeed the case for certain functionals on \emph{nuclear} spaces is the content of Minlos' theorem \cite{Minlos1959}, which we shall briefly discuss in the following.
    
    Minlos' Theorem is a generalisation of Bochner's Theorem \cite[Thm.~IX.9]{Reed1975}. Bochner's Theorem establishes a one-to-one correspondence between probability measures on $\R$ and continuous positive-definite functions $f$ on $\R$ such that $f(0)=1$. We first need the following
    \begin{definition}[{\cite[Def.~1.5.1]{Obata1994}}]
        Let $\mathfrak{X}$ be a nuclear space. A function $\mathcal{C} \, : \, \mathfrak{X} \to \C$ is called a characteristic functional if \emph{(i)} $\mathcal{C}$ is continuous, \emph{(ii)} $\mathcal{C}$ is of positive type, i.e., for all $n \in \N$, $\alpha_1, \ldots, \alpha_n \in \C$ and $f_1, \ldots, f_n \in \mathfrak{X}$, $\sum_{i,j = 1}^n \overline{\alpha_i} \alpha_j \, \mathcal{C} (f_i-f_j) \geq 0$ and \emph{(iii)} $\mathcal{C}$ is normalised, i.e., $\mathcal{C} (0) = 1$.
    \end{definition}
    We now state a version of Minlos' theorem that is particularly useful for this work.
    \begin{theorem}[{Minlos \cite{Minlos1959}, \cite[Thm.~4.3,~Thm.~4.4]{Vakhania1987}}]\label{th:Minlos}
        Let $\mathfrak{X}$ be a nuclear space. Every characteristic functional on $\mathfrak{X}$ is the Fourier transform of a Radon probability measure on $\mathfrak{X}^*_\beta$.
    \end{theorem}
    Minlos' Theorem provides a practical way to construct a Gaussian field theory. More precisely, we will construct characteristic functionals of the form \eqref{eq:FT_Gaussian_measure} on a nuclear space and Minlos' Theorem then assures the existence of a Gaussian measure with desired mean and covariance on its topological dual. We first need the following result.
    \begin{lemma}\label{lem:characteristic_functional}
        Let $\mathcal{H}$ be a real Hilbert space with inner product $\braket{.,.}$ and $\mathfrak{X}$ a nuclear space such that the inclusion $\iota : \mathfrak{X} \hookrightarrow \mathcal{H}$ is continuous. Suppose $m \in \mathfrak{X}^*$ and $T \in \mathfrak{B} (\mathcal{H})$ is self-adjoint and positive. Then, the functional $\mathcal{C} : \mathfrak{X} \to \C$ defined by
        \begin{equation}\label{eq:characteristic_functional}
            \mathcal{C} (f) = \exp \left[ - \frac{1}{2} \braket{T (\iota f), \iota f} + \I m (f) \right] \; , \qquad f \in \mathfrak{X} \; ,
        \end{equation}
        is a characteristic functional.
    \end{lemma}
    \begin{proof}
        It is clear that $\mathcal{C}$ normalised and continuous. We still need to show that $\mathcal{C}$ is of positive type. Fix $n \in \N$, $\alpha_1, \ldots, \alpha_n \in \C$ and $f_1, \ldots, f_n \in \mathfrak{X}$. Then,
        \begin{equation}\label{eq:characteristic_functional_proof}
            \sum_{i,j = 1}^n \overline{\alpha_i} \alpha_j \, \mathcal{C} (f_i-f_j) = \sum_{i,j = 1}^n \overline{\beta_i} \beta_j \, \exp \left[ A_{ij} \right] \; ,
        \end{equation}
        where $\beta_i = \alpha_i \, \overline{\mathcal{C} (f_i)} \in \C$ and $A_{ij} = \braket{f_i, T f_j}$ (we suppress the inclusion $\iota$ for the remainder of the proof). Let $\mathcal{H}^\C$ be the complexification of $\mathcal{H}$ and denote by $\braket{.,.}_\C$ the canonical inner product in $\mathcal{H}^\C$ \cite{Sharma1988}. Furthermore, let $T^\C$ be the complexification of $T$. We now can write, for all $z \in \C^n$,
        \begin{align*}
            \sum_{i, j = 1}^n \overline{z_i} z_j A_{ij} &= \sum_{i, j = 1}^n \overline{z_i} z_j \braket{f_i, T f_j} \\
            &= \sum_{i, j = 1}^n \braket{z_i f_i, T^\C (z_j f_j)}_\C \\
            &= \left\langle \sum_{i = 1}^n z_i f_i , T^\C \left( \sum_{j = 1}^n z_j f_j \right) \right\rangle_\C \; . \numberthis
        \end{align*}
        By assumption, $T$ is positive and so is $T^\C$. By defining $h = \sum_{i = 1}^n z_i f_i \in \mathcal{H}^\C$, we see that
        \begin{equation}
            \sum_{i, j = 1}^n \overline{z_i} z_j A_{ij} = \braket{h, T^\C h}_\C \geq 0 \; .
        \end{equation}
        
        Therefore, $A$ is a positive matrix. The Schur product theorem \cite[Thm.~VII]{Schur1911} states that the element-wise product of two positive matrices is again a positive matrix. Hence, using the power series representation of the exponential function, we see that the matrix $M_{ij} = \exp [A_{ij}]$ is positive, i.e.,
        \begin{equation}
            \sum_{i,j = 1}^n \overline{\alpha_i} \alpha_j \, \mathcal{C} (f_i-f_j) = \sum_{i,j = 1}^n \overline{\beta_i} \beta_j \, \exp \left[ A_{ij} \right] \geq 0 \; .
        \end{equation}
        Therefore, $\mathcal{C}$ is of positive type.
    \end{proof}

    By Minlos' theorem, a characteristic functional of the form \eqref{eq:characteristic_functional} is the Fourier transform of a Radon probability measure on $\mathfrak{X}^*_\beta$. Due to the self-adjointness and positivity of $T$, the bilinear form $\braket{\iota ., T (\iota .)}$ on $\mathfrak{X} \times \mathfrak{X}$ is symmetric and the corresponding quadratic form is positive. Furthermore, the functional $m$ on $\mathfrak{X}$ is linear. Thus, by Theorem \ref{th:FT_Gaussian_measure}, a characteristic functional of the form \eqref{eq:characteristic_functional} is the Fourier transform of a Radon \emph{Gaussian} measure on $\mathfrak{X}^*_\beta$. The covariance operator of such a Gaussian measure $\mu$ is given by\footnote{Note that for the case we consider in the main part of the work, i.e., $\mathfrak{X} = \mathcal{D} (\Omega)$, $R_\mu$ is a continuous linear map from $\mathcal{D} (\Omega)$ to $\mathcal{D}^*_\beta (\Omega)$. Therefore, by the Schwartz Kernel Theorem \cite[Thm.~51.7]{Treves2006}, $R_\mu$ has a distributional kernel, which is of course just the Green's function of the operator $- \Delta + m^2$ obeying some boundary conditions.} $R_\mu = {^\mathrm{t} \iota} \circ T \circ \iota : \mathfrak{X} \to \mathfrak{X}^*$, $f \mapsto T (\iota f) \circ \iota$, where $T (\iota f)$ is understood as a continuous linear functional on $\mathcal{H}$. In particular, the covariance induced by $R_\mu$ is given by
    \begin{equation}
        \mathrm{Cov} (f,f) = R_\mu (f) (f) = \braket{T (\iota f), \iota f} \; , \qquad f \in \mathfrak{X} \; .
    \end{equation}
    The mean of $\mu$ is given by $a_\mu = m \in \mathfrak{X}^*$.
    
    For Gaussian measures which are defined by characteristic functionals of the form \eqref{eq:characteristic_functional}, we will, for simplicity, call the self-adjoint and positive operator $T \in \mathfrak{B} (\mathcal{H})$ its covariance operator. Furthermore, we call $T^{-1}$ (if it exists) the precision operator of the measure $\mu$. Throughout this paper, following the notation for Gaussian measures on $\R^n$, we denote a Gaussian measure which is defined by the characteristic functional \eqref{eq:characteristic_functional} by $\mathcal{N} (m,T)$.
    
    To summarize this Section, we have shown that the functional \eqref{eq:characteristic_functional} is a characteristic functional. Minlos' Theorem then tells us that this characteristic functional is the Fourier transform of a Radon probability measure on the strong topological dual $\mathfrak{X}^*_\beta$. Furthermore, as the operator $T$ is positive and self-adjoint, this probability measure is, by Theorem \ref{th:FT_Gaussian_measure}, a Gaussian measure. Thus, we can construct centred Gaussian measures by specifying suitable covariance operators on a real Hilbert space. We demonstrate this procedure by constructing the Gaussian measures needed in the main part of this work. Let $\mathfrak{X} = \mathcal{D} (\Omega)$, the $LF$-space of test functions supported in a bounded open subset $\Omega \subset \R^d$. Furthermore, choose $\mathcal{H} = L^2 (\Omega)$ and let $T = \hat{G}_\mathrm{X}$ be the inverse of some self-adjoint extensions of $(-\Delta + m^2)|_{C_0^\infty (\Omega)}$ corresponding to the choice of ``$\mathrm{X}$''-boundary conditions on $\partial \Omega$, for details see Section \ref{sec:covariance_operators}. By the above discussion, there exists a centred Radon Gaussian measure $\mu = \mathcal{N} (0, \hat{G}_\mathrm{X})$ on $\mathcal{D}^*_\beta (\Omega)$ with covariance
    \begin{equation}
        \mathrm{Cov} (f,f) = \braket{\hat{G}_\mathrm{X} (\iota f), \iota f}_{L^2 (\Omega)} \; , \qquad f \in \mathcal{D} (\Omega) \; .
    \end{equation}
    We call $\mathcal{N} (0, \hat{G}_\mathrm{X})$ the free scalar field theory of mass $m$ over $\Omega$ with ``$\mathrm{X}$''-boundary conditions.
    
    \section{Equivalence of Gaussian Measures}\label{app:equivalence}
    
    In this Section, we formulate necessary and sufficient conditions for two centred Gaussian measures $\mu = \mathcal{N} (0,\hat{C}_\mu)$ and $\nu = \mathcal{N} (0,\hat{C}_\nu)$ to be equivalent. Recall from Section \ref{sec:info_theory} that two measures are called equivalent if they have the same null sets or -- in the language of probability theory -- have the same set of impossible events. A Theorem by Feldman and H\'ajek (Theorem \ref{th:FeldmanHajek} below) states that two Gaussian measures are either equivalent of mutually singular. The concepts of equivalence and mutual singularity of Gaussian measures are important for our work as the relative entropy between two Gaussian measures is finite precisely when the measures are equivalent and is defined to be $+ \infty$ otherwise. An interpretation of this result in terms of the distinguishability of two probability measures is outlined in Section \ref{sec:info_theory}.

    Throughout this Section we make two assumptions. First, we assume we have a double $(\mathfrak{X}, \mathcal{H})$, where $\mathcal{H}$ is a real Hilbert space with inner product $\braket{.,.}_\mathcal{H}$ and norm $\|.\|_\mathcal{H}$ and $\mathfrak{X} \subset \mathcal{H}$ is a nuclear space such that the inclusion $\iota : \mathfrak{X} \hookrightarrow \mathcal{H}$ is continuous with dense image. The choice we need in the main part of this paper, namely $\mathcal{H} = L^2 (\Omega)$ and $\mathfrak{X} = \mathcal{D} (\Omega)$ for some bounded open $\Omega \subset \R^d$, fulfils this assumption. Secondly, we assume $\hat{C} : \mathcal{H} \to \mathcal{H}$ is a self-adjoint, strictly positive (hence injective) and compact operator. Furthermore, its inverse $\hat{C}^{-1}$ is an unbounded, densely defined self-adjoint operator that is bounded from below by some $c > 0$. As discussed in Section \ref{sec:covariance_operators}, the operators $\hat{G}_\mathrm{X}$ on $L^2 (\Omega)$, which are the compact inverses of self-adjoint extensions of $(-\Delta + m^2)|_{C_0^\infty (\Omega)}$ corresponding to ``$\mathrm{X}$''-boundary conditions, have this property.
    
    We recall Lemma \ref{lem:characteristic_functional}, which states that, for every $m \in \mathfrak{X}^*$,
    \begin{equation}\label{eq:charfunc}
            \mathcal{C} (f) = \exp \left[ - \frac{1}{2} \braket{\hat{C} (\iota f), \iota f} + \I m (f) \right] \; , \qquad f \in \mathfrak{X} \; ,
        \end{equation}
    is a characteristic functional. Hence, by Minlos' Theorem (Theorem \ref{th:Minlos}), $\mathcal{C} (f)$ is the Fourier transform of a Radon Gaussian measure $\mu = \mathcal{N} (m, \hat{C})$ on $\mathfrak{X}^*_\beta$.
    
    The following result from Gaussian measure theory is central to our work.
    \begin{theorem}[{Feldman-H\'ajek \cite{Feldman1958,Hajek1958}, \cite[Thm.~2.7.2]{Bogachev2015}}]\label{th:FeldmanHajek}
        Let $X$ be a locally convex space and $\mu$ and $\nu$ two Gaussian measures on $X$. Then $\mu$ and $\nu$ are either equivalent or mutually singular.
    \end{theorem}
    
    Given a centred Radon Gaussian measure $\mu$ on a locally convex space $X$, we define the Hilbert space $L^2 (\mu) \coloneqq L^2 (X, \mu)$ of equivalence classes of $\mu$-measurable functions which are square-integrable with respect to $\mu$ in the usual way. By the definition of a Gaussian measure, for every $f \in X^*$, the push-forward measure $f_* \mu$ is a Gaussian measure on $\R$. By the change of variables formula \cite[Eq.~(0.1)]{Bogachev2014}, we see that
    \begin{equation}
        \| f \|_{L^2 (\mu)}^2 = \int_X \left( f (\varphi) \right)^2 \, \D \mu (\varphi) = \int_\R x^2 \; \D(f_* \mu) (x) < + \infty \; .
    \end{equation}
    Thus, every (equivalence class of) $f \in X^*$ is also an element of $L^2 (\mu)$. For the discussion of the equivalence of two Gaussian measures as well as for the derivation of an explicit formula for the relative entropy, we need two Hilbert spaces which characterize a Gaussian measure. These are the reproducing kernel Hilbert space and the Cameron-Martin space, which we introduce now.
    \begin{definition}[{\cite[Ch.~5]{Bogachev2014}, \cite[Sec.~2.2]{Bogachev2015}}]
        Let $\mu$ be a Radon Gaussian measure on a locally convex space $X$. The closure of the set $X^*$ with respect to the norm of $L^2 (\mu)$ and equipped with the $L^2 (\mu)$-inner product is called the \emph{reproducing kernel Hilbert space (RKHS)} of the measure $\mu$ and is denoted by $X^*_\mu$. The elements of $X^*_\mu$ are called the $\mu$-measurable linear functionals. 
    \end{definition}
    For the specific example of the centred Gaussian measure $\mu = \mathcal{N} (0, \hat{C})$, we note that, for any $f, g \in \mathfrak{X}$,
    \begin{equation}
            \braket{ f, g }_{L^2(\mu)} = \int_{\mathfrak{X}^*_\beta} f (\varphi) g (\varphi) \; \D \mu (\varphi) = R_\mu (f) (g) = \braket{\hat{C}^{1/2} f, \hat{C}^{1/2} g}_\mathcal{H} \; ,
    \end{equation}
    where we suppressed the inclusion of $\mathfrak{X}$ into $\mathcal{H}$. The RKHS of the Gaussian measure $\mu = \mathcal{N} (0, \hat{C})$, denoted $\mathfrak{X}_\mu$, is then the closure of the set $\mathfrak{X}$ in the norm $\|\hat{C}^{1/2} .\|_\mathcal{H}$ equipped with the inner product $\braket{\hat{C}^{1/2} ., \hat{C}^{1/2} .}_\mathcal{H}$.
    
    \begin{definition}[{cf. \cite[Def.~3.24]{Hairer2009}}]
        Let $\mu$ be a Radon Gaussian measure on a locally convex space $X$. Consider the set
        \begin{equation}
            \mathring{\mathsf{H}}_\mu \coloneqq \left\{ \varphi \in X \, : \, \exists \, \hat{\varphi} \in X^* \text{ such that } \varphi = R_\mu (\hat{\varphi}) \right\} = R_\mu [ X^* ] \subset X \; ,
        \end{equation}
        i.e., the image of $X^*$ under the map $R_\mu$. We define on $\mathring{\mathsf{H}}_\mu \times \mathring{\mathsf{H}}_\mu$ the inner product $\llangle \varphi, \psi \rrangle \coloneqq R_\mu (\hat{\varphi}) (\hat{\psi})$ and norm $\| \varphi \|_\mu^2 = \llangle \varphi, \varphi \rrangle$. The completion of $\mathring{\mathsf{H}}_\mu$ with respect to $\| . \|_\mu$, equipped with the inner product $\llangle ., . \rrangle$, is called the \emph{Cameron-Martin space (CMS)} of $\mu$ and is denoted by $\mathsf{H}_\mu$. One can show that $\mathsf{H}_\mu = R_\mu [X^*_\mu]$ \cite[Thm.~3.2.3]{Bogachev2015}.
    \end{definition}
    For the specific example of the Gaussian measure $\mu = \mathcal{N} (0, \hat{C})$, we note that $R_\mu = {^\mathrm{t} \iota} \circ \hat{C} \circ \iota$. In the following, we will suppress the inclusion $\iota$ of $\mathfrak{X}$ into $\mathcal{H}$. For any $\varphi, \psi \in \hat{C} [\mathfrak{X}]$, we define the bilinear form $ (\varphi, \psi) \mapsto \llangle \varphi, \psi \rrangle$ by
    \begin{equation}
        \llangle \varphi, \psi \rrangle = \braket{\hat{C}^{-1/2} \varphi , \hat{C}^{-1/2} \psi}_\mathcal{H} = \braket{\hat{C} \hat{\varphi} , \hat{\psi}}_\mathcal{H} = R_\mu (\hat{\varphi}) (\hat{\psi}) \; .
    \end{equation}
     The CMS $\mathsf{H}_\mu$ of the Gaussian measure $\mu = \mathcal{N} (0, \hat{C})$ is thus the closure of $\mathring{\mathsf{H}}_\mu = \hat{C} [\mathfrak{X}]$ with respect to the norm $\|.\|^2_\mu = \llangle .,. \rrangle = \braket{\hat{C}^{-1/2} ., \hat{C}^{-1/2} .}_\mathcal{H} = \|\hat{C}^{-1/2} .\|_\mathcal{H}^2$, equipped with the inner product $\llangle .,. \rrangle$.

     Before we continue with the derivation of conditions for equivalence of Gaussian measures, it is instructive to consider explicit examples of reproducing kernel Hilbert spaces and Cameron-Martin spaces encountered in the main part of this work. Specifically, consider the centred Gaussian measures $\mu = \mathcal{N} (0,\hat{G}_\mathrm{X})$, where $\mathrm{X} \in \{0, \mathrm{D}\}$, for definitions see Section \ref{sec:covariance_operators}. In other words, we study the RKHS and CMS of a free scalar field theory of mass $m > 0$ with free or Dirichlet boundary conditions.

     We start with free boundary conditions. By definition, the RKHS of $\mu_0 = \mathcal{N} (0,\hat{G}_0)$, $\mathfrak{X}_{\mu_0}$, is the closure of $C_0^\infty (\Omega)$ in the norm $\|\hat{G}_0^{1/2} . \|_{L^2 (\Omega)}$. But for every $f \in C_0^\infty (\Omega)$,
     \begin{equation}
        \|\hat{G}_0^{1/2} f \|_{L^2 (\Omega)} = \|D^{-1/2} f \|_{L^2 (\R^d)} \; ,
     \end{equation}
     where $D$ is the unique self-adjoint extension of $(- \Delta + m^2)|_{C_0^\infty}$. But the norm $\|D^{-1/2} . \|_{L^2 (\R^d)}$ is equivalent to the norm of the Hilbert-Sobolev space $H^{-1} (\R^d)$. Therefore, the RKHS of $\mu_0$ is the closure of $C_0^\infty (\Omega)$ in $\|.\|_{-1}$, which is just the Hilbert-Sobolev space $\widetilde{H}^{-1} (\Omega)$, see Section \ref{app:sobolev_spaces} for details. The CMS of $\mu_0$ is given by $\mathsf{H}_{\mu_0} = \hat{G}_0 [\widetilde{H}^{-1} (\Omega)]$, cf. \cite[Thm.~3.2.3]{Bogachev2015}. By noticing that we can write $\hat{G}_0 = r_{\Omega} \circ D^{-1}$, where $r_{\Omega}$ denotes the restriction of a function in $H^{+1} (\R^d)$ to $\Omega$, \cite[Lem.~3.2 \& Eq.~(19)]{Chandler2017} imply that $\mathsf{H}_{\mu_0} = H^{+1} (\Omega)$, see also \cite[Thm.~3.12(iii)]{Chandler2017}.

     Next, we consider a Dirichlet field $\mu_\mathrm{D} = \mathcal{N} (0, \hat{G}_\mathrm{D})$. Again, by definition, the RKHS of $\mu_\mathrm{D}$ is the closure of $C_0^\infty (\Omega)$ in the norm $\|\hat{G}_\mathrm{D}^{1/2} . \|_{L^2 (\Omega)}$. Using \cite[Cor.~II.25]{Guerra1975a} (cf. also the discussion in Section \ref{sec:mutual_info}), this norm is equivalent to the norm in $H^{-1} (\Omega)$ and thus the RKHS of $\mu_\mathrm{D}$ is given by $\mathfrak{X}_{\mu_\mathrm{D}} = H_0^{-1} (\Omega) = H^{-1} (\Omega)$, where the second equality follows from \cite[Cor.~3.29(ii)]{Chandler2017}. The CMS of $\mu_\mathrm{D}$ is given by $\mathsf{H}_{\mu_\mathrm{D}} = \hat{G}_\mathrm{D} [H^{-1} (\Omega)] = H_0^{+1} (\Omega)$, see also \cite[Ch.~6]{Evan2010}. We summarize the results of the last two paragraphs in Table \ref{tab:RKHS_CMS}.

     \begin{table}[t]
        \caption{Reproducing kernel Hilbert space (RKHS) and Cameron-Martin space (CMS) of a free scalar field theory over a bounded region $\Omega \subset \R^d$ with free and Dirichlet boundary conditions.}\label{tab:RKHS_CMS}
        \begin{tabular}{ccc}
            \toprule
            Boundary condition & RKHS & CMS  \\
            \midrule
            Free      & $\widetilde{H}^{-1} (\Omega)$ & $H^{+1} (\Omega)$  \\
            Dirichlet & $H_0^{-1} (\Omega) = H^{-1} (\Omega)$ & $H_0^{+1} (\Omega)$  \\
            \botrule
        \end{tabular}
    \end{table}
    
    Notice that $H_0^{+1} (\Omega)$ is precisely the form domain of the Dirichlet Laplacian $- \Delta_\mathrm{D} + m^2$, see, e.g., \cite[Sec.~XIII.15]{Reed1978}. It is not a coincidence that the form domain of the precision operator of a Gaussian measure $\mathcal{N} (0,\hat{C})$ coincides with its CMS, as can be seen from the following
    \begin{lemma}\label{lem:CM_space}
        For a Gaussian measure $\mu = \mathcal{N} (0, \hat{C})$, $\mathsf{H}_\mu$ coincides with $\mathcal{Q} (\hat{C}^{-1}) = \mathfrak{D} (\hat{C}^{-1/2}) = \hat{C}^{1/2} [\mathcal{H}]$ as a set.
    \end{lemma}
    \begin{proof}
        Fix $\varphi \in \hat{C}^{1/2} [\mathcal{H}]$. Then, there exists $h \in \mathcal{H}$ such that $\varphi = \hat{C}^{1/2} h$. By assumption, $\mathfrak{X}$ is dense in $\mathcal{H}$, $\hat{C}^{1/2}$ is continuous and $\hat{C}^{1/2} [\mathcal{H}]$ is dense in $\mathcal{H}$. Thus $\hat{C}^{1/2} [\mathfrak{X}]$ is also dense in $\mathcal{H}$. Hence, there exists a sequence $\{ h_n \}_{n=1}^\infty$ in $\hat{C}^{1/2} [\mathfrak{X}]$ such that $h_n \to h$ in $\mathcal{H}$-norm. Then,
        \begin{equation}
            \lim_{n \to \infty} \| \varphi - \hat{C}^{1/2}  h_n \|_\mu = \lim_{n \to \infty} \| \hat{C}^{-1/2} ( \varphi - \hat{C}^{1/2}  h_n ) \|_\mathcal{H} = \lim_{n \to \infty} \| h - h_n \|_\mathcal{H} = 0 \; .
        \end{equation}
        This means there exists a sequence in $\hat{C} [\mathfrak{X}]$ that converges to $\varphi$ in $\|.\|_\mu$ and hence $\varphi \in \mathsf{H}_\mu$. Therefore, $\hat{C}^{1/2} [\mathcal{H}] \subset \mathsf{H}_\mu$.
        
        Conversely, fix $\varphi \in \mathsf{H}_\mu$. Then, there exists $\{\varphi_n\}_{n=1}^\infty$ in $\hat{C} [\mathfrak{X}]$ such that $\varphi_n \to \varphi$ with respect to $\|.\|_\mu$. From the definition of $\|.\|_\mu$, we see that $\{\varphi_n\}_{n=1}^\infty$ being Cauchy in $\mathsf{H}_\mu$ implies that  $\{ \hat{C}^{-1/2} \varphi_n \}_{n=1}^\infty$ is Cauchy in $\mathcal{H}$. As $\mathcal{H}$ is complete, there exists $h \in \mathcal{H}$ such that $\hat{C}^{-1/2} \varphi_n \to h$ with respect to $\mathcal{H}$-norm. The sequence $\{ \varphi_n \}_{n=1}^\infty = \hat{C}^{1/2} [\{ \hat{C}^{-1/2} \varphi_n \}_{n=1}^\infty]$ is the image of a converging sequence under a continuous map and is hence convergent. Furthermore, its limit is $\hat{C}^{1/2} h \eqqcolon \psi \in \hat{C}^{1/2} [\mathcal{H}]$. Therefore, $\{ \varphi_n \}_{n=1}^\infty$ converges to an element in $\hat{C}^{1/2} [\mathcal{H}]$ with respect to $\mathcal{H}$-norm. Now $\|\varphi_n - \psi\|_\mu = \| \hat{C}^{-1/2} (\varphi_n - \psi) \|_\mathcal{H} \to 0$ as $n \to \infty$. As the limit of a convergent sequence is unique in a normed space, we conclude that $\varphi = \psi \in \hat{C}^{1/2}[\mathcal{H}]$. Therefore, $\mathsf{H}_\mu \subset \hat{C}^{1/2} [\mathcal{H}]$.
    \end{proof}
    The above Lemma shows that we can interpret any element of the Cameron-Martin space of $\mu = \mathcal{N} (0, \hat{C})$ as a vector in $\mathcal{Q} (\hat{C}^{-1}) = \hat{C}^{1/2} [\mathcal{H}]$, the form domain of the precision operator $\hat{C}^{-1}$. In the setting of statistical field theory, as discussed in Section \ref{sec:covariance_operators}, $\hat{C}^{-1}$ is a differential operator on $L^2 (\Omega)$ and the quadratic form, defined for all $f \in \mathfrak{D} (\hat{C}^{-1})$ as $\braket{f, \hat{C}^{-1} f}_{L^2 (\Omega)}$, is interpreted as the action of the free theory. Hence, in this setting, the CMS can be thought of as the set of ``field configurations'' that are, in a sense, sufficiently regular and therefore have a finite action. If $\mathfrak{X}^*_\beta$ is infinite dimensional, which is the case for $\mathfrak{X}^*_\beta = \mathcal{D}^*_\beta (\Omega)$, then $\mu (\mathsf{H}_\mu) = 0$ \cite[Thm.~2.4.7]{Bogachev2015}. In other words, the set of field configurations with finite action has measure zero.
    
    Note that the covariance operator $\hat{C}$ is a unitary operator from $\mathfrak{X}_\mu$ onto $\mathsf{H}_\mu$ \cite[Sec.~2.4]{Bogachev2015}, $\hat{C}^{1/2}$ is a unitary operator from $\mathcal{H}$ onto $\mathsf{H}_\mu$ and $\hat{C}^{-1/2}$ is a unitary operator from $\mathcal{H}$ onto $\mathfrak{X}_\mu$. In particular, if $\{ \phi_n \}_{n=1}^\infty$ is an orthonormal basis in $\mathcal{H}$, then $\{ \hat{C}^{1/2} \phi_n \}_{n=1}^\infty$ is an orthonormal basis in $\mathsf{H}_\mu$ and $\{ \hat{C}^{-1/2} \phi_n \}_{n=1}^\infty$ is an orthonormal basis in $\mathfrak{X}_\mu$.
    
    We can now state conditions for the equivalence of two Gaussian measures. The following theorem constitutes the basis for the rest of this Section.
    \begin{theorem}[{\cite[Thm.~6.4.6]{Bogachev2015}}]\label{thm:equivalence_master}
        Two centred Radon Gaussian measures $\mu$ and $\nu$ on a locally convex space $X$ are equivalent precisely when $\mathsf{H}_\mu$ and $\mathsf{H}_\nu$ coincide as sets and there exists an invertible operator $A \in \mathfrak{L} (\mathsf{H}_\mu)$ such that $A A^* - I \in \mathrm{HS} (\mathsf{H}_\mu)$ and $\| h \|_\nu =  \| A^{-1} h \|_\mu$ for all $h \in \mathsf{H}_\mu$.
    \end{theorem}
    For the specific kind of Gaussian measures considered in this work, we can formulate the above condition in a more explicit way. But first we need the following
    \begin{proposition}
        Suppose $\hat{C}^{1/2}_\mu [\mathcal{H}] = \hat{C}^{1/2}_\nu [\mathcal{H}]$. Then, the operator $\hat{B} \coloneqq \hat{C}_\mu^{-1/2} \hat{C}_\nu^{1/2}$ is bounded and boundedly invertible.
    \end{proposition}
    \begin{remark}
        Obviously the above result also holds if we interchange $\mu$ and $\nu$.
    \end{remark}
    \begin{proof}
        Consider the operator $\hat{C}_\mu^{-1/2} \hat{C}_\nu^{1/2}$ on $\mathcal{H}$. By the injectivity of the covariance operators $\hat{C}_\mu$ and $\hat{C}_\nu$, this operator is also injective. Suppose $(h_n, \hat{C}_\mu^{-1/2} \hat{C}_\nu^{1/2} h_n) \to (h,g)$ in $\mathcal{H} \oplus \mathcal{H}$. We can now use the continuity of $\hat{C}_\nu^{1/2}$ together with the self-adjointness of $\hat{C}_\mu^{-1/2}$ to obtain
        \begin{equation}
            \braket{f, g}_\mathcal{H} = \lim_{n \to \infty} \braket{f, \hat{C}_\mu^{-1/2} \hat{C}_\nu^{1/2} h_n}_\mathcal{H} = \braket{\hat{C}_\mu^{-1/2} f, \hat{C}_\nu^{1/2} h}_\mathcal{H} = \braket{f, \hat{C}_\mu^{-1/2} \hat{C}_\nu^{1/2} h}_\mathcal{H}
        \end{equation}
        for all $f \in \hat{C}_\mu^{1/2}[\mathcal{H}]$. By assumption, $\hat{C}_\mu^{1/2}[\mathcal{H}]$ is dense in $\mathcal{H}$, and together with the continuity of the inner product this implies that $\braket{f, g}_\mathcal{H} = \braket{f, \hat{C}_\mu^{-1/2} \hat{C}_\nu^{1/2} h}_\mathcal{H}$ for all $f \in \mathcal{H}$. Hence, $g = \hat{C}_\mu^{-1/2} \hat{C}_\nu^{1/2} h$ and the graph of $\hat{C}_\mu^{-1/2} \hat{C}_\nu^{1/2}$ is closed. Therefore, by the closed graph Theorem \cite[Thm.~III.12]{Reed1981}, $\hat{C}_\mu^{-1/2} \hat{C}_\nu^{1/2}$ is bounded. By the symmetry of the above reasoning in the indices $\mu$ and $\nu$, the inverse $(\hat{C}_\mu^{-1/2} \hat{C}_\nu^{1/2})^{-1} = \hat{C}_\nu^{-1/2} \hat{C}_\mu^{1/2}$ is also bounded.
    \end{proof}
    
    Equipped with the above Proposition, we can state the following Lemma, which provides a necessary and sufficient condition for the equivalence of two centred Gaussian measures of the kind studied in this work.
    \begin{lemma}\label{lem:equivalence_weak_1}
        Two centred Gaussian measures $\mu = \mathcal{N} (0, \hat{C}_\mu)$ and $\nu = \mathcal{N} (0, \hat{C}_\nu)$ are equivalent precisely when $\hat{C}^{1/2}_\mu [\mathcal{H}] = \hat{C}^{1/2}_\nu [\mathcal{H}]$ and $\hat{B} \hat{B}^* - I$ is a Hilbert-Schmidt operator on $\mathcal{H}$.
    \end{lemma}
    \begin{remark}
        Here, $\hat{B}^*$ is the continuous extension of the operator $\hat{C}_\nu^{1/2} \hat{C}_\mu^{-1/2}$ defined on $\hat{C}^{1/2}_\mu [\mathcal{H}]$.
    \end{remark}
    \begin{remark}
        This condition should be compared to the necessary and sufficient condition given in the Feldman–H\'ajek Theorem \cite{Feldman1958,Hajek1958} for the case of Gaussian measures on a Hilbert space, see also \cite[Sec.~2.3.2]{DaPrato2014}.
    \end{remark}
    \begin{proof}
        Suppose $\hat{C}^{1/2}_\mu [\mathcal{H}] = \hat{C}^{1/2}_\nu [\mathcal{H}]$ and $\hat{B} \hat{B}^* - I$ is a Hilbert-Schmidt operator on $\mathcal{H}$. By Lemma \ref{lem:CM_space}, $\hat{C}^{1/2}_\mu [\mathcal{H}] = \hat{C}^{1/2}_\nu [\mathcal{H}]$ implies that $\mathsf{H}_\mu$ and $\mathsf{H}_\nu$ coincide as sets. We define the operator $A : \mathsf{H}_\mu \to \mathsf{H}_\mu$ via $h \mapsto \hat{C}_\nu^{1/2} \hat{C}_\mu^{-1/2} h$. Clearly, $\| \hat{C}^{-1/2}_\nu h \|_\mathcal{H} = \| \hat{C}^{-1/2}_\mu (\hat{B}^*)^{-1} h \|_\mathcal{H}$ for all $h \in \hat{C}^{1/2}_\mu [\mathcal{H}]$, which implies that $\|h\|_\nu = \|A^{-1} h\|_\mu$ for all $h \in \mathsf{H}_\mu$. The equivalence of the norms $\|.\|_\mu$ and $\|.\|_\nu$ (cf. Lemma \ref{lem:same_range}) implies that $A^{-1}$ is bounded. More precisely, there exists $\gamma > 0$ such that $\|h\|_\nu \leq \gamma \|h\|_\mu$ for all $h \in \hat{C}^{1/2}_\mu [\mathcal{H}] = \hat{C}^{1/2}_\nu [\mathcal{H}]$ and thus
        \begin{equation}
            \|A^{-1}\| = \sup_{\|h\|_\mu = 1} \|A^{-1} h\|_\mu = \sup_{\|h\|_\mu = 1} \|h\|_\nu \leq \sup_{\|h\|_\mu = 1} \gamma \|h\|_\mu = \gamma \; .
        \end{equation}
        We still need to find the adjoint of $A$. Recalling that $\hat{C}_\mu$ and $\hat{C}_\nu$ are self-adjoint on $\mathcal{H}$ and using the definition of the inner product on $H_\mu$, we see that
        \begin{align*}
            \llangle f, A g \rrangle &= \braket{\hat{C}^{-1/2}_\mu f, \hat{C}^{-1/2}_\mu \hat{C}^{1/2}_\nu \hat{C}^{-1/2}_\mu g }_\mathcal{H} \\
            &= \braket{\hat{C}^{-1/2}_\mu \hat{C}^{1/2}_\mu \hat{C}^{1/2}_\nu \hat{C}^{-1}_\mu f, \hat{C}^{-1/2}_\mu g }_\mathcal{H} \\
            &= \llangle \hat{C}^{1/2}_\mu \hat{C}^{1/2}_\nu \hat{C}^{-1}_\mu f, g \rrangle \; . \numberthis
        \end{align*}
        for all $f,g \in H_\mu$. Hence, $A^* = \hat{C}^{1/2}_\mu \hat{C}^{1/2}_\nu \hat{C}^{-1}_\mu$. Let $\{\phi_n\}_{n=1}^\infty$ be an orthonormal basis in $\mathcal{H}$ contained in $\hat{C}^{1/2}_\mu [\mathcal{H}]$ (for example, the eigenbasis of $\hat{C}_\mu$). Then $\{ \hat{C}^{1/2}_\mu \phi_n \}_{n=1}^\infty$ is an orthonormal basis in $\mathsf{H}_\mu$. Then,
        \begin{align*}\label{eq:HS_calc}
            \| \hat{B} \hat{B}^* - I \|^2_{\mathrm{HS} (\mathcal{H})} &= \sum_{n=1}^\infty \| (\hat{C}^{-1/2}_\mu \hat{C}_\nu \hat{C}^{-1/2}_\mu - I) \, \phi_n \|_\mathcal{H}^2 \\
            &= \sum_{n=1}^\infty \| (\hat{C}_\nu \hat{C}^{-1}_\mu - I) \, \hat{C}^{1/2}_\mu \phi_n \|_\mu^2 \\
            &= \| A A^* - I \|^2_{\mathrm{HS} (\mathsf{H}_\mu)} \; , \numberthis
        \end{align*}
        which implies that $A A^* - I \in \mathrm{HS} (\mathsf{H}_\mu)$. In summary, $\mathsf{H}_\mu$ and $\mathsf{H}_\nu$ coincide as sets and there exists an invertible operator $A$ on $\mathsf{H}_\mu$ such that $A A^* - I \in \mathrm{HS} (\mathsf{H}_\mu)$ and $\| h \|_\nu =  \| A^{-1} h \|_\mu$ for all $h \in \mathsf{H}_\mu$. Therefore, by Theorem \ref{thm:equivalence_master}, $\mu \sim \nu$.
    
        Conversely, suppose $\mu \sim \nu$. Then, $\mathsf{H}_\mu$ and $\mathsf{H}_\nu$ coincide as sets, which implies that $\hat{C}^{1/2}_\mu [\mathcal{H}] = \hat{C}^{1/2}_\nu [\mathcal{H}]$. Furthermore, there exists an invertible $A \in \mathfrak{L} (\mathsf{H}_\mu)$ such that $A A^* - I \in \mathrm{HS} (\mathsf{H}_\mu)$ and $\| h \|_\nu = \| A^{-1} h \|_\mu$. By the definitions of $\|.\|_\mu$ and $\|.\|_\nu$, the last property can be written as $\| A^{-1} h \|_\mu = \| \hat{C}^{1/2}_\mu \hat{C}_\nu^{-1/2} h \|_\mu$ for all $h \in \hat{C}^{1/2}_\mu [\mathcal{H}]$. Thus, $A^{-1}$ coincides with $U \hat{C}^{1/2}_\mu \hat{C}_\nu^{-1/2}$, where $U$ is an orthogonal transformation on $\mathsf{H}_\mu$. Furthermore, $A = \hat{C}^{1/2}_\nu \hat{C}^{-1/2}_\mu U^*$ and, by the same calculation as above, $A^* = U \hat{C}^{1/2}_\mu \hat{C}^{1/2}_\nu \hat{C}^{-1}_\mu$. Thus, $A A^* = \hat{C}_\nu \hat{C}^{-1}_\mu$, which, by \eqref{eq:HS_calc}, implies that $\hat{B} \hat{B}^* - I$ is a Hilbert-Schmidt operator on $\mathcal{H}$.
    \end{proof}
    
    Following \cite{Pinski2015}, we can give a condition in terms of the precision operators of the Gaussian measures that is equivalent to the one given in the above Lemma. Suppose $\mu = \mathcal{N} (0,\hat{C}_\mu)$ and $\nu = \mathcal{N} (0,\hat{C}_\nu)$ are two centred Gaussian measures such that $\hat{C}^{1/2}_\mu [\mathcal{H}] = \hat{C}^{1/2}_\nu [\mathcal{H}]$. Recall that the precision operators $\hat{C}^{-1}_\mu$ and $\hat{C}^{-1}_\nu$ are strictly positive, densely defined unbounded self-adjoint operators on $\mathcal{H}$. To these precision operators we can associate bilinear forms $\mathfrak{q}_\mu$ and $\mathfrak{q}_\nu$, which, since we assume $\hat{C}^{1/2}_\mu [\mathcal{H}] = \hat{C}^{1/2}_\nu [\mathcal{H}]$, have coinciding form domains, i.e., $\mathcal{Q} (\hat{C}^{-1}_\mu) = \mathcal{Q} (\hat{C}^{-1}_\nu)$ (for details on the theory of quadratic forms needed in the following, see Appendix \ref{app:quadratic_forms} and references therein). In particular, this means that we can add and subtract these forms. We define the form $\Delta \mathfrak{q}$ on $\hat{C}^{1/2}_\mu [\mathcal{H}] = \hat{C}^{1/2}_\nu [\mathcal{H}]$ by
    \begin{equation}
        \Delta \mathfrak{q} (f,g) \coloneqq \mathfrak{q}_\mu (f,g) - \mathfrak{q}_\nu (f,g) = \braket{\hat{C}^{-1/2}_\mu f , \hat{C}^{-1/2}_\mu g} - \braket{\hat{C}^{-1/2}_\nu f , \hat{C}^{-1/2}_\nu g} \; .
    \end{equation}

    By Lemma \ref{lem:equivalence_weak_1}, a necessary condition for the equivalence of $\mu$ and $\nu$ is that $\hat{B} \hat{B}^* - I$, where $\hat{B} \coloneqq \hat{C}_\mu^{-1/2} \hat{C}_\nu^{1/2}$, is a Hilbert-Schmidt operator on $\mathcal{H}$. If $\hat{C}^{1/2}_\mu [\mathcal{H}] = \hat{C}^{1/2}_\nu [\mathcal{H}]$, then, by \cite[Lem.~6.3.1(ii)]{Bogachev2015}, $\hat{B} \hat{B}^* - I$ is a Hilbert-Schmidt operator on $\mathcal{H}$ if and only if $\hat{B}^* \hat{B} - I$ is. An elementary calculation yields
    \begin{equation}
        \| \hat{B}^* \hat{B} - I \|^2_{\mathrm{HS}(\mathcal{H})} = \sum_{n=1}^\infty \sum_{m=1}^\infty \left( \Delta \mathfrak{q} (\phi_n, \psi_m) \right)^2 \; ,
    \end{equation}
    where $\{\phi_n\}_{n=1}^\infty$ and $\{\psi_n\}_{n=1}^\infty$ are any two orthonormal bases in $\mathsf{H}_\nu$. We can combine this result with Lemma \ref{lem:equivalence_weak_1} and arrive at the following Theorem, which is the main result of this Section.
    \begin{theorem}\label{thm:equivalence}
        Two centred Gaussian measures $\mu = \mathcal{N} (0, \hat{C}_\mu)$ and $\nu = \mathcal{N} (0, \hat{C}_\nu)$ are equivalent precisely when $\hat{C}^{1/2}_\mu [\mathcal{H}] = \hat{C}^{1/2}_\nu [\mathcal{H}]$ and one of the following equivalent conditions holds:
        \begin{enumerate}
            \item $\hat{B} \hat{B}^* - I$, where $\hat{B} \coloneqq \hat{C}_\mu^{-1/2} \hat{C}_\nu^{1/2}$, is a Hilbert-Schmidt operator on $\mathcal{H}$,
            \item $\hat{B}^* \hat{B} - I$ is a Hilbert-Schmidt operator on $\mathcal{H}$,
            \item The form $\Delta \mathfrak{q} \coloneqq \mathfrak{q}_\mu - \mathfrak{q}_\nu$ satisfies
            \begin{equation}
                \sum_{n=1}^\infty \sum_{m=1}^\infty \left( \Delta \mathfrak{q} (\phi_n, \psi_m) \right)^2 < + \infty
            \end{equation}
            for any two orthonormal bases $\{\phi_n\}_{n=1}^\infty$ and $\{\psi_n\}_{n=1}^\infty$ in $\mathsf{H}_\nu$.
        \end{enumerate}
    \end{theorem}
    \begin{remark}
        Of course, this statement also holds if we interchange $\mu$ and $\nu$.
    \end{remark}
    \begin{remark}
        Note that condition $\emph{3.}$ in the above Theorem implies that $\Delta \mathfrak{q}$ is infinitesimally form bounded with respect to $\mathfrak{q}_\nu$, see \cite[Sec.~4.2]{Pinski2015}.
    \end{remark}

    In summary, two Gaussian measures $\mu = \mathcal{N} (0,\hat{C}_\mu)$ and $\nu = \mathcal{N} (0,\hat{C}_\nu)$ describing, for example, free scalar field theories over a bounded region $\Omega \subset \R^d$, are equivalent (and thus the relative entropy between them is finite) if and only if their precision operators $\hat{C}^{-1}_\mu$ and $\hat{C}^{-1}_\nu$ have the same form domain and the operator $\hat{B} \hat{B}^* - I$ is Hilbert-Schmidt on $\mathcal{H}$, or, equivalently, the difference in the forms associated with the precision operators is square summable in the sense of item \emph{3.} in Theorem \ref{thm:equivalence}. The necessity for $\hat{B} \hat{B}^* - I$ to be Hilbert-Schmidt will become clear in the next Section, Appendix \ref{app:KL_formula}, where we show that the relative entropy between $\mu$ and $\nu$ is given by the logarithm of the regularized Fredholm determinant of $\hat{B} \hat{B}^*$, which is finite precisely when $\hat{B} \hat{B}^* - I$ is Hilbert-Schmidt. The requirement that the form domains of the precision operators coincide is useful when showing that two Gaussian measures are mutually singular, i.e., the relative entropy between them is infinite. Consider, for example, the relative entropy between a Dirichlet field and a field with free boundary conditions, both of mass $m > 0$. As discussed above, the form domain of the Dirichlet precision operator is $H_0^{+1} (\Omega)$ while the form domain of the precision operator with free boundary conditions is $H^{+1} (\Omega)$, cf. Table \ref{tab:RKHS_CMS}. But in general $H_0^{+1} (\Omega) \subsetneq H^{+1} (\Omega)$, for example when $\Omega$ is $C^0$ \cite[Cor.~3.29(vii)]{Chandler2017}. In such a case we can conclude that these two Gaussian measures are mutually singular and the relative entropy between a Dirichlet field and a field with free boundary conditions is infinite, see also the discussion in Section \ref{sec:diff_bcs}.

    \section{A Formula for the Relative Entropy}\label{app:KL_formula}
    
    In this Section, we derive an explicit expression for the relative entropy in terms of the eigenvalues of the operator $(\hat{C}_\nu^{-1/2} \hat{C}_\mu^{1/2}) (\hat{C}_\nu^{-1/2} \hat{C}_\mu^{1/2})^*$. More precisely, let $\mu = \mathcal{N} (0, \hat{C}_\mu)$ and $\nu = \mathcal{N} (0, \hat{C}_\nu)$ be two equivalent non-degenerate centred Gaussian measures as defined in Appendix \ref{app:func_integrals}. We recall that in this case the relative entropy $D_\mathrm{KL} (\mu \| \nu)$ is given by
    \begin{equation}
        D_\mathrm{KL} (\mu \| \nu) = \int_{\mathfrak{X}_\beta^*} \log \left[ \frac{\D \mu}{\D \nu} (\varphi) \right] \; \D \mu (\varphi) \; ,
    \end{equation}
    where $\D \mu / \D \nu$ is the Radon-Nikodym derivative of $\mu$ with respect to $\nu$. 
    
    We first derive an expression for the density $\D \mu / \D \nu$ in terms of the eigenvalues of the operator $\hat{S} \coloneqq \hat{B} \hat{B}^*$, where $\hat{B} \coloneqq \hat{C}_\nu^{-1/2} \hat{C}_\mu^{1/2}$. Recall that $\mu \sim \nu$ implies that the operator $\hat{S} - I$ is Hilbert-Schmidt on $\mathcal{H}$.    
    \begin{theorem}[{\cite[Col.~6.4.11]{Bogachev2015}}]
        Let $\mu = \mathcal{N} (0, \hat{C}_\mu)$ and $\nu = \mathcal{N} (0, \hat{C}_\nu)$ be two equivalent non-degenerate centred Gaussian measures. Then, the Radon-Nikodym density of $\mu$ with respect to $\nu$ is given by
        \begin{equation}\label{eq:density_eigenvalues}
            \frac{\D \mu}{\D \nu} (\varphi) = \exp \left[ \frac{1}{2} \sum_{n=1}^\infty \left( \frac{\alpha_n - 1}{\alpha_n} ( \eta_n (\varphi) )^2 - \log \alpha_n \right) \right] \; ,
        \end{equation}
        where $\{\alpha_n\}_{n=1}^\infty$ is the sequence of eigenvalues of $\hat{S}$ with corresponding eigenvectors $\{\phi_n\}_{n=1}^\infty$ in $\mathcal{H}$ and $\eta_n$ is the inclusion of $\hat{C}_\nu^{-1/2} \phi_n$ into $L^2 (\mathfrak{X}_\beta^*,\nu) \eqqcolon L^2 (\nu)$.
    \end{theorem}
    
    Notice that this Theorem is essentially \cite[Col.~6.4.11]{Bogachev2015}, which gives the analogous result for Gaussian measures on Hilbert spaces. The proof in \cite[Col.~6.4.11]{Bogachev2015} works also for the case considered here, but for convenience, we shall reproduce the proof following this reference as well as \cite{Michalek1999,Minh2021}.
    
    \begin{proof}
        Since $\hat{S}-I$ is Hilbert-Schmidt on $\mathcal{H}$, there exists, by the Hilbert-Schmidt theorem, an orthonormal basis $\{\phi_n\}_{n=1}^\infty$ in $\mathcal{H}$, such that $(\hat{S} - I) \phi_n = (\alpha_n - 1) \phi_n$ for all $n \in \N$, where the $\alpha_n$ are the eigenvalues of $\hat{S}$. Furthermore, the series $\sum_{n=1}^\infty (\alpha_n - 1)^2$ converges. The sequence $\{\hat{C}_\nu^{-1/2} \phi_n\}_{n=1}^\infty$ is an orthonormal basis in $\mathfrak{X}_\nu$, the RKHS of $\nu$. We recall that $\mathfrak{X}_\nu \subset L^2 (\nu)$ and we denote by $\eta_n$ the inclusion of $\hat{C}_\nu^{-1/2} \phi_n$ in $L^2 (\nu)$.
    
        We now show that the series in the exponent on the right-hand side of \eqref{eq:density_eigenvalues} converges in $L^2 (\nu)$. First, define
        \begin{equation}
            s_N (\varphi) \coloneqq \sum_{n=1}^N \left( \frac{\alpha_n - 1}{\alpha_n} ( \eta_n (\varphi) )^2 - \log \alpha_n \right) \in L^2 (\nu) \; .
        \end{equation}
        Recall that, by Wick's theorem \cite[Prop.~I.2]{Simon1974}, for all $f,g \in L^2 (\nu)$,
        \begin{equation}
            \braket{f^2,g^2}_{L^2(\nu)} = \|f\|^2_{L^2(\nu)} \|g\|^2_{L^2(\nu)} + 2 \braket{f,g}^2_{L^2(\nu)} \; ,
        \end{equation}
        where the squares have to be understood as the pointwise products of the linear functionals $f$ and $g$ on $\mathfrak{X}^*$. In particular, for two orthonormal basis vectors $\eta_n$ and $\eta_m$, $\braket{\eta_n^2,\eta_m^2}_{L^2(\nu)} = 1 + 2 \delta_{nm}$. From this, we immediately see that \cite[Lem.~13]{Minh2021}
        \begin{align*}
            \frac{1}{2} \braket{(\eta_n^2-1),(\eta_m^2-1)}_{L^2(\nu)} &= \frac{1}{2} \left( \braket{\eta_n^2,\eta_m^2}_{L^2(\nu)} - \braket{\eta_n^2,1}_{L^2(\nu)} - \braket{1,\eta_m^2}_{L^2(\nu)} + \braket{1,1}_{L^2(\nu)} \right) \\
            &= \frac{1}{2} \left( 1 + 2 \delta_{nm} - \|\eta_n\|^2_{L^2(\nu)} - \|\eta_m\|^2_{L^2(\nu)} + 1 \right) \\
            &= \delta_{nm} \; .
        \end{align*}
        Therefore, $\{\frac{1}{\sqrt{2}} (\eta^2_n - 1)\}_{n=1}^{\infty}$ is an orthonormal system in $L^2 (\nu)$. Following \cite[Lem.~14]{Minh2021}, we rewrite $s_N$ as
        \begin{equation}
            s_N = \sum_{n=1}^N \left( \frac{\sqrt{2} ( \alpha_n - 1 )}{\alpha_n} \frac{1}{\sqrt{2}} ( \eta_n^2 - 1 ) + \frac{\alpha_n - 1}{\alpha_n} - \log \alpha_n \right) = r_N + t_N \; ,
        \end{equation}
        where
        \begin{equation}
            r_N \coloneqq \sum_{n=1}^N \frac{\sqrt{2} ( \alpha_n - 1 )}{\alpha_n} \frac{1}{\sqrt{2}} ( \eta_n^2 - 1 ) \; , \qquad t_N \coloneqq \sum_{n=1}^N \left( \frac{\alpha_n - 1}{\alpha_n} - \log \alpha_n \right) \; .
        \end{equation}
    
        As $L^2 (\nu)$ is complete, we only need to show that $\{r_N\}_{N = 1}^{\infty}$ and $\{t_N\}_{N = 1}^{\infty}$ are Cauchy. We start with $\{t_N\}_{N = 1}^{\infty}$. Without loss of generality, assume $N > M$. Then,
        \begin{align*}
            \|t_N - t_M\|_{L^2(\nu)} &= \left\| \sum_{n=M+1}^N \left( \frac{\alpha_n - 1}{\alpha_n} - \log \alpha_n \right) \right\|_{L^2 (\nu)} \\
            &\leq \sum_{n=M+1}^N \left\| \frac{\alpha_n - 1}{\alpha_n} - \log \alpha_n \right\|_{L^2 (\nu)} \\
            &= \sum_{n=M+1}^N \left| \frac{\alpha_n - 1}{\alpha_n} - \log \alpha_n \right| \; . \numberthis
        \end{align*}
        Hence, convergence of $\sum_{n=1}^\infty | (\alpha_n - 1) \alpha_n^{-1} - \log \alpha_n |$ in $\R$ implies convergence of $\{t_N\}_{N = 1}^{\infty}$ in $L^2 (\nu)$. To see that this is indeed the case, recall that $\sum_{n=1}^{\infty} (\alpha_n - 1)^2 < + \infty$ as $\hat{S} - I \in \mathrm{HS} (\mathcal{H})$. This implies that $\sum_{n=1}^\infty | (\alpha_n - 1) \alpha_n^{-1} - \log \alpha_n |$ converges, as shown in \cite[p.~336f]{Michalek1999}. Therefore, $t \coloneqq \lim_{N \to \infty} t_N \in L^2 (\nu)$ exists. Following an analogous argument, convergence of $\sum_{n=1}^\infty (1 - \alpha_n^{-1})^2$ in $\R$ implies convergence of $\{r_N\}_{N = 1}^{\infty}$ in $L^2 (\nu)$. The convergence of $\sum_{n=1}^\infty (1 - \alpha_n^{-1})^2$ follows from the convergence of $\sum_{n=1}^\infty (1 - \alpha_n)^2$ \cite[p.~336]{Michalek1999}. Therefore, $s \coloneqq \lim_{N \to \infty} s_N \in L^2 (\nu)$ exists. We can rewrite $s (\varphi)$ as
        \begin{equation}
            s (\varphi) = \sum_{n=1}^\infty \left( 1 - \frac{1}{\alpha_n} \right) \left[ (\eta_n (\varphi))^2 - 1 \right] - \underbrace{\sum_{n=1}^\infty \left( \frac{1}{\alpha_n} - \log \frac{1}{\alpha_n} - 1 \right)}_{\in \, \R} \; .
        \end{equation}
        Therefore, by \cite[Cor.~6.4.10]{Bogachev2015}, $\rho (\varphi) \coloneqq \exp [\frac{1}{2} s (\varphi)]$ is in $L^1 (\nu)$ and $\| \rho \|_{L^1 (\nu)}^{-1} \rho \cdot \nu$ is a Gaussian measure.
    
        Finally, we show that $\D \mu / \D \nu = \rho$, i.e., that $\mu = \rho \cdot \nu$. First, note that
        \begin{equation}\label{eq:etaeta_mu}
            \braket{\eta_n, \eta_m}_{L^2 (\mu)} = \braket{\hat{C}^{1/2}_\mu \hat{C}_\nu^{-1/2} \phi_n , \hat{C}^{1/2}_\mu \hat{C}_\nu^{-1/2} \phi_m}_\mathcal{H} = \braket{\hat{S} \phi_n, \phi_m}_\mathcal{H} = \alpha_n \delta_{nm} \; .
        \end{equation}
        Fix $\xi \in \mathfrak{X}_\nu \subset L^2 (\nu)$. Such a $\xi$ can  be written as $\xi = \sum_{n=1}^\infty c_n \eta_n$ with $\sum_{n=1}^\infty c_n^2 < + \infty$. As $\{\eta_n\}_{n=1}^\infty$ is an orthonormal basis in $\mathfrak{X}_\nu$, every $\eta_n$ is a standard Gaussian random variable and thus, using the change of variables formula, we see that
        \begin{align*}
            \int_{\mathfrak{X}_\beta^*} \exp \left[ \frac{\alpha_n - 1}{2 \alpha_n} \eta_n^2 + \I c_n \eta_n \right] \D \nu (\varphi) &= \frac{1}{\sqrt{2 \pi}} \int_{-\infty}^{+\infty} \exp \left[ \frac{\alpha_n - 1}{2 \alpha_n} x^2 + \I c_n x - \frac{1}{2} x^2 \right] \D x \\
            &= \sqrt{\alpha_n} \, \exp \left[ - \frac{1}{2} c^2_n \alpha_n \right] \; . \numberthis
        \end{align*}
        In particular, we therefore have
        \begin{equation}
            \int_{\mathfrak{X}^*_\beta} \E^{\I \xi (\varphi)} \rho (\varphi) \; \D \nu (\varphi) = \exp \left[ - \frac{1}{2} \sum_{n=1}^\infty c_n^2 \alpha_n \right] \; .
        \end{equation}
        An analogous calculation, using \eqref{eq:etaeta_mu}, yields
        \begin{equation}
            \int_{\mathfrak{X}^*_\beta} \E^{\I \xi (\varphi)} \; \D \mu (\varphi) = \exp \left[ - \frac{1}{2} \sum_{n=1}^\infty c_n^2 \alpha_n \right] \; .
        \end{equation}
        Thus, $\hat{\mu} = \widehat{\rho \cdot \nu}$ and, as a Gaussian measure is uniquely determined by its Fourier transform, cf. \cite[Lem.~7.13.5]{Bogachev2007}, $\mu = \rho \cdot \nu$. Therefore, we conclude that the Radon-Nikodym density of $\mu$ with respect to $\nu$ is given by \eqref{eq:density_eigenvalues}.
    \end{proof}
    
    \begin{corollary}
        Let $\mu = \mathcal{N} (0, \hat{C}_\mu)$ and $\nu = \mathcal{N} (0, \hat{C}_\nu)$ be equivalent. Then, the relative entropy $D_\mathrm{KL} (\mu \| \nu)$ is given by
        \begin{equation}\label{eq:KL_eigenvalues}
            D_\mathrm{KL} (\mu \| \nu) = \frac{1}{2} \sum_{n=1}^\infty \left( \alpha_n - \log \alpha_n - 1 \right) \; .
        \end{equation}
    \end{corollary}
    \begin{remark}
        This result should be compared to \cite[Ch.~2]{Michalek1999}.
    \end{remark}
    \begin{proof}
        By a straightforward calculation, we see that
        \begin{align*}
            D_\mathrm{KL} (\mu \| \nu) &= \int_{\mathfrak{X}^*_\beta} \log \left[ \frac{\D \mu}{\D \nu} (\varphi) \right] \; \D \mu (\varphi) \\
            &= \frac{1}{2} \sum_{n=1}^\infty \int_{\mathfrak{X}^*_\beta} \left( \frac{\alpha_n - 1}{\alpha_n} ( \eta_n (\varphi) )^2 - \log \alpha_n \right) \; \D \mu (\varphi) \\
            &= \frac{1}{2} \sum_{n=1}^\infty \left( \frac{\alpha_n - 1}{\alpha_n \,} \| \eta_n \|^2_{L^2 (\mu)} - \log \alpha_n \right) \\
            &= \frac{1}{2} \sum_{n=1}^\infty \left( \alpha_n - \log \alpha_n - 1 \right) \; .
        \end{align*}
    \end{proof}
    It is worthwhile to compare this expression of the relative entropy with that of the relative entropy between two non-degenerate multivariate normal distributions. Let $\mu = \mathcal{N} (\boldsymbol{0}, \Sigma_\mu)$ and $\nu = \mathcal{N} (\boldsymbol{0}, \Sigma_\nu)$ be two centred Gaussian measures on $\R^N$ with covariance matrices $\Sigma_\mu$ and $\Sigma_\nu$, respectively. The relative entropy of $\mu$ with respect to $\nu$ takes the form \cite[Ch.~1]{Pardo2018}
    \begin{equation}\label{eq:KL_multivariate_Gaussian}
        D_\mathrm{KL} (\mu \| \nu) = \frac{1}{2} \left[ \mathrm{tr} \left( \Sigma^{-1}_\nu \Sigma_\mu - \mathds{1}_N \right) - \log \frac{\det \Sigma_\mu}{\det \Sigma_\nu} \right] \; .
    \end{equation}
    If we denote by $\{\alpha_n\}_{n=1}^N$ the set of eigenvalues of the matrix $S \coloneqq \Sigma^{-1/2}_\nu \Sigma_\mu \Sigma^{-1/2}_\nu$, this expression takes the form
    \begin{equation}\label{eq:KL_eigenvalues_finite}
        D_\mathrm{KL} (\mu \| \nu) = \frac{1}{2} \sum_{n=1}^N \left( \alpha_n - \log \alpha_n - 1 \right) \; .
    \end{equation}
    Thus, the expression \eqref{eq:KL_eigenvalues} coincides with the na\"ive infinite dimensional generalization of \eqref{eq:KL_eigenvalues_finite}.
    
    Finally, we mention another way to write \eqref{eq:KL_eigenvalues}. Consider the function
    \begin{equation}
        \Psi(x) = \left( (1+x) \, \E^{ - x } \right) - 1 \; ,
    \end{equation}
    and define, for all bounded self-adjoint operators $T$ on $\mathcal{H}$, $\Psi(T)$ via the functional calculus. It can be shown that if $T$ is Hilbert-Schmidt, then $\Psi(T)$ is of trace class \cite[Lem.~6.1]{Simon1977}. We then define, for all Hilbert-Schmidt operators $T$ \cite[Sec.~6]{Simon1977} (see also \cite{Gohberg1978}),
    \begin{equation}\label{eq:reg_fredholm_det}
        \mathrm{det}_2 (I + T) \coloneqq \mathrm{det}_1 (I + \Psi(T)) = \mathrm{det}_1 \left( (I+T) \, \E^{ - T } \right) \; ,
    \end{equation}
    where $\mathrm{det}_1$ is the Fredholm determinant \cite[Ch.~3]{Simon2005}. We call $\mathrm{det}_2$ the regularized Fredholm determinant and it can be thought of as a variant of the Fredholm determinant (which is defined only for operators of trace class) that is also valid for Hilbert-Schmidt operators. For the specific case $T = \hat{S} - I$, we can use \cite[Thm.~6.2]{Simon1977} to obtain the result
    \begin{equation}
        \mathrm{det}_2 (I + (\hat{S}-I)) = \mathrm{det}_2 (\hat{S}) = \prod_{n = 1}^\infty \alpha_n \, \E^{ - \alpha_n + 1 } \; .
    \end{equation}
    In particular, we see that we can write the relative entropy between two centred Radon Gaussian measures $\mu = \mathcal{N} (0,\hat{C}_\mu)$ and $\nu = \mathcal{N} (0,\hat{C}_\nu)$ as
    \begin{equation}\label{eq:DKL_logdet2}
        D_\mathrm{KL} (\mu \| \nu) = - \frac{1}{2} \log \mathrm{det}_2 (\hat{S}) \; .
    \end{equation}
    From the above expression we see that the relative entropy is finite if and only if $\hat{C}^{1/2}_\mu [\mathcal{H}] = \hat{C}^{1/2}_\nu [\mathcal{H}]$, which is required in order for $\hat{S}$ to be well-defined, and $\hat{S} - I \in \mathrm{HS} (\mathcal{H})$. But this is just the necessary and sufficient condition for the equivalence of $\mu$ and $\nu$, cf. Theorem \ref{thm:equivalence}. This reflects the fact that the relative entropy between two Gaussian measures is finite if and only if the two measures are equivalent, see also the discussion in Section \ref{sec:info_theory}.

    \section{Closed-Form Expression for Dirichlet Relative Entropy}\label{app:1D_Dirichlet_RE}

    In this Section, we derive the closed-form expression of the relative entropy between two one-dimensional Dirichlet fields with different masses on an interval of length $L$ given in \eqref{eq:DKL_Dirichlet_1d}. We use techniques known from the evaluation of Matsubara sums in thermal field theory, see, e.g., \cite[Ch.~3]{Bellac2000} and \cite{Altland2010,Nieto1995}. More precisely, we rewrite the series representation of the Dirichlet relative entropy as a contour integral with a suitable choice of weight function. This computational trick is also known as Sommerfeld-Watson transformation \cite{Sommerfeld1949,Watson1918}.
    
    We first make some general remarks on the Sommerfeld-Watson transformation. Let $f$ be a function that is analytic in a neighbourhood of the real line. Furthermore, assume that the series
    \begin{equation}
        \sum_{n = -\infty}^{+ \infty} f (n)
    \end{equation}
    exists. Let $\varpi$ be a meromorphic function with simple poles of unit residue at the points $n \in \Z$ and no other poles. This function shall be referred to as a weight function. Assuming that $f$ is meromorphic and that the pointwise product $f (z) \varpi (z)$ vanishes sufficiently fast as $|z| \to \infty$, we can use the residue theorem together with a deformation of the integration contour\footnote{For details, see, e.g., \cite{Nieto1995}.} to arrive at
    \begin{equation}\label{eq:residue_theorem}
        \sum_{n = -\infty}^{+ \infty} f (n) = - \sum_i \mathrm{Res}_{z = z_i} \left( f (z) \varpi (z) \right) \; ,
    \end{equation}
    where the sum on the right-hand side runs over all poles $z_i$ of $f$.
    
    We now apply this result to the Dirichlet relative entropy in one dimension. The starting point of the calculation is the series representation of this Dirichlet relative entropy given by
    \begin{equation}\label{eq:DKL_Dirichlet_1d_series_app}
        D^\mathrm{D}_\mathrm{KL} (\mu_1 \| \mu_2) = \frac{1}{2} \sum_{n=1}^\infty \left[ \frac{m_1^2 - m_2^2}{\left( \frac{n \pi}{L} \right)^2 + m_2^2} - \log \left( \frac{m_1^2 - m_2^2}{\left( \frac{n \pi}{L} \right)^2 + m_2^2} + 1 \right) \right] \; ,
    \end{equation}
    obtained by using \eqref{eq:DKL_series} and plugging in the eigenvalues of the one-dimensional Dirichlet Laplacian on an interval of length $L$, i.e., $\lambda_n = (n \pi)^2 / L^2$, $n \in \N$. We first write this series as
    \begin{equation}\label{eq:DKL_Dirichlet_1d_series_app_2}
        D^\mathrm{D}_\mathrm{KL} (\mu_1 \| \mu_2) = \frac{1}{4} \sum_{n=-\infty}^\infty \left[ \frac{m_1^2 - m_2^2}{\left( \frac{n \pi}{L} \right)^2 + m_2^2} - \log \left( \frac{\left( \frac{n \pi}{L} \right)^2 + m_1^2}{\left( \frac{n \pi}{L} \right)^2 + m_2^2} \right) \right] - \frac{1}{2} \Omega_0 \; ,
    \end{equation}
    where $\Omega_0$ is the zero-mode given by
    \begin{equation}
        \Omega_0 = \frac{1}{2} \left[ \frac{m_1^2}{m_2^2} - \log \left( \frac{m_1^2}{ m_2^2} \right) - 1 \right] \; .
    \end{equation}
    
    Note that, since we consider the one-dimensional case, the series in \eqref{eq:DKL_Dirichlet_1d_series_app_2} can be written as the sum of two convergent series, cf. the discussion below Lemma \ref{lem:equivalence_mass_independent}. More precisely, we have
    \begin{equation}
        D^\mathrm{D}_\mathrm{KL} (\mu_1 \| \mu_2) = s_1 - s_2 - \frac{1}{2} \Omega_0 \; ,
    \end{equation}
    where
    \begin{equation}
        s_1 \coloneqq \frac{1}{4} \sum_{n=-\infty}^\infty \frac{m_1^2 - m_2^2}{\left( \frac{n \pi}{L} \right)^2 + m_2^2} \qquad \mathrm{and} \qquad s_2 \coloneqq \frac{1}{4} \sum_{n=-\infty}^\infty \log \left( \frac{\left( \frac{n \pi}{L} \right)^2 + m_1^2}{\left( \frac{n \pi}{L} \right)^2 + m_2^2} \right) \; .
    \end{equation}

    First, consider the sum $s_1$. We define the function $f_1 (z) \coloneqq (m_1^2 - m_2^2) / ((z \pi / L)^2 + m_2^2)$, which has simple poles at $z = \pm \I L m_2 / \pi$. As a weight function, we choose $\varpi (z) = \pi \cot (\pi z)$. Since $f_1$ has only simple poles, the residues are readily computed, and upon using \eqref{eq:residue_theorem}, we find
    \begin{equation}
        s_1 = - \frac{\pi}{4} \sum_{\alpha = \pm} \mathrm{Res}_{z = \alpha \I L m_2 / \pi} \left( \frac{m_1^2 - m_2^2}{\left( \frac{z \pi}{L} \right)^2 + m_2^2} \cot (\pi z) \right) = \frac{L (m_1^2 - m_2^2)}{4 m_2 \tanh (L m_2)} \; .
    \end{equation}

    \begin{figure}[t]
        \hspace*{-2.5em}\subfloat[\label{subfig:contour1}]{%
            \hspace*{2.3em}\includegraphics[width=0.5\textwidth]{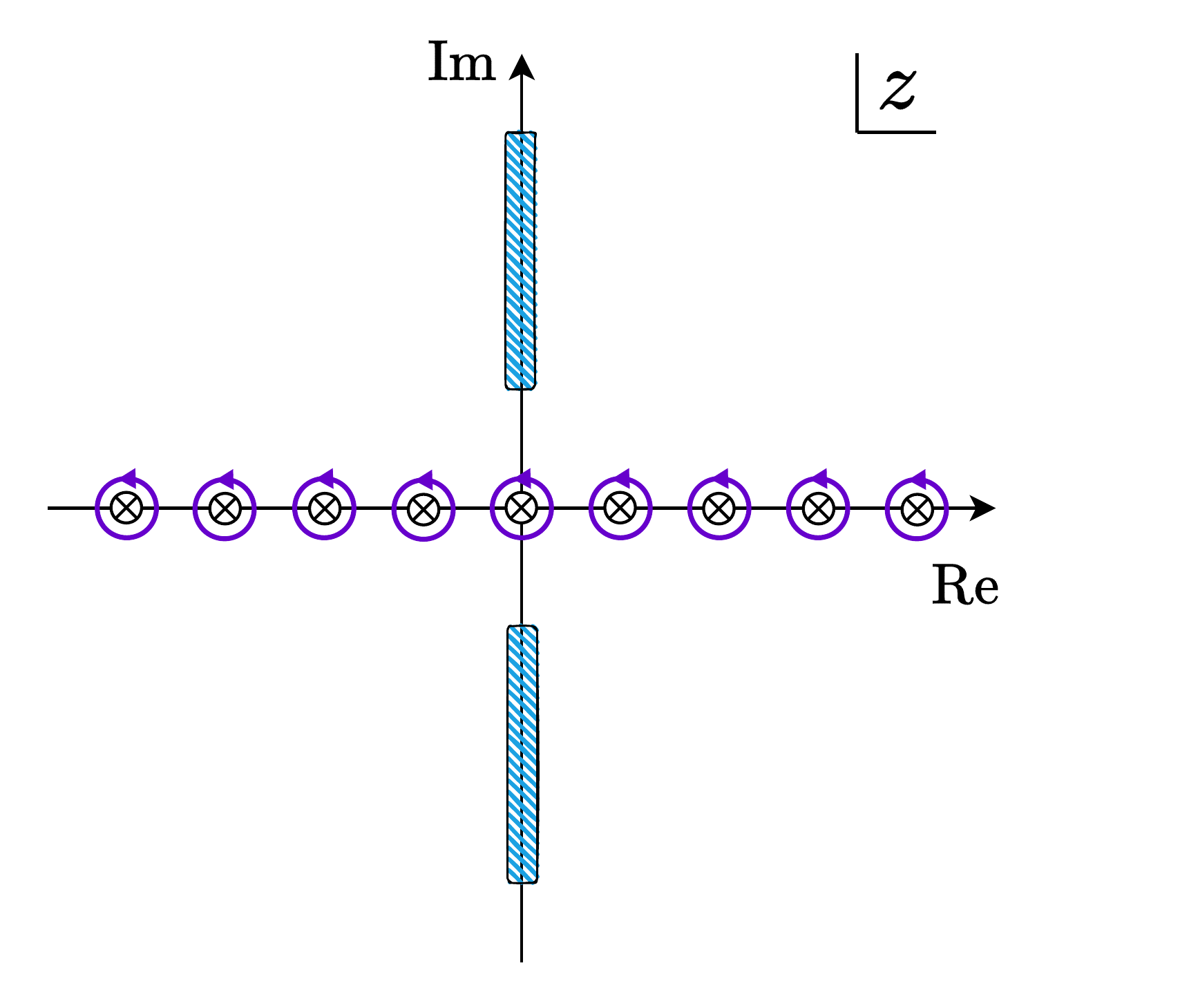}%
        }
        \hfill
        \subfloat[\label{subfig:contour2}]{%
            \hspace*{2.3em}\includegraphics[width=0.5\textwidth]{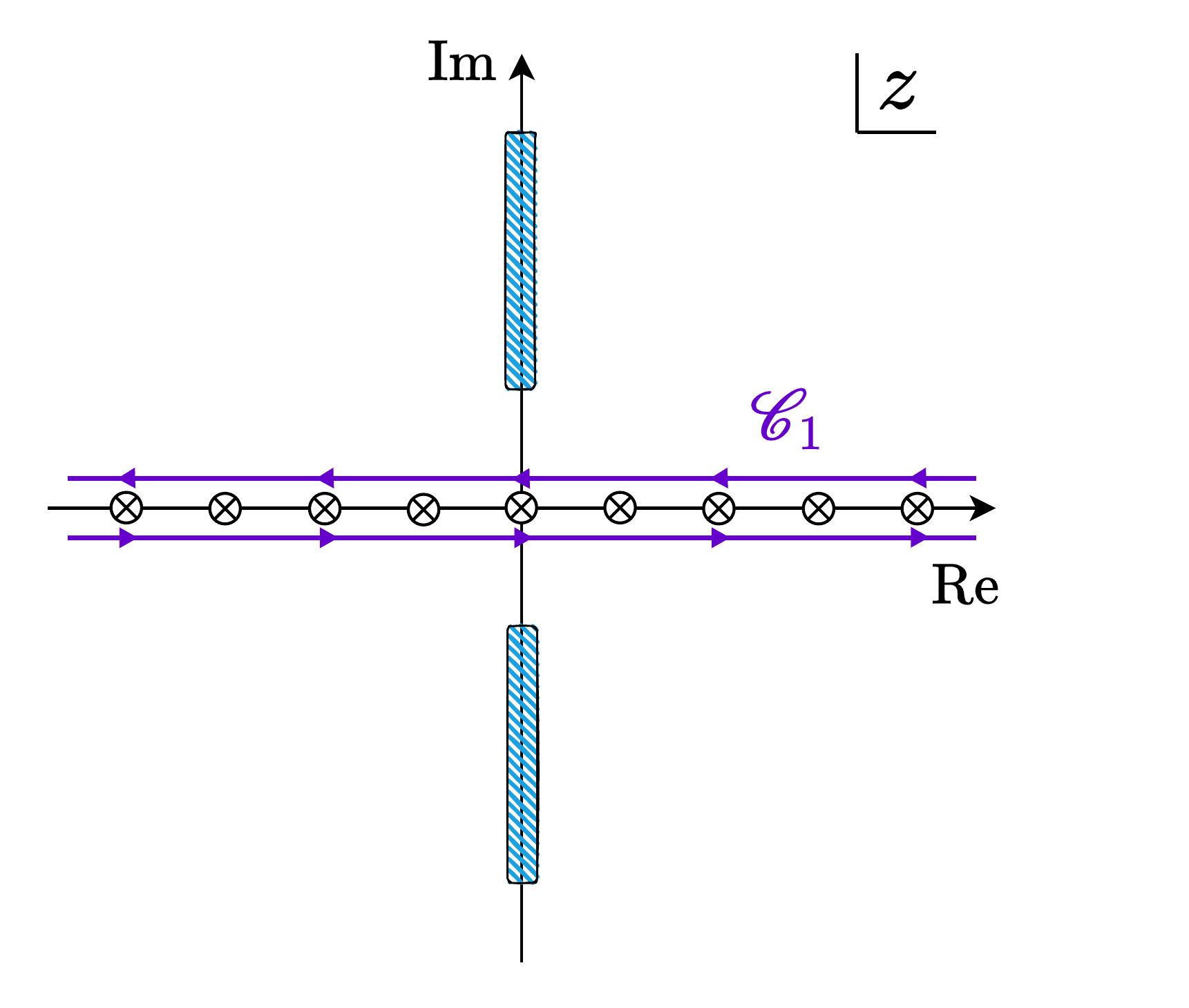}%
        }
        \vfill
        \hspace*{-2.5em}\subfloat[\label{subfig:contour3}]{%
            \hspace*{2.3em}\includegraphics[width=0.5\textwidth]{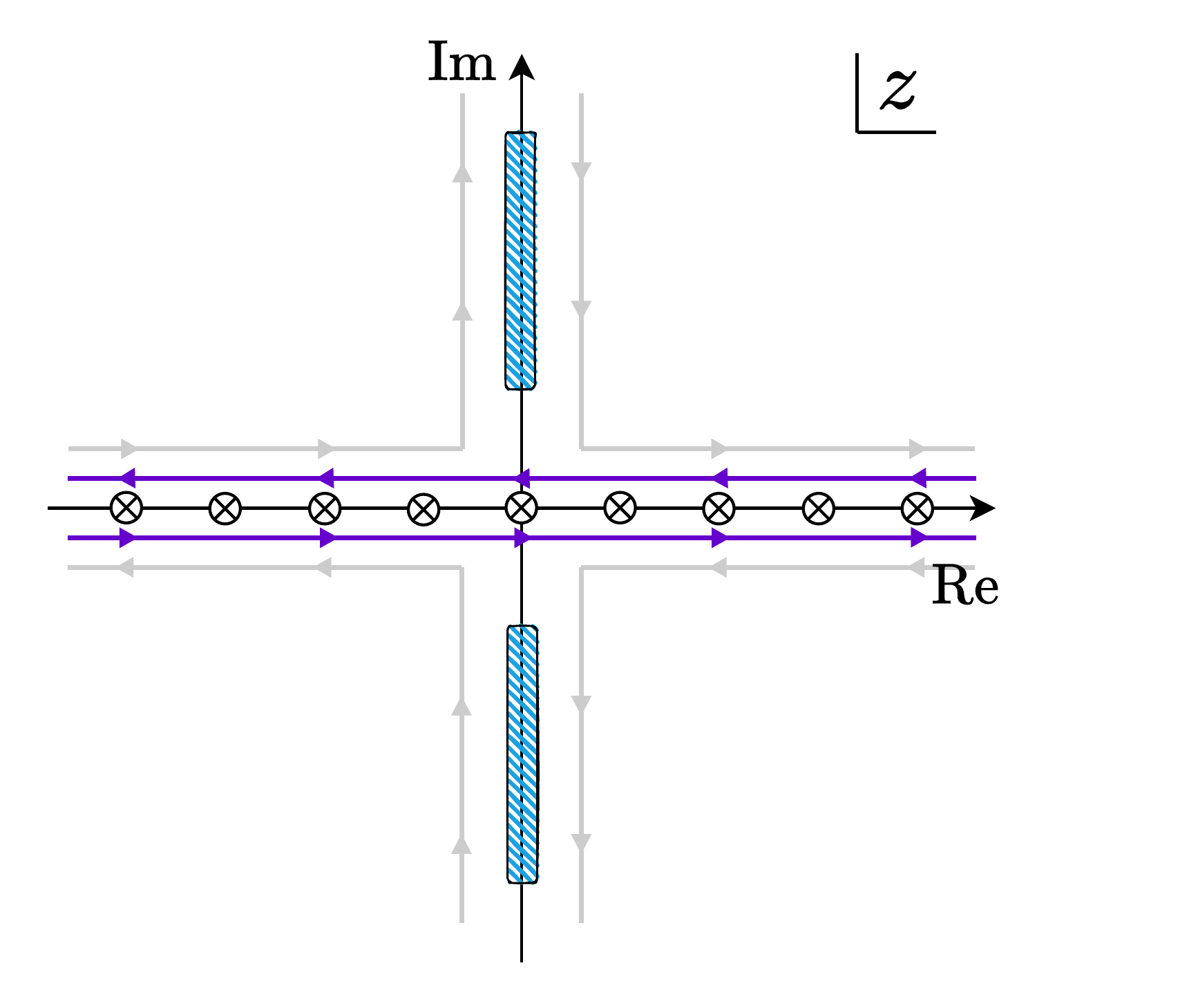}%
        }
        \hfill
        \subfloat[\label{subfig:contour4}]{%
            \hspace*{2.3em}\includegraphics[width=0.5\textwidth]{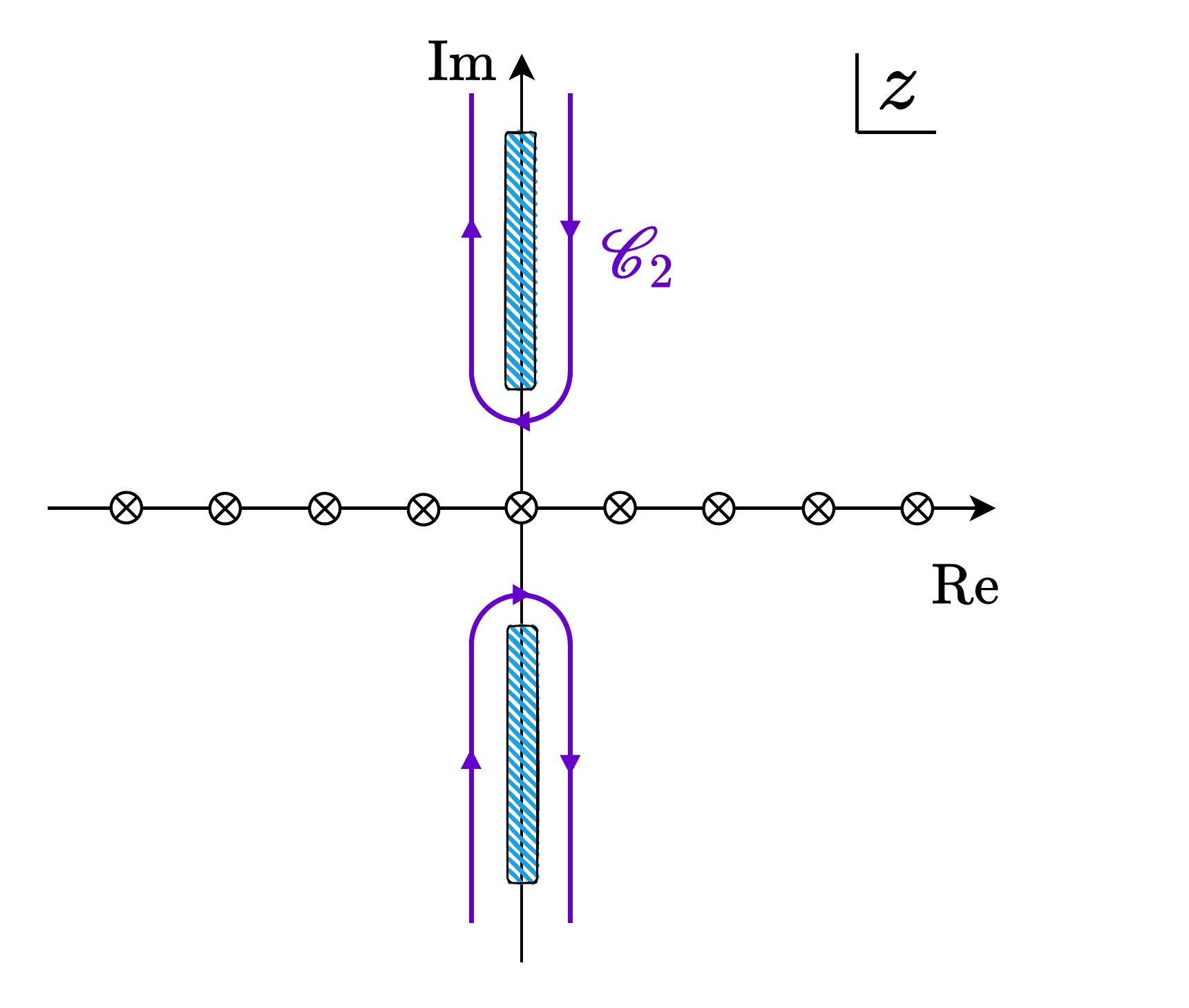}%
        }
        \caption{Contours in the complex plane used to calculate the series $s_2$. The crosses on the real axis indicate the poles of the weight function $\varpi$ at integer points. The blue dashed lines indicate the branch cuts of the function $f_2$ on the imaginary axis. \textbf{(a)} Using the residue theorem, $s_2$ can be evaluated by integrating the point-wise product of $f_2$ and $\varpi$ along circles enclosing the poles of $\varpi$ on the real axis. \textbf{(b)} The contour in (a) can be deformed into the contour $\mathscr{C}_1$, running just above and just below the real axis. \textbf{(c)} Ancilla contours, depicted in light grey and thought of as being closed at infinity using quarter circles, enclose regions of analyticity of the point-wise product of $f_2$ and $\varpi$. \textbf{(d)} The contour $\mathscr{C}_2$, which runs along the imaginary axis, enclosing the branch cuts of the function $f_2$.}
        \label{fig:contours}
    \end{figure}

    The treatment of the sum $s_2$ is more involved due to the presence of branch cuts. We define the function $f_2 (z) \coloneqq \log \left( ((z \pi / L)^2 + m_1^2) / ((z \pi / L)^2 + m_2^2) \right)$ and assume, without loss of generality, that $m_1 > m_2$. This function has branch cuts from $- \I m_1 L / \pi$ to $- \I m_2 L / \pi$ and from $\I m_2 L / \pi$ to $\I m_1 L / \pi$. Using the residue theorem, we can write the sum $s_2$ as
    \begin{equation}
        s_2 = \frac{1}{4} \sum_{n=-\infty}^{+ \infty} f_2 (n) = \frac{1}{4} \oint_{\mathscr{C}_1} f_2 (z) \varpi (z) \, \frac{\D z}{2 \pi \I} \; ,
    \end{equation}
    where $\mathscr{C}_1$ is the contour shown in Fig. \ref{subfig:contour2} and $\varpi$ is again a suitable weight function.
    
    Following \cite[Sec.~2.6]{Bellac2000}, we define, for any function $g$, its \emph{discontinuity} $\mathrm{Disc} [g]$, as the limit (if it exists, in the distributional sense)
    \begin{equation}
        \mathrm{Disc} [g] (z) \coloneqq \lim_{\varepsilon \to 0^+} \left[ g (z +  \I\varepsilon) - g (z - \I \varepsilon) \right] \; .
    \end{equation}
    We can now deform the contour $\mathscr{C}_1$ into the contour $\mathscr{C}_2$ shown in Fig. \ref{subfig:contour4}. Then, using the analyticity of $\varpi$ away from the real axis, we find
    \begin{equation}\label{eq:s2_contour_integral}
        s_2 = \frac{1}{4} \int_{- \I \infty}^{+ \I \infty} \mathrm{Disc} [f_2] (z) \, \varpi (z) \, \frac{\D z}{2 \pi \I} = \frac{1}{4} \int_{- \infty}^{+ \infty} \mathrm{Disc} [f_2] (\I x) \, \varpi (\I x) \, \frac{\D x}{2 \pi} \; ,
    \end{equation}
    where the integral over the real axis is understood as an integral over $\R \setminus (- \eta, + \eta)$ and subsequently taking the limit $\eta \to 0^+$. Upon noticing that the discontinuity of the logarithm on the real axis is given by $\mathrm{Disc} [\log] (x) = 2 \I \pi \Theta (- x)$, where $\Theta$ is the Heaviside step function, we find
    \begin{equation}
        \begin{split}
            &\mathrm{Disc} [f_2] (\I x) = \\
            &2 \I \pi \left[ \Theta \left( - x - \frac{L m_1}{\pi} \right) - \Theta \left( - x - \frac{L m_2}{\pi} \right) + \Theta \left( - x + \frac{L m_1}{\pi} \right) - \Theta \left( - x + \frac{L m_2}{\pi} \right) \right] \; .
        \end{split}
    \end{equation}

    Using again the weight function $\varpi (z) = \pi \cot (\pi z)$, we can write \eqref{eq:s2_contour_integral} as
    \begin{equation}
        s_2 = - \frac{\pi}{4} \int_{- L m_1/\pi}^{- L m_2/\pi} \coth (\pi x) \; \D x + \frac{\pi}{4} \int_{L m_2/\pi}^{L m_1/\pi} \coth (\pi x) \; \D x \; .
    \end{equation}
    Noticing that
    \begin{equation}
        \pi \int_a^b \coth (\pi x) \; \D x = \pi \int_a^b \frac{\cosh (\pi x)}{\sinh (\pi x)} \; \D x = \int_{\sinh (\pi a)}^{\sinh (\pi b)} \frac{\D \xi}{\xi} = \log \left( \frac{\sinh (\pi b)}{\sinh (\pi a)} \right) \; ,
    \end{equation}
    we arrive at the result
    \begin{equation}
        s_2 = \frac{1}{2} \log \left( \frac{\sinh (L m_1)}{\sinh (L m_2)} \right) \; .
    \end{equation}
    Upon collecting terms, we finally arrive at the expression in \eqref{eq:DKL_Dirichlet_1d} for the relative entropy for two Dirichlet fields on an interval of length $L$ with different masses.

\end{appendices}


\bibliography{references}

\end{document}